\documentclass{article}
\usepackage[preprint]{tmlr}
\usepackage{graphicx} 
\usepackage[letterpaper, margin=1in]{geometry}
\usepackage{amsthm}
\usepackage{amsmath}
\usepackage{float}
\usepackage{mathtools}
\mathtoolsset{showonlyrefs}
\usepackage{amssymb}
\usepackage{graphicx}
\usepackage{lmodern}
\usepackage{subcaption}
\usepackage{xcolor}
\usepackage{mathrsfs}
\usepackage{comment}
\usepackage{lastpage} 
\usepackage{extramarks} 
\usepackage[bb=dsserif]{mathalpha} 
\usepackage{epstopdf} 
\usepackage{amsmath}
\usepackage{algorithm}
\usepackage{algcompatible}
\usepackage{bbm}
\usepackage{enumerate}
\usepackage[utf8]{inputenc}

\usepackage{tikz}
\usepackage{pgfplots}
\usepackage{url}
\usepackage{booktabs} 
\usepackage{multirow}
\usepgfplotslibrary{groupplots,dateplot}
\usetikzlibrary{patterns,shapes.arrows, fadings}
\usetikzlibrary{automata} 
\usetikzlibrary{shapes.geometric}
\usetikzlibrary{calc,shadows}
\usetikzlibrary{positioning} 
\usetikzlibrary{arrows, arrows.meta, decorations.pathreplacing, backgrounds} 
\tikzset{node distance=2.5cm, 
every state/.style={ 
semithick,
fill=gray!10},
initial text={}, 
double distance=2pt, 
every edge/.style={ 
draw,
->,>=stealth', 
auto,
semithick}}
\pgfplotsset{compat=newest}
\usepgfplotslibrary{fillbetween}

\newlength\figH
\newlength\figW
\theoremstyle{definition}
\newtheorem{definition}{\protect\definitionname}
\newtheorem{assumption}{\protect\assumptionname}
\newtheorem{lemma}{\protect\lemmaname}
\newtheorem{proposition}{\protect\propositionname}
\newtheorem{mechanism}{\protect\mechanismname}

\newtheorem{remark}{Remark}

\theoremstyle{plain}
\newtheorem{theorem}{\protect\theoremname}
\newtheorem{problem}{Problem}

\DeclareMathOperator*{\argmax}{arg\,max}
\DeclareMathOperator*{\argmin}{arg\,min}
\DeclareMathOperator{\adj}{Adj}
\DeclareMathOperator{\erf}{erf}

\DeclareMathOperator{\Ber}{Ber}

\newcommand{\pushright}[1]{\ifmeasuring@#1\else\omit\hfill$\displaystyle#1$\fi\ignorespaces}

\makeatother

\providecommand{\assumptionname}{Assumption}
\providecommand{\definitionname}{Definition}
\providecommand{\lemmaname}{Lemma}
\providecommand{\theoremname}{Theorem}
\providecommand{\mechanismname}{Mechanism}
\providecommand{\propositionname}{Proposition}
\newcommand{\norm}[1]{\left\lVert#1\right\rVert}

\newcommand{\E}[1]{\mathbb{E}\left[#1\right]}

\newcommand{\Adjacent}[2]{\adj_{n,b, d}(#1, #2)=1}
\newcommand{\Adjacentr}[2]{\adj_{n,\rho, G_{sym}}(#1, #2)=1}
\newcommand{\Adjacentbo}[2]{\adj_{n,b, d_H}(#1, #2)=1}

\newcommand{\ones}{\mathbb{1}}

\newcommand{\prob}[1]{\mathbb{P}\left(#1 \right)}

\newcommand{\mechbo}{\mathcal{M}_2}

\newcommand{\randperm}{\mathfrak{R}}

\newboolean{show_comments}
\setboolean{show_comments}{true} 
\ifthenelse{\boolean{show_comments}}{
    \newcommand{\Alex}[1]{\textcolor{blue}{Alex:~#1}} 
    \newcommand{\mh}[1]{\textcolor{red}{MH:~#1}} 
    \newcommand{\brandon}[1]{\textcolor{orange}{bf:~#1}} 
} {
    \newcommand{\Alex}[1]{}
    \newcommand{\mh}[1]{}
    \newcommand{\brandon}[1]{}
}

\title{Differential Privacy for Markov Chain State Trajectories}
\author{\name Alexander Benvenuti \email abenvenuti3@gatech.edu \\
      \addr School of Electrical and Computer Engineering\\
      Georgia Institute of Technology
      \AND
      \name Matthew Hale \email mhale30@gatech.edu \\
      \addr School of Electrical and Computer Engineering\\
      Georgia Institute of Technology
      }
      \def\month{MM}  
\def\year{YYYY} 
\def\openreview{\url{https://openreview.net/forum?id=XXXX}} 

\begin{document}

\maketitle

\begin{abstract}
Data-driven systems may require state trajectories of Markov chains to function because these trajectories contain information that is useful to the system, e.g., a product's credit risk, a user's physical location, or a user's internet browsing behavior. However, sharing such state trajectories can reveal sensitive information about users, which presents a privacy threat. Therefore, we develop a new framework for privatizing the state trajectories in a Markov chain using differential privacy. Our framework privatizes state trajectories \emph{online}, in the sense that a private state trajectory is generated at the same time as the sensitive one it approximates. 
We treat Markov chains as weighted directed graphs whose edge weights are the negative logarithms of the transition probabilities. 
Then, each state in a private state trajectory is chosen by minimizing its distance 
to the corresponding state in the sensitive state trajectory, where the notion of distance is equal to the total edge weight along a shortest path. 
We prove that with high probability the private state trajectory remains close to the sensitive one, which maintains high utility for downstream uses of private data. Additionally, we prove that private state trajectories are consistently in the \emph{typical set} of state trajectories generated by the underlying Markov chain, which means that private state trajectories have similar statistical properties to actual state trajectories produced by the underlying Markov chain. Numerical simulations show that under~$3$-differential privacy, 
the mechanism we introduce exhibits
up to an~$80\%$ decrease in entropy compared to the state of the art, which illustrates that private state trajectories generated by our framework more closely resemble their corresponding sensitive state trajectory while maintaining the same level of privacy.
\end{abstract}

\section{Introduction}
Markov chains are used to model a variety of behaviors including credit migration of financial products~\citep{kalkbrener2026markov, hurd2007affine}, city traffic patterns~\citep{gong2011iterative}, and internet traffic patterns~\citep{dong2022personalized}. State trajectories generated by such Markov chains may be used to infer the financial health over time of a product~\citep{amed2026fsl}, individuals' daily commutes~\citep{gong2011iterative}, or their internet browsing behavior~\citep{rendle2010factorizing}, all of which 
are sensitive. Even when anonymizing data, users' identities may be easily reconstructed,
even from short state trajectories~\citep{de2013unique}. Without privacy protections, inferences from these trajectories can threaten users' privacy. 

Therefore, in this paper we develop a framework to privatize state trajectories generated by Markov chains. This framework generates a private state trajectory \emph{online}, i.e., a private state trajectory is generated simultaneously with the sensitive state trajectory it approximates. Our framework does so in a way that respects the Markov chain dynamics, in the sense that all state-to-state transitions in a private state trajectory are feasible under the transition probabilities of the given Markov chain. 

We develop this framework using differential privacy. Differential privacy is a statistical notion of privacy 
that was originally 
designed to preserve the privacy of sensitive databases when queried~\citep{dwork2006calibrating}. Differential privacy has been widely used because of its strong guarantees, namely 
that it is (i) immune to post-processing, in the sense that post-hoc computations on differentially private data remain differentially private and (ii) robust to side information, in that learning additional information about the underlying sensitive data does not substantially weaken differential privacy. These properties have led to the adoption of differential privacy in optimization~\citep{hsu2014differential, wang2016differentially, munoz2021private,benvenuti2026differentially}, planning~\citep{gohari2020privacy, chen2023differential}, and machine learning~\citep{ponomareva2023dp, jayaraman2019evaluating, blanco2022critical}, among other disciplines.

Prior work on privatizing Markov chain state trajectories treats Markov chains as non-deterministic finite-state automata~\citep{chen2023differentialsymbolic}. The approach used in that work is agnostic to the transition probabilities in the Markov chain, and it often selects private states by performing a uniform random walk over the Markov chain's state space. 
This approach can be undesirable because such random walks often do not yield ``typical'' trajectories generated by the Markov chain, which harms utility. 

In contrast, we treat Markov chains as weighted, directed graphs, where the edge weights are the negative logarithms of the transition probabilities. Private trajectories are genereted by randomly selecting one state at a time using an approach similar to the permute-and-flip mechanism for differential privacy~\citep{mckenna2020permute}. We calibrate the probability of selecting a candidate next private state by using the shortest-path distance between (i) each candidate next private state and (ii) the 
corresponding sensitive, non-private state 
in the trajectory being privatized. We show that with high probability, the resulting private state trajectories remain close (in a shortest-path distance sense) to the sensitive state trajectories they approximate. Additionally, we show that with high probability a privatized state trajectory is in the ``typical set'' of state trajectories generated by the Markov chain. The typical set is the set of state trajectories most likely to be generated by the Markov chain, and the generation of private state trajectories in this set ensures that private state trajectories preserve their utility for downstream analysis.

To summarize, we make the following contributions:
\begin{itemize}
    \item We introduce a new notion of privacy for Markov chain state trajectories and develop a differential privacy mechanism with respect to this new definition (Definition~\ref{def:traj_adjacency}, Mechanism~\ref{mech:PF}, Theorem~\ref{thm:dp_bo_adj}).
    \item We give a concentration bound on the error induced by Mechanism~\ref{mech:PF} (Theorem~\ref{cor:hajek_bound}).
    \item We bound the expected empirical entropy of the private state trajectories generated by Mechanism~\ref{mech:PF} and give concentration bounds on the empirical entropy (Theorems~\ref{thm:avg_entropy} and~\ref{thm:cont_entropy}).
    \item We empirically validate Mechanism~\ref{mech:PF} on three real datasets: credit migration, city traffic, and internet traffic, and we show it produces private trajectories with~$80\%$ less entropy than the prior state-of-the-art mechanism.
    We also empirically show that Mechanism~\ref{mech:PF} reduces the probability of large errors at any time step by up to~$4$ orders of magnitude (Section~\ref{sec:sims}).
\end{itemize}

\subsection{Related work}
Our mechanism uses the shortest-path distances on a graph as a measure of utility in private state trajectory approximations. This utility is a notion that has frequently appeared in the geo-indistinguishability literature~\citep{zhu2024ldgi,min2025road, wang2025semantic}. We differ from these works in three major ways. First, those works are concerned with preserving the privacy of individual states in a trajectory on a graph, while we consider privacy at the trajectory level. Second, the state trajectories we consider are generated via Markov chains, while those works consider deterministic walks on graphs. Third, these works select the next state in a private trajectory using the exponential mechanism, while we propose a new approach based on the permute-and-flip mechanism~\citep{mckenna2020permute}.

The authors in~\cite{chen2023differentialsymbolic, benvenuti2026differential,bertazzi2025differential, chen2026context} have previously considered privacy for the trajectories generated by Markov chains. 
We differ from~\cite{bertazzi2025differential} in two ways. First, they consider the execution of a Markov chain as a randomization algorithm itself, whereas we design a framework to privatize the state trajectories generated by a given Markov chain. Second, we consider the state trajectories themselves to be the sensitive data, not the underlying data used to generate the Markov chain. We differ from the work in~\cite{benvenuti2026differential} because we generate private state trajectories in an online fashion, which privatizes state trajectories as new transitions occur in the Markov chain, while that work privatizes state trajectories in an offline manner, i.e., the whole trajectory is privatized at once after runtime. Similarly,~\cite{guner2023learning} privatize state trajectories used to empirically determine transition probabilities in a Markov chain, while we privatize state trajectories generated by an existing Markov chain. We differ from~\cite{chen2026context} in two ways. First, their mechanism requires a context-specific distance function for a given problem, while we use the Markov chain transition probabilities to form our distance function, which may be used in any Markov chain setting. Second, they require approximating their mechanism by a deep neural network to implement their framework, while we may directly implement ours without a neural network. The authors in~\cite{fallin2023differential, benvenuti2025markov} also consider privacy for Markov chains. However, they privatize transition probabilities, or the data used to compute them, rather than state trajectories generated by the Markov chain. 

Our work is most similar to~\cite{chen2023differentialsymbolic}, which developed a mechanism for privatizing general symbolic system state trajectories online and applied it to Markov chains. We differ by designing a mechanism specifically for Markov chains. Our approach uses the transition probabilities to generate private trajectories that have more structural similarity to the underlying sensitive trajectories than the mechanism in~\cite{chen2023differentialsymbolic}. We illustrate this difference empirically in Section~\ref{sec:sims}.

\subsection{Notation}
We use~$\mathbb{N}$ to denote the non-negative integers,
$\mathbb{R}$ to denote the real numbers, and~$\mathbb{R}_+$ to denote the non-negative real numbers. 
For~$n\in\mathbb{N}$ we use~$[n] = \{0, 1,\ldots, n\}$, and we use~$\ones[X]$ as the indicator function for the event~$X$. We use~$\Theta(A)$ to denote the set of probability distributions over a non-empty, finite set~$A$. We use~$\Ber(p)$ to denote a Bernoulli distribution with parameter~$p$, and~$\mathbb{S}^{n}$ to denote the set of all~$n\times n$ stochastic matrices, i.e., the set of all~$n\times n$ matrices with non-negative entries whose rows sum to~$1$. We use~$\erf(x)$ to denote the error function, namely~$\erf(x) = \int_0^x e^{-u^2}du$. We use~$\log$ to denote the natural logarithm, and we adopt the convention that~$-\log(0) = \infty$. For a function~$f$ the output of the operator~$\argmin_x f(x)$ is the set of all points~$x$ that minimize~$f$. 
The~$\argmax$ operator is set-valued in the analogous way. 

\section{Background and Problem Statements}
This section reviews Markov chains and differential privacy, then presents problem statements.
\subsection{Markov Chains}
A stochastic process~$(Y_t)_{t\in\mathbb{N}}$ on a state space~$\mathcal{S}$ is a Markov chain if it satisfies the Markov property, namely
\begin{equation}
    \prob{Y_{t+1} = y_{t+1} \mid Y_0 = y_0, Y_1 = y_1, \ldots, Y_t = y_t}  = \prob{Y_{t+1} = y_{t+1} \mid Y_t = y_t},
\end{equation}
where~$y_t\in\mathcal{S}$ for all~$t\in\mathbb{N}$.
Throughout this work, we denote Markov chains by tuples of the form~$M = (\mathcal{S}, P, p_0)$, where~$\mathcal{S}$ is the non-empty, finite set of states of the Markov chain,~$P\in\mathbb{S}^{|\mathcal{S}|}$ is the transition matrix, and~$p_0\in\Theta(\mathcal{S})$ is the initial state distribution.
We assume that the tuple~$M = (\mathcal{S}, P, p_0)$ is publicly available.
For example, in a city traffic scenario, this assumption implies that the roads, the probability of transitioning from one road to another, and the probability of a user beginning on some road are all publicly available information. Given that city maps are freely available and aggregate traffic data is frequently shared by cities, the assumption that $M = (\mathcal{S}, P, p_0)$ is publicly available is easily satisfied in that setting. 

\begin{remark}
    Throughout this work, we denote a state in a sequence indexed in time using subscripts, e.g., with~$s_t$. To refer to individual states in the state space~$\mathcal{S}$ we use~$i,j\in\mathcal{S}$. 
\end{remark}

The probability of transitioning from state~$i\in\mathcal{S}$ to state~$j\in\mathcal{S}$ is~$P_{i j} = \prob{j\mid i}$. For~$n\in\mathbb{N}$, we denote the set of all state trajectories of length~$n+1$ by~$\mathcal{S}^{n+1}$. Throughout this work, we use~$N(s)$ to denote the neighborhood set of a state~$s$, i.e.,~$N(s) = \{y\in\mathcal{S} : \prob{y \mid s}>0\}$. 
We consider irreducible, time-homogeneous Markov chains.
A Markov chain is irreducible if it is possible to reach state~$j$ from state~$i$ in finitely many steps for all~$i, j\in\mathcal{S}$, and a Markov chain is time-homogeneous if~$P$ is not a function of time.

\begin{assumption}\label{ass:MC}
    The Markov chain~$M = (\mathcal{S}, P, p_0)$ is ergodic.
\end{assumption}

While the mechanism developed in this work only requires Markov chains to be irreducible, we require Markov chains that are ergodic for our analysis of that mechanism, 
and ergodic Markov chains are defined as follows.
Let~$p(t)\in\Theta(\mathcal{S})$ be the probability mass function (PMF) of the state distribution of the Markov chain
at time~$t$, i.e.,
for all~$i \in [|\mathcal{S}|] \setminus \{0\}$ the value of~$p_{i}(t)$ is the probability of the Markov chain being in state~$i$ at time~$t$. 
In an ergodic Markov chain~$M = (\mathcal{S}, P, p_0)$, the PMF~$p(t)$ converges to a stationary distribution~$\mu = \lim_{t\to\infty} p(t)$, where~$\mu$ satisfies~$\mu^T = \mu^T P$.

\subsection{Differential Privacy}
 The goal of differential privacy is to make ``similar'' trajectories appear approximately indistinguishable. The notion of ``similar'' is defined by an adjacency relation.
\begin{definition}[Adjacency] \label{def:word_adj}
    Fix a Markov chain~$M = (\mathcal{S}, P, p_0)$ satisfying Assumption~\ref{ass:MC}. Let~$d:\mathcal{S}\times\mathcal{S}\to \mathbb{R}$ be a metric on the state space of~$M$. Fix a number of state-to-state transitions~$n\in\mathbb{N}$ and an adjacency parameter~$b>0$. Two trajectories~$w, v\in\mathcal{S}^{n+1}$ are said to be~$(d, b)$-adjacent if~$d(w, v)\leq b$. 
    We use~$\Adjacent{w}{v}$ to denote that~$w,v$ are $(d, b)$-adjacent, and we call~$(d, b)$ an ``adjacency pair''. \hfill$\blacktriangle$
\end{definition}
We leave the choice of~$d$ in Definition~\ref{def:word_adj} open for now, and
we sometimes call~$w$ and~$v$ ``adjacent'' if~$d$ and~$b$ are clear from context. 
In Appendix~\ref{ap:bo_mech} we state the adjacency relation used in the prior state of the art in~\cite{chen2023differentialsymbolic}, and in Section~\ref{sec:mech_design} we introduce a new adjacency relation for state trajectories generated by Markov chains. 
To show why a new adjacency relation is needed, 
in Appendix~\ref{ap:alt} we develop a mechanism similar to the one developed in Section~\ref{sec:mech_design} but with
the adjacency relation from~\cite{chen2023differentialsymbolic}, and
Appendix~\ref{ap:alt}
shows 
that it leads to poor accuracy, which motivates the use of
a new adjacency relation.

Differential privacy is enforced by a randomized mapping called a ``mechanism'', which we denote by~$\mathcal{M}$. 

\begin{definition}[Differential Privacy; \cite{dwork2014algorithmic}]\label{def:word_dp}
    Fix a probability space~$(\Omega, \mathcal{F}, \mathbb{P})$, an adjacency pair~$(d, b)$, a number of state-to-state transitions~$n\in\mathbb{N}$, and a privacy parameter~$\epsilon>0$. A mechanism~$\mathcal{M}: \mathcal{S}^{n+1}\times \Omega\to \mathcal{S}^{n+1}$ 
    is~$\epsilon$-differentially private if, for all~$w, v$ that are~$(d, b)$-adjacent in the sense of Definition~\ref{def:word_adj} and all~$L\subseteq \mathcal{S}^{n+1}$,~$\mathcal{M}$ satisfies
       $\prob{\mathcal{M}(w)\in L}\leq e^{\epsilon}\prob{\mathcal{M}(v)\in L}.$ \hfill$\blacktriangle$
\end{definition}
Differential privacy for state trajectories is not explicitly defined in~\citep{dwork2014algorithmic}. However, Definition~\ref{def:word_dp} follows immediately from the definition 
of differential privacy 
stated in~\citep{dwork2014algorithmic}. If a privacy mechanism~$\mathcal{M}$ satisfies Definition~\ref{def:word_dp}, then all the guarantees of differential privacy hold for~$\mathcal{M}$, including immunity to post-processing and robustness to side information. See~\cite{dwork2014algorithmic} for more details on the guarantees 
provided by differential privacy. The parameter~$\epsilon$ quantifies the strength of privacy, and a smaller value of~$\epsilon$ implies stronger privacy. Typical values of~$\epsilon$ are~$\epsilon\in[0.1, 10]$~\citep{hsu2014differential}.  

\subsection{Problem Statements}
In this work, we solve the following problems:
\begin{problem}\label{prob:1}
    Fix a Markov chain~$M = (\mathcal{S}, P, p_0)$. Develop a mechanism that generates~$\epsilon$-differentially private approximations for the state trajectories generated by~$M$.
\end{problem}

For trajectories~$w = s_0s_1\ldots s_n$, 
the solution to Problem~\ref{prob:1} formulates a privacy mechanism that approximates each state~$s_t$ with a randomized version~$s_t'$ such that the entire state trajectory~$w$ is kept private. The challenges in doing so are two-fold: (i) because the mechanism operates \emph{online} it only has access to the history of the trajectory up to time~$t$, and (ii) it must produce a trajectory that is feasible with respect to the transition probabilities of the underlying Markov chain. Problem~\ref{prob:1} is solved in Section~\ref{subsec:sol_prob1}. 

\begin{problem}\label{prob:2}
    Bound the probability that private state trajectories contain large deviations from the sensitive state trajectories they approximate.
\end{problem}

Problem~\ref{prob:2} analytically quantifies the accuracy of the mechanism developed to solve Problem~\ref{prob:1}, and we solve Problem~\ref{prob:2} in Section~\ref{subsec:sol_prob2}. 

\begin{problem}\label{prob:3}
    Bound the probability that the proposed mechanism produces a private trajectory that is in the~$\eta$-typical set of the Markov chain that generated the underlying sensitive trajectory.
\end{problem}

Problem~\ref{prob:3} quantifies how realistic a private state trajectory is for the underlying Markov chain. State trajectories that do not appear in the typical set may easily be identified as 
trajectories that were not produced by the Markov chain
because they may contain many unlikely transitions. Thus, even if a private state trajectory is 
not an accurate 
approximation of the sensitive one, 
if the private trajectory is in the typical set for a Markov chain, then 
it will still likely appear as a ``typical'' trajectory, which preserves its usefulness for downstream analysis. Problem~\ref{prob:3} is solved in Section~\ref{sec:typical}.

\begin{problem}\label{prob:4}
    Empirically compare the performance of the proposed mechanism to the state of the art using real datasets.
\end{problem}

Section~\ref{sec:sims} presents the solution to Problem~\ref{prob:4} using three examples: credit migration, city traffic, and internet traffic, while comparing to the most similar work to ours, namely~\cite{chen2023differentialsymbolic}.
\section{Structure-Aware Mechanism Design}\label{sec:mech_design}
In this section, we solve Problems~\ref{prob:1} and~\ref{prob:2}. We aim to use a Markov chain's dynamics so that when an incorrect transition is taken, the private state trajectory on average eventually returns to a neighborhood near the sensitive state trajectory it approximates. This property yields private state trajectories with similar structure to sensitive state trajectories, 
which preserves utility for downstream usage of private data.


\subsection{Proposed Mechanism}\label{subsec:sol_prob1}
To design a mechanism that uses a Markov chain's topology, we define the graph induced by a Markov chain.

\begin{definition}[Markov Chain Induced Graph]\label{def:MC-graph}
    A Markov chain~$M = (\mathcal{S}, P, p_0)$ induces a weighted directed graph~$\mathcal{G} = (\mathcal{S}, E, W)$, where~$E = \{(i, j) \mid P_{ij}>0\}$ is the edge set, 
    and~$W = \{W_{ij}\}_{i, j \in \mathcal{S}}$
    is a set of edge weights, where~$W_{ij} = -\log(P_{ij})$.
    \hfill$\blacktriangle$
\end{definition}

 Many choices of edge weights exist for Definition~\ref{def:MC-graph}, such as the diffusion distance~\citep{coifman2006diffusion}, hitting times~\citep{boyd2021metric}, and the negative log-probabilities of the transitions, i.e.,~$W_{ij} = -\log(P_{ij})$~\citep{rabiner2002tutorial}. The negative log-probabilities is often chosen when computing paths in a graph, which we do here. Negative log-probabilities are chosen because the probabilities themselves do not induce a metric for measuring distances on a graph, while the use of the negative log-probabilities gives a quasi-metric for measuring distances on a directed graph. See Appendix~\ref{app:supporting} for details on the properties of quasi-metrics.  Let~$\mathcal{V}(i,j)$ be the set of all trajectories of any length 
 between states~$i \in \mathcal{S}$ and~$j \in \mathcal{S}$. Then we define the shortest-path distance~$G:\mathcal{S}\times \mathcal{S}\to \mathbb{R}$ as 
\begin{equation}\label{eq:G}
    G(i, j) = \min_{v\in\mathcal{V}(i,j)} \sum_{k=1}^{|v|} W_{v_{k-1}, v_{k}},
\end{equation}
where~$|v|$ denotes the length of the word~$v$. 

\begin{proposition}\label{prop:shortest_path}
    Fix a Markov chain~$M = (\mathcal{S}, P, p_0)$ satisfying Assumption~\ref{ass:MC}, and let~$w = s_0s_1\cdots s_n$ be a length~$n+1$ state trajectory 
    generated by~$M$ with~$s_0 = u$ and~$s_n = q$. Using the edge weights~$W_{ij} = -\log(P_{ij})$ for~$i,j\in\mathcal{S}$, then the state trajectory~$w^* = s_0s^*_1\cdots s_{n-1}^* s_n$ that minimizes the total graph distance from~$u$ to~$q$, namely~$G(u, q)$, is the most likely state trajectory in the Markov chain from~$u$ to~$q$. \hfill$\lozenge$
\end{proposition}

Proposition~\ref{prop:shortest_path} says that on the graph~$\mathcal{G}$ in Definition~\ref{def:MC-graph}, the shortest path between nodes~$i$ and~$j$ in the graph sense corresponds to the most likely path between states~$i$ and~$j$ in the Markov chain sense. While Proposition~\ref{prop:shortest_path} is well-known in the literature, e.g.,~\citep{rabiner2002tutorial}, the authors are unaware of a proof of it in the literature, 
and thus we provide one in Appendix~\ref{app:proofs}, which also has proofs of all new results.

To develop a mechanism that does not exhibit any random walking behavior when generating private outputs, we weigh all possible transitions from a private state~$s_{t-1}'$ 
to a candidate private state~$s_t'$ 
by the shortest path distance between~$s_t'$ and the true state~$s_t$. The shortest path distance from~$s_t'$ to~$s_t$ is~$G(s'_t, s_t)$, where~$G$ is from~\eqref{eq:G}. All of the shortest-path distances in~$\mathcal{G}$ can be computed simultaneously via the Floyd-Warshall algorithm~\citep{floyd1962algorithm}. We then use the shortest-path distances to define adjacency for state trajectories.
\begin{definition}[Trajectory Adjacency]\label{def:traj_adjacency}
     Fix a Markov chain~$M = (\mathcal{S}, P, p_0)$ satisfying Assumption~\ref{ass:MC} and an adjacency parameter~$\rho>0$. Let~$\mathcal{G} = (\mathcal{S}, E, W)$ be defined as in Definition~\ref{def:MC-graph}. Two state trajectories~$w = s_0\cdots s_n\in \mathcal{S}^{n+1}$ and $v = \tilde{s}_0\cdots \tilde{s}_n\in\mathcal{S}^{n+1}$ are adjacent if
    \begin{equation}\label{eq:adj_relation}
        \sum_{t=0}^n G_{sym}(s_t, \tilde{s}_t) \leq \rho,
    \end{equation}
    where~$G_{sym}(x, y) = \max \{G(x, y), G(y, x)\}$, and~$G$ is from~\eqref{eq:G}.\hfill$\blacktriangle$
\end{definition}

Using the notation of Definition~\ref{def:word_adj}, we write~$\adj_{n,\rho,G_{sym}}(w, v) = 1$
if two trajectories are adjacent in the sense of Definition~\ref{def:traj_adjacency}, and we write
$\adj_{n, \rho, G_{sym}}(w, v) = 0$ otherwise. 
Definition~\ref{def:traj_adjacency} is in line with the notion of adjacency for numerical trajectories in~$\ell_p$-spaces~\citep{le2013differentially}, where~$p$ is typically taken to be~$2$. However, Definition~\ref{def:traj_adjacency} is analogous to the~$\ell_1$-norm because the state space is on a graph, and thus there is no Euclidean distance between states. Definition~\ref{def:traj_adjacency} may also be viewed from an information-theoretic perspective: two state trajectories~$w, v\in\mathcal{S}^{n+1}$ are adjacent if the maximum difference in their Shannon information content is bounded above by~$\rho$ nats. The adjacency parameter~$\rho$ may then be chosen based on how unlikely one state trajectory can be relative to another while still being considered ``similar''.
For example, using~$\rho = -\log(0.1)$ implies that differences between trajectories that occur with probability~$0.1$ are protected by differential privacy.
Unlike the adjacency relation in~\cite{chen2023differentialsymbolic}, where adjacent trajectories may only differ in a single entry, our notion of adjacency allows for trajectories to differ in multiple ways: (i) they may differ in many entries such that the disagreeing
state-to-state transitions in both trajectories occur with high probability, or (ii) they may differ in only a few entries such that the disagreeing state-to-state transitions in both trajectories occur with low probability.

To enforce differential privacy, we will create a mechanism that is similar in spirit to the permute-and-flip mechanism~\citep{mckenna2020permute}. The permute-and-flip mechanism assigns probabilities to all possible outputs based on a utility score, which encodes how well a private state trajectory~$w'$ approximates a sensitive state trajectory~$w$. Throughout this paper, for two state trajectories~$w = s_0s_1\cdots s_n\in\mathcal{S}^{n+1}$ and~$w' = s_0's_1'\cdots s_n'\in\mathcal{S}^{n+1}$, we use~$G$ from~\eqref{eq:G} to define the utility function 
\begin{equation} \label{eq:udef}
u(w, w') = -\sum_{t\in[n]}G(s_t', s_t),
\end{equation}
which 
encodes the fact that~$w'$ is a better approximation for a sensitive state trajectory~$w$
if the distance from~$w'$ to~$w$ is small. Unlike the permute-and-flip mechanism, however, we cannot select entire private state trajectories based on this utility because we are privatizing sensitive trajectories online. Instead, the Markov property allows us to select each private state in a similar fashion to the permute-and-flip mechanism at each time step while guaranteeing privacy for the entire trajectory.

\begin{remark}\label{rem:aware}
    We refer to mechanisms that measure similarity using the transition probabilities as \emph{structure-aware} and mechanisms that do not as \emph{structure agnostic}. 
\end{remark}

 The utility function~$u$ in~\eqref{eq:udef} admits a direct information-theoretic interpretation. Because the graph edge weights are defined as negative log-probabilities, the shortest-path distance~$G(s_t', s_t)$ represents the minimum Shannon information content required to specify the true state~$s_t$ given the private state~$s_t'$. 
    
To calibrate our mechanism, we require the sensitivity of~$u$ in~\eqref{eq:udef}, which encodes the maximum amount that the utility function can change on adjacent inputs.

\begin{lemma}[Sensitivity]\label{lem:new_sens}
    Fix a Markov chain~$M = (\mathcal{S}, P, p_0)$ satisfying Assumption~\ref{ass:MC} and an adjacency parameter~$\rho>0$. Let~$\mathcal{G} = (\mathcal{S}, E, W)$ be as defined in Definition~\ref{def:MC-graph}. Let~$w = s_0\cdots s_n\in \mathcal{S}^{n+1}$ and $v = \tilde{s}_0\cdots \tilde{s}_n\in\mathcal{S}^{n+1}$ be two adjacent state trajectories in the sense of Definition~\ref{def:traj_adjacency}. Let~$u$ be from~\eqref{eq:udef}. Then
    \begin{equation}
       \Delta u =  \max_{y\in\mathcal{S}^{n+1}}\max_{\substack{w, v\in\mathcal{S}^{n+1}\\\Adjacentr{w}{v}}} |u(w, y) - u(v, y)| \leq \rho. 
       \tag*{$\blacklozenge$}
    \end{equation}
\end{lemma}

We use Algorithm~\ref{algo:PF} to implement the proposed mechanism.
In line~$6$ we use $\mathfrak{R}(N(s_t'))$ to denote a uniformly randomly selected permutation of the elements of the set~$N(s_t')$, and in line~$8$
we use~$r_i$ for~$i \in \{1, 2, \ldots, |N(s_t')|\}$
to denote the~$i^{th}$ element in the resulting permutation.
At time step~$t = 0$, the first state of the private trajectory, namely~$s_0'$, is randomly drawn from the 
initial distribution~$p_0$
of the Markov chain~$M$.
For times~$t > 0$, Algorithm~\ref{algo:PF} proceeds as follows. 
We randomly permute the elements of~$N(s'_{t-1})$, which is the set of candidate next states 
that are feasible 
from the previous private state~$s'_{t-1}$. 
A biased coin is flipped to determine if a candidate next state is selected to the be 
the private output state~$s_t'$ at time~$t$. 
The coin's bias is computed using the utility of the candidate output state.
In Algorithm~\ref{algo:PF}, we define~$ \min_{y\in N(s_{t}')} G(y, s_{t+1})$ as the 
shortest distance from any state in~$N(s_t')$ to~$s_{t+1}$. 
We see that~$\min_{y\in N(s_{t}')} G(y, s_{t+1})-G(s_{t+1}', s_{t+1}) = 0$
when~$s_{t+1}' \in \argmin_{y\in N(s_t')} G(y, s_{t+1})$, 
and~$p_{s_{t+1}'} = 1$ for all such states. 
Therefore, Algorithm~\ref{algo:PF} is guaranteed to terminate for all~$t\in[n]$
because the \textbf{for} loop in lines~$8$-$15$ 
is guaranteed to terminate
for all~$t \in [n]$. 
We now formally state the proposed mechanism.

\begin{algorithm}[t]
    \caption{Private Markov Chain State Trajectory Generation}
    \label{algo:PF}
    \begin{algorithmic}[1]
    \STATEx \textbf{Inputs}: (i) Neighborhood set of~$s_{t}'$,~$N(s_{t}')$, (ii) the set of all shortest-path distances~$\{G(i, j) \mid i,j\in\mathcal{S}$\}, (iii) next sensitive state~$s_{t+1}$, (iv) adjacency parameter~$\rho$, and (v) privacy parameter~$\epsilon$ 
    \STATEx \textbf{Outputs}: Next private state~$s_{t+1}'$
    \vspace{1mm}
    \IF{$t  = 0$}
    \STATE \textbf{return} $s_0'\sim p_0$
    \ENDIF
    \STATE $R \gets \randperm(N(s_{t}'))$
    \STATE $i \gets 1$
    \FOR {$r_i\in R$}
    \STATE $s_{t+1}'\gets r_i$
    \STATE $p_{s_{t+1}'} \gets \exp\left(\frac{\epsilon}{2\rho}(\min_{y\in N(s_{t}')} G(y, s_{t+1})-G(s_{t+1}', s_{t+1})) \right)$
    \IF{$\Ber(p_{s_{t+1}'})$}
    \STATE \textbf{return} $s_{t+1}'$
    \ENDIF
    \STATE $i \gets i+1$
    \ENDFOR
    \end{algorithmic}
\end{algorithm}

\begin{mechanism}[Solution to Problem~\ref{prob:1}]\label{mech:PF}
    Fix a Markov chain~$M = (\mathcal{S}, P, p_0)$ satisfying Assumption~\ref{ass:MC}, an adjacency parameter~$\rho>0$, and a privacy parameter~$\epsilon\geq0$. 
    At each time~$t$, given the current state of the sensitive state trajectory, the mechanism~$\mathcal{M}_1$ selects the output state~$s_t'$ according to Algorithm~\ref{algo:PF}. \hfill $\triangle$
\end{mechanism}

Prior to runtime, the Floyd-Warshall algorithm~\citep{floyd1962algorithm} 
can be used to compute the shortest-path distances between all pairs of states, 
which results in a time complexity of~$\mathcal{O}(|\mathcal{S}|^3)$ for privatizing an entire trajectory using Algorithm~\ref{algo:PF}. 
However, the time complexity at time step~$t$ is~$\mathcal{O}\big(|N(s_{t-1}')|\big)$, since the shortest-path distances may all be acquired at runtime via table-lookup after using the Floyd-Warshall algorithm
once before runtime. 
The~$\mathcal{O}(|\mathcal{S}|^3)$ complexity is the same as~\cite{chen2023differentialsymbolic}, and thus any improvement in performance over their mechanism is gained without any increase in computational complexity.
Now we formally state the privacy guarantee of Mechanism~\ref{mech:PF}.

\begin{theorem}\label{thm:dp_bo_adj}
    Fix a Markov chain~$M = (\mathcal{S}, P, p_0)$ satisfying Assumption~\ref{ass:MC} and a sensitive state trajectory~$w = s_0, \ldots, s_n$. Fix an adjacency parameter~$\rho>0$ and a privacy parameter~$\epsilon\geq 0$. Mechanism~\ref{mech:PF} provides~$\epsilon$-differential privacy to~$w$ with respect to the adjacency relation in Definition~\ref{def:traj_adjacency}.
\end{theorem}


Theorem~\ref{thm:dp_bo_adj} says that 
our mechanism privatizes trajectories
by generating a random state at each time step. 

\subsection{Drift Analysis}\label{subsec:sol_prob2}
The goals of Mechanism~\ref{mech:PF} are to guarantee that (i) the private state trajectory remains close to the sensitive one in the graph 
$\mathcal{G}$
if they start together, i.e., if~$s_0 = s_0'$, and (ii) the private state trajectory will return to a neighborhood around the sensitive state trajectory if they drift apart at some point. 
Mathematically, this latter property may be stated in terms of the Foster-Lypaunov condition~\citep{meyn1992stability}, namely 
\begin{equation}\label{eq:neg_drift}
        \mathbb{E}_{s_{t+1}, s_{t+1}'}\left[G(s'_{t+1}, s_{t+1}) - G(s'_{t}, s_{t})\mid (s_t, s_t')\right] < 0.
\end{equation}
We say that the set of joint states~$(s_t, s_t')\in\mathcal{S}\times\mathcal{S}$ that satisfy~\eqref{eq:neg_drift} has ``negative drift''.
The stochastic process of joint state pairs~$(s_t, s_t')$ is a hidden Markov model (HMM), 
where the hidden state is the true state~$s_t$ and the observation is the output state~$s_t'$. 
Then, the process of joint state pairs~$(s_t, s_t')$ is the Markov chain
\begin{equation}\label{eq:joint_chain}
    M_{Z} = (\mathcal{Z}, P^{Z}_{\epsilon}, z_0),
\end{equation}
where~$\mathcal{Z} = \mathcal{S}\times \mathcal{S}$ is the joint state space,~$P^{Z}_{\epsilon}\in\mathbb{S}^{|\mathcal{S}|^2}$ is the joint transition probability matrix, and~$z_0\in\Theta(\mathcal{Z})$ is the initial distribution over all initial joint states. Thus, to analyze privacy for length~$n+1$ trajectories, we may analyze the drift properties of the joint state trajectory~$(Z_t)_{t\in[n]}$, where~$Z_t = (s_t, s_t')$.

Next, we derive conditions under which a joint state~$(s_t, s_t')\in\mathcal{Z}$ has negative drift when~$s_t'$ is selected using Mechanism~\ref{mech:PF}. To do so, we define the following quantities. Let
\begin{equation}\label{eq:V}
    V(Z_t) = G(s_t', s_t)
\end{equation}
be the shortest-path distance from~$s_t'$ to~$s_t$, which we refer to as the ``Foster-Lyapunov function'', and let
\begin{equation}\label{eq:z_cont}
    \mathcal{Z}_{attract}(\epsilon) = \{Z_t\in\mathcal{Z} : \E{V(Z_{t+1}) - V(Z_t)\mid Z_t} < 0\}
\end{equation}
be the set of states with negative drift for~$t\geq 1$, which we refer to as the ``attractive set''. The dependence of~$\mathcal{Z}_{attract}$ on~$\epsilon$ follows from the fact that the transition probabilities namely~$P^{Z}_{\epsilon}$ depend on~$\epsilon$. 

Let~$s_{t}^{\prime, *} \in \arg\min_{j\in N(s_{t}')} G(j, s_{t+1})$ be a 
next state along 
a (potentially non-unique) shortest path from~$s_t'$ to~$s_{t+1}$. Define
\begin{equation}\label{eq:bigS}
    S_{t}^{\prime,\dagger} = 
    \begin{cases}
        \arg\min_{j\in N(s_{t}')\setminus \{s_t^{\prime, *}\}} G(j, s_{t+1}) & \text{if } N(s_{t}')\setminus \{s_t^{\prime, *}\}\neq \emptyset\\
        \{s_{t}^{\prime, *}\} & \text{if } N(s_{t}')\setminus \{s_t^{\prime, *}\} =\emptyset
    \end{cases}.
\end{equation}
This set equals the set of next states along a (potentially non-unique) second-shortest path
from~$s_t'$ to~$s_{t+1}$ 
if one exists, and it equals the singleton~$\{s_{t}^{\prime, *}\}$ otherwise. 
Let~$s_{t}^{\prime,\dagger} \in  S_{t}^{\prime,\dagger}$.
We also define 
\begin{align}
    W_{s_t', y_{good}} &= \min_{y\in N(s_t')} G(y, s_{t+1}) \label{eq:w}\\
    D_{s_t' ,max} &= \max_{v\in N(s_t')} G(v, s_{t+1})  \label{eq:D}\\
    \delta_{s_t} &=\sum_{j\in N(s_t)} -P_{s_tj}\log(P_{s_tj})\label{eq:delta_s}\\
    \gamma_{s_{t+1},s_t'} &= G(s_t^{\prime, \dagger}, s_{t+1}) - G(s_t^{\prime, *}, s_{t+1})\label{eq:gamma}\\
    \pi(s_{t}^{\prime, *}\mid s_{t+1},s_t') &= \frac{1}{1+(|N(s_t')|-1)\exp\left(-\frac{\epsilon}{2\rho}\gamma_{s_t'}\right)}.\label{eq:pi1}
\end{align}

In words,~$W_{s_t', y_{good}}$ is the edge weight along the shortest-distance path between the states~$s_t'$ and~$s_{t+1}$, the constant~$D_{s_t ,max}$ is the  
longest shortest-path distance to~$s_{t+1}$ from any of the candidate output states in~$N(s_t')$, and~$\delta_{s_t}$ is an upper bound on the expected change in~$G$ 
when~$s_{t}$ transitions to~$s_{t+1}$. Formally, $\delta_{s_t}$ is equal to~$\mathbb{E}_{s_{t+1}}\left[G(s'_t, s_{t+1}) -G(s_t', s_t) \mid (s_t, s_t')\right]$. Additionally,~$\gamma_{s_{t+1},s_t'}$ is the difference in edge weight between the highest utility transition and the second-highest utility transition in the private trajectory being constructed.
Finally, consider partitioning~$N(s_t')$ into: (i) the set of ``best'' next states, in the sense that they all attain the shortest-path distance to the true state, equal to~$ G(s_t^{\prime, *}, s_{t+1})$, 
and (ii) the set of non-best candidate next states. When all non-best candidate next states have the same shortest-path distance to the true state~$G(s_t^{\prime, \dagger}, s_{t+1})$, we have that ~$ \pi(s_{t}^{\prime, *}\mid s_{t+1},s_t')$ is the probability that the best next state is selected. 
See Appendix~\ref{ap:exp_mechanism} 
for a derivation of the distribution that~$\pi$ follows from. Figure~\ref{fig:attractive_intuition} shows a visual representation of~\eqref{eq:w},~\eqref{eq:D}, and~$G(s_t^{\prime, *}, s_{t+1})$, which is used to compute~\eqref{eq:gamma}.

\begin{figure}
    \centering
\definecolor{darkgrey}{HTML}{444444}
\definecolor{mechblue}{HTML}{5aa1d8}
\definecolor{chenred}{HTML}{e65555}
\resizebox{0.45\linewidth}{!}{
\begin{tikzpicture}[
    >=Stealth,
    state/.style={circle, draw=darkgrey, fill=white, thick, minimum size=0.8cm, text=darkgrey},
    target/.style={circle, draw=darkgrey, fill=gray!20, very thick, minimum size=0.9cm},
    shortest/.style={->, line width=1.5pt, color=darkgrey},
    second/.style={->, line width=1.5pt, dashed, color=mechblue},
    longest/.style={->, line width=1.5pt, dotted, color=chenred},
    background/.style={->, line width=0.6pt, color=gray!30}
]

    \node[target] (x) at (0, 0) {$s_t'$};
    \node[target] (y) at (8, 0) {$s_{t+1}$};

    \node[state] (v1) at (4, 2) {$v_1$};
    
    \node[state] (v2) at (2.6, 0) {$v_2$};
    \node[state] (v3) at (5.4, 0) {$v_3$};
    
    \node[state] (v4) at (2, -2) {$v_4$};
    \node[state] (v5) at (4, -2) {$v_5$};
    \node[state] (v6) at (6, -2) {$v_6$};

    \draw[background] (v1) -- (v3);
    \draw[background] (v2) -- (v5);
    \draw[background] (v5) -- (v3);
    \draw[background] (v4) -- (v2);
    \draw[background] (v6) -- (v3);
    \draw[background] (v2) -- (v1);

    \draw[shortest] (x) to[out=45, in=180] node[above, sloped, font=\small, text=darkgrey] {$W_{s_t', y_{good}}$} (v1);
    \draw[shortest] (v1) to[out=0, in=135] (y);

    \draw[second] (x) -- (v2) node[midway, above, font=\small, text=mechblue] {$G(s_t^{\dagger, *}, s_{t+1})$};
    \draw[second] (v2) -- (v3);
    \draw[second] (v3) -- (y);

    \draw[longest] (x) to[out=-45, in=180] node[below, sloped, font=\small, text=chenred] {$D_{s_t,max}$} (v4);
    \draw[longest] (v4) -- (v2);
    \draw[longest] (v2) -- (v5);
    \draw[longest] (v5) -- (v3);
    \draw[longest] (v3) to[out=0, in=-135] (y);

\end{tikzpicture}
}
    \caption{Visualization of the constants used to compute~\eqref{eq:w},~\eqref{eq:D} and~\eqref{eq:gamma}, where~$v_1, v_2$, and~$v_4$ are candidate output states at~$t+1$.}
    \label{fig:attractive_intuition}
\end{figure}
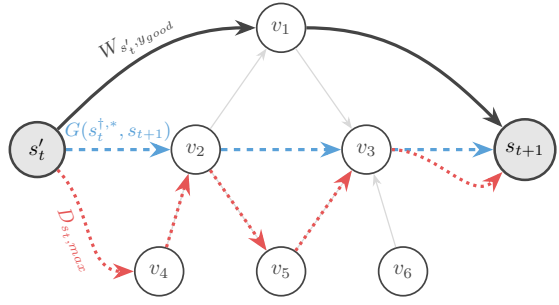

\begin{lemma}\label{thm:old_adj_safe_set}
    Fix a Markov chain~$M = (\mathcal{S}, P, p_0)$ satisfying Assumption~\ref{ass:MC}. Fix an adjacency parameter~$\rho>0$ and a privacy parameter~$\epsilon\geq 0$.  Let~$\mathcal{G} = (\mathcal{S}, E, W)$ be from Definition~\ref{def:MC-graph}. Let~$(Z_t)_{t\in[n]}$ be the joint stochastic process of the sensitive state~$s_t$ and the output state~$s_t'$ generated by the Markov
    chain~$M_Z$ in~\eqref{eq:joint_chain}. 
    Consider~$\mathcal{Z}_{attract}(\epsilon)$ from~\eqref{eq:z_cont}. 
    A sufficient condition for~$(s_t, s_t')\in\mathcal{Z}_{attract}(\epsilon)$ is
    \begin{equation}
W_{s_t', y_{good}} > \left(\frac{1 - \pi(s_{t}^{\prime, *} \mid s_{t+1},s_t')}{\pi(s_{t}^{\prime, *} \mid s_{t+1},s_t')}\right) D_{s_t,max} + \frac{\delta_{s_{t}}}{\pi(s_{t}^{\prime, *} \mid s_{t+1},s_t')} \text{ for all }s_{t+1}\in N(s_t),
\end{equation}
where~$W_{s_t', y_{good}}$ is from~\eqref{eq:w},~$D_{s_t', max}$ is from~\eqref{eq:D},~$\delta_{s_t}$ is from~\eqref{eq:delta_s}, and~$\pi(s_{t}^{\prime, *} \mid s_{t+1},s_t')$ is from~\eqref{eq:pi1}. \hfill$\blacklozenge$ 
\end{lemma}

Next, Theorem~\ref{thm:old_adj_safe_set} provides a way for users to measure how much of the state space~$\mathcal{Z}$ is in the set~$\mathcal{Z}_{attract}(\epsilon)$. 

\begin{theorem}[Solution to Problem~\ref{prob:2}]\label{cor:hajek_bound}
    Fix a Markov chain~$M = (\mathcal{S}, P, p_0)$ satisfying Assumption~\ref{ass:MC}, an adjacency parameter~$\rho>0$, and a privacy parameter~$\epsilon\geq 0$. Let~$\mathcal{G} = (\mathcal{S}, E, W)$ be as defined in Definition~\ref{def:MC-graph}. Let~$(Z_t)_{t\in[n]}$ be the joint stochastic process of the sensitive state~$s_t$ and the output state~$s_t'$ generated by~$M_Z$ from~\eqref{eq:joint_chain}. 
    Consider $V(Z_t)$ from~\eqref{eq:V} and~$\mathcal{Z}_{attract}(\epsilon)$ from~\eqref{eq:z_cont}. Let
\begin{equation}\label{eq:b_epsilon}
B_{\epsilon} = \max_{(s_t, s_t')\in \mathcal{Z}_{attract}^c(\epsilon)} G(s_t', s_t),
\end{equation}
where $\mathcal{Z}_{attract}^c(\epsilon) = \mathcal{Z}\setminus \mathcal{Z}_{attract}(\epsilon)$.
Then for all~$v>B_{\epsilon}$,
\begin{equation}\label{eq:hajek}
\prob{V(Z_t)>v} \leq C_{t} \exp\left(-\alpha (v-B_{\epsilon})\right),
\end{equation}
where
\begin{align}
C_{t} &=  \zeta^te^{-\alpha B_{\epsilon}}\sum_{x\in\mathcal{S}}\sum_{y\in\mathcal{S}} p_0(x)p_0(y) e^{\alpha G(x,y)}+ \frac{1-\zeta^t}{1-\zeta}e^{\alpha D_{\mathcal{Z}_{attract}(\epsilon)}},\label{eq:Cz}\\
D_{\mathcal{Z}_{attract}(\epsilon)} &= \max_{(s_t, s_t')\in\mathcal{Z}_{attract}(\epsilon)} \max_{w\in N(s_t'), y\in N(s_t)} G(w, y) - G(s_t', s_t),\label{eq:Dz}\\
        \gamma_{\mathcal{Z}_{attract}(\epsilon)} &= \min_{(s_t, s_t')\in \mathcal{Z}_{attract}(\epsilon)} \left| \sum_{y\in N(s_t')} \pi(y \mid s_t') \big(G(y, s_{t+1}) - G(s_t',s_t)\big)\right|,\label{eq:gammaZ}\\
\end{align}
$\alpha = \frac{2\gamma_{\mathcal{Z}_{attract}(\epsilon)}}{D_{\mathcal{Z}_{attract}(\epsilon)}},\quad
    \zeta = 1 - \frac{\gamma_{\mathcal{Z}_{attract}(\epsilon)}^2}{D_{\mathcal{Z}_{attract}}(\epsilon)}$, and we have used $G$ from~\eqref{eq:G} 
and~$\pi$ from~\eqref{eq:pi1}. 
\end{theorem}

In Theorem~\ref{cor:hajek_bound}, the quantity in~\eqref{eq:Dz} is the most that~$V(Z_t)$ can change inside of~$\mathcal{Z}_{attract}(\epsilon)$ across one timestep, 
the quantity in~\eqref{eq:gammaZ} is the smallest average change in~$V(Z_t)$ across one timestep
in~$\mathcal{Z}_{attract}(\epsilon)$, and the quantity in~\eqref{eq:b_epsilon} is the largest value of~$V(Z_t)$ outside of~$\mathcal{Z}_{attract}(\epsilon)$. Theorem~\ref{cor:hajek_bound} shows that the probability of the private state trajectory having large deviations away from the underlying sensitive state trajectory is exponentially decreasing in the size of the deviation, which ensures that the private state trajectory maintains structure-aware similarity to the sensitive one 
(in the sense of Remark~\ref{rem:aware}) with high probability. In the next section, we show how this structure-aware similarity leads to private state trajectories that would typically be produced by the Markov chain~$M$.

\section{Information Preservation and Typicality}\label{sec:typical}
In this section, we solve Problem~\ref{prob:3}. 
Specifically, we answer the question of ``how likely is it that a private state trajectory was generated by the Markov chain~$M$?'' 
To do this, we bound the expected entropy from private state trajectories. The empirical entropy of a state trajectory~$Y = (Y_t)_{t\in[n]}$ with~$n\geq 2$ is
\begin{equation}\label{eq:entropy_y}
    H_n(Y) = -\frac{1}{n-1} \log (\prob{Y}) = \frac{1}{n-1}\sum_{t=1}^n  G(Y_{t-1}, Y_t),
\end{equation}
where~$G$ is from~\eqref{eq:G}.
The average empirical entropy of a Markov chain is referred to as its ``entropy rate'', which we denote~$H(\mathcal{S})$~\citep[Chapter 4]{cover1999elements}. In a Markov chain~$M = (\mathcal{S}, P,  p_0)$,
\begin{equation}\label{eq:empirical_entropy}
    H(\mathcal{S}) = \lim_{n\to\infty} -\frac{1}{n} \log(\prob{s_0s_1\cdots s_n}) =  -\sum_{i\in \mathcal{S}}\sum_{j\in \mathcal{S}} \mu_i P_{ij}\log(P_{ij}),
\end{equation}
where~$\mu$ is the stationary distribution of the Markov chain~$M$ and 
 the second equality is a consequence of the asymptotic equipartition property (AEP)~\citep{breiman1957individual}. If a private state trajectory's empirical entropy is substantially greater than the entropy rate~$H(\mathcal{S})$, then the private state trajectory contains many state-to-state transitions
that occur with low probability in the Markov chain dynamics, which 
implies that the trajectory may have been generated using different dynamics than 
those of the Markov chain~$M$.
To state the average empirical entropy bound, let 
\begin{equation}\label{eq:asym_coeff}
    K_{asym} = \max_{\substack{u, v \in \mathcal{S} \\ u\neq v}}\frac{G(u, v)}{G(v, u)} 
\end{equation}
be the largest asymmetry ratio between shortest-path distances, i.e., the maximum ratio of (i) the shortest-path distance from state~$u$ to state~$v$ and (ii) the shortest-path distance from state~$v$ to state~$u$.

\begin{theorem}\label{thm:avg_entropy}
     Fix a Markov chain~$M = (\mathcal{S}, P, p_0)$ satisfying Assumption~\ref{ass:MC} with stationary distribution~$\mu$, an adjacency parameter~$\rho>0$, and a privacy parameter~$\epsilon\geq 0$. Let~$w = (s_t)_{t\in[n]}$ for~$n\geq 2$ be a sensitive state trajectory and let~$w' = (s_t')_{t\in[n]} = (\mathcal{M}_1(s_t))_{t\in[n]}$ 
     be its corresponding private state trajectory generated from applying Mechanism~\ref{mech:PF} to~$w$. 
     Suppose~$p_0 = \mu$ so that~$s_0\sim \mu$. 
     Then
     \begin{multline}
         \E{H_n(w')}\leq\\ \frac{n}{n-1}\left(H(\mathcal{S})+\left(\frac{1}{\psi_2}+K_{asym}\right)\left(\frac{\sqrt{\pi}}{2\sqrt{\frac{2}{D_{max}}}}\erf\left(B_{\epsilon}\sqrt{\frac{2}{D_{max}}}\right)+C_{\infty}\frac{D_{\mathcal{Z}_{attract}(\epsilon)}}{2\gamma_{\mathcal{Z}_{attract}(\epsilon)}}\right)\right) +\frac{\left(1+\psi_2 K_{asym}\right)\bar{V}_0}{(n-1)(1-\psi_2)},
     \end{multline}
     where~$\psi_2\in(0, 1)$ is the second largest eigenvalue modulus of the transition matrix for the joint Markov chain~$M_{\mathcal{Z}}$ from~\eqref{eq:joint_chain}, we use~$\bar{V}_0 = \sum_{x\in\mathcal{S}}\sum_{y\in\mathcal{S}} p_0(x)p_0(y) G(x,y)$,~$D_{max}$ is the diameter of the graph~$\mathcal{G}$ from Definition~\ref{def:MC-graph},~$B_{\epsilon}$ is from~\eqref{eq:b_epsilon},~$D_{\mathcal{Z}_{attract}(\epsilon)}$ is from~\eqref{eq:Dz},~$\gamma_{\mathcal{Z}_{attract}(\epsilon)}$ is from~\eqref{eq:gammaZ},~$C_t$ is from~\eqref{eq:Cz},~$K_{asym}$ is from~\eqref{eq:asym_coeff}, and
      \begin{equation}\label{eq:c_inf}
        C_{\infty} = \lim_{t\to \infty} C_t = \frac{1}{1-\zeta}e^{\alpha D_{\mathcal{Z}_{attract}(\epsilon)}}.
    \end{equation}
\end{theorem}

Theorem~\ref{thm:avg_entropy} quantifies 
how much the expected empirical entropy rate of a private state trajectory differs 
from the true entropy rate~$H(\mathcal{S})$. While the expected change in entropy is bounded by Theorem~\ref{thm:avg_entropy}, the probability that a single trajectory's entropy differs from the entropy rate~$H(\mathcal{S})$ by a substantial amount is also of interest. Therefore, we next state a McDiarmid-like concentration bound on the empirical entropy of private state trajectories. This bound depends on the ``mixing time'' of the Markov chain~$M$, i.e., the time for~$p(t)$ to converge to~$\mu$, defined as
\begin{equation}\label{eq:mix_time}
    t_{mix} = \min\left\{t\geq 0 :\max_{y\in\mathcal{S}}\left[\max_{A\subseteq \mathcal{S}} \left|\prob{Y_t\in A \mid Y_0 = y}-\mu_A \right| \right]\leq \frac{1}{4}    \right\},
\end{equation}
where~$\mu_A$ is the stationary distribution of the subset~$A\subseteq \mathcal{S}$, which is equal to~$\mu$ but with entries for states~$s \in \mathcal{S} \backslash A$ removed.

\begin{theorem}[Solution to Problem~\ref{prob:3}]\label{thm:cont_entropy}
    Let all conditions of Theorem~\ref{thm:avg_entropy} hold. Let~$t_{mix}$ be defined as in~\eqref{eq:mix_time}. Then for~$n\geq 2$ we have
        $\mathbb{P}\left(H_n(w') - \frac{n}{n-1}H(\mathcal{S})\geq \phi_{\epsilon} +\nu\right) \leq \exp\left(-\frac{2\nu^2(n-1)^2}{9(4n-6)D_{max}^2t_{mix}} \right),$
    where~$\phi_{\epsilon} = \frac{n}{n-1}\left(\frac{1}{\psi_2}+K_{asym}\right)\left(\frac{\sqrt{\pi}}{2\sqrt{\frac{2}{D_{max}}}}\erf\left(B_{\epsilon}\sqrt{\frac{2}{D_{max}}}\right)+C_{\infty}\frac{D_{\mathcal{Z}_{attract}(\epsilon)}}{2\gamma_{\mathcal{Z}_{attract}(\epsilon)}}\right) +\frac{\left(1+\psi_2K_{asym}\right)\bar{V}_0}{(n-1)(1-\psi_2)}$ and~$C_{\infty}$ is from~\eqref{eq:c_inf}.
\end{theorem}
Theorem~\ref{thm:cont_entropy} bounds the probability that the entropy of a private state trajectory deviates from the entropy rate of the original Markov chain~$M$ by an amount~$\nu+\phi_{\epsilon}$. For a Markov chain~$M$, the $\eta$-typical set~$A_{\eta}^{(n)}$ is the set of length~$n+1$ state trajectories satisfying~$H(\mathcal{S})-\eta \leq -\frac{1}{n} \log(\prob{s_0s_1\cdots s_n}) \leq H(\mathcal{S})+\eta$. Then, by setting~$\eta = \nu +\phi_{\epsilon}$, Theorem~\ref{thm:cont_entropy} allows users to bound the probability that a private state trajectory is in~$A_{\eta}^{(n)}$ given choices of~$\epsilon$ and~$\nu$. 
Doing so ensures that the private state trajectories remain~$\eta$-typical
with high probability. 
Private trajectories that lie within~$A_{\eta}^{(n)}$ appear as though they were generated by~$M$, which ensures that private trajectories generated by Mechanism~\ref{mech:PF} appear as though they were generated by~$M$.
\section{Empirical Results}\label{sec:sims}

In this section, we solve Problem~\ref{prob:4}. We consider three Markov chain examples of varying sizes whose state trajectories represent sensitive information. First, we consider a small Markov chain ($|\mathcal{S}| = 9$) in the setting of credit migration, where the transition dynamics in~$P$ represent the probability of transitioning from one credit rating to another~\citep{nickell2000stability}. Second, we consider a vehicle traffic example ($|\mathcal{S}| = 43$), specifically traffic in the city of Gainesville, Florida, as seen in~\cite{chen2023differentialsymbolic}, where the transition dynamics in~$P$ encode the probability of driving from one segment of road onto another. Finally, we consider internet traffic in the form of a condensed subset of Wikipedia articles in the Wikispeedia dataset ($|\mathcal{S}| = 3,817$), where the transition dynamics in~$P$ encode the probability of transitioning from one Wikipedia article to another~\citep{snapnets}. In all examples, we consider the initial distribution~$p_0$ to be the Dirac distribution centered on the initial state~$s_0$. See Appendix~\ref{app:sims} for data processing details, details on calibrating
the mechanism from~\cite{chen2023differentialsymbolic}, additional simulation results, and sample private state trajectories.


\subsection{Error Probability}\label{subsec:error_prob}
Figure~\ref{fig:probability_gainesville} shows the probability of the shortest-path distance ever exceeding an error amount~$v$ for~$\epsilon = 1$, computed via Monte-Carlo simulation with $10,000$ privately generated state trajectories based on the same sensitive state trajectory. In the Gainesville traffic example, we see that large errors, i.e.,~$v>5$ nats, appear frequently when using the mechanism from~\cite{chen2023differentialsymbolic}, and such errors appear with probability~$0.39$. However, using Mechanism~\ref{mech:PF}, the probability of such errors is~$0.06$, which is $84\%$ smaller. Thus, we find that in a city traffic example, Mechanism~\ref{mech:PF} yields private state trajectories that, with high probability, maintain structural similarity to the underlying sensitive state trajectories. Across all examples, we see the probability of large errors decrease by multiple orders of magnitude relative to~\cite{chen2023differentialsymbolic}, which highlights the ability of our mechanism to produce more accurate private state trajectories than the state of the art.
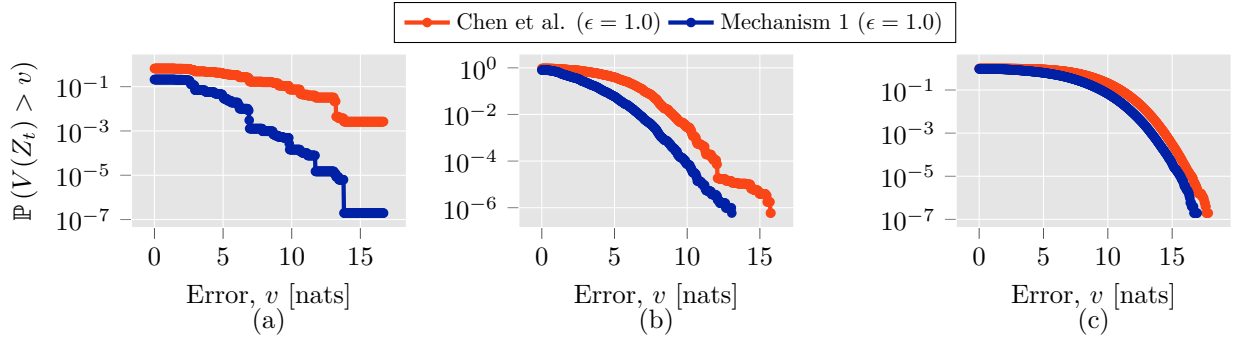
\begin{figure*}
\centering
    \begin{subfigure}{0.3\linewidth}
        \centering
%
%
\definecolor{chocolate2267451}{RGB}{226,74,51}
\definecolor{dimgray85}{RGB}{85,85,85}
\definecolor{gainsboro229}{RGB}{229,229,229}
\definecolor{lightgray204}{RGB}{204,204,204}
\definecolor{steelblue52138189}{RGB}{52,138,189}
\definecolor{black}{RGB}{0, 0, 0}
\definecolor{GTblue}{RGB}{0, 48, 87}
\definecolor{GTgold}{RGB}{179, 163, 105}
\definecolor{UFOrange}{RGB}{250, 70, 22}
\definecolor{UFblue}{RGB}{0, 33, 165}
\begin{tikzpicture}

\begin{axis}[%
width=0.22\figW,
height=0.51\figH,
axis background/.style={fill=gainsboro229},
axis line style={white},
scale only axis,
xlabel=\textcolor{black}{Error, $v$ [nats]},
xtick style={color=dimgray85},
x grid style={white},
yminorticks=true,
y grid style={white},
ylabel=\textcolor{black}{$\prob{V(Z_t)>v}$},
xmajorgrids,
ymajorgrids,
yminorgrids,
tick align=outside,
tick pos=left,
yminorticks=true,
max space between ticks=20,
legend style={nodes={scale=0.85, transform shape}, at={(3.0,1.3)}},
legend columns=2,
ymode = log,
]
        






        \addplot[UFOrange, ultra thick, mark={*}, mark size={1.0 pt}, unbounded coords=discard] table [col sep=comma, x=v, y=CCDF_Chen] {Figures/moodys_ccdf_data.csv};
        \addlegendentry{Chen et al. ($\epsilon = 1.0$)}
        \addplot[UFblue, ultra thick, mark={*}, mark size={1.0 pt}] table [col sep=comma, x=v, y=CCDF_Mechanism4, unbounded coords=discard] {Figures/moodys_ccdf_data.csv};
        \addlegendentry{Mechanism~\ref{mech:PF} ($\epsilon = 1.0$)}

\coordinate (ttl) at (axis description cs:0.5,-0.67);
\end{axis}
\node[anchor=south,text width=7cm,align=center] at (ttl) 
{(a)};
\end{tikzpicture}%

        
    \end{subfigure}
    \hfill 
    \begin{subfigure}{0.34\textwidth}
        \centering
%
%
\definecolor{chocolate2267451}{RGB}{226,74,51}
\definecolor{dimgray85}{RGB}{85,85,85}
\definecolor{gainsboro229}{RGB}{229,229,229}
\definecolor{lightgray204}{RGB}{204,204,204}
\definecolor{steelblue52138189}{RGB}{52,138,189}
\definecolor{black}{RGB}{0, 0, 0}
\definecolor{GTblue}{RGB}{0, 48, 87}
\definecolor{GTgold}{RGB}{179, 163, 105}
\definecolor{UFOrange}{RGB}{250, 70, 22}
\definecolor{UFblue}{RGB}{0, 33, 165}
\begin{tikzpicture}

\begin{axis}[%
width=0.22\figW,
height=0.51\figH,
axis background/.style={fill=gainsboro229},
axis line style={white},
scale only axis,
xlabel=\textcolor{black}{Error, $v$ [nats]},
xtick style={color=dimgray85},
x grid style={white},
yminorticks=true,
y grid style={white},
xmajorgrids,
ymajorgrids,
yminorgrids,
tick align=outside,
tick pos=left,
yminorticks=true,
max space between ticks=20,
legend pos = north east,
legend style={nodes={scale=0.65, transform shape}},
legend columns=1,
ymode = log,
]
        






        \addplot[UFOrange, ultra thick, mark={*}, mark size={1.0 pt}, unbounded coords=discard] table [col sep=comma, x=v, y=CCDF_Chen] {Figures/gainesville_ccdf_data.csv};
        \addplot[UFblue, ultra thick, mark={*}, mark size={1.0 pt}] table [col sep=comma, x=v, y=CCDF_Mechanism4, unbounded coords=discard] {Figures/gainesville_ccdf_data.csv};
\coordinate (ttl) at (axis description cs:0.5,-0.67);
\end{axis}
\node[anchor=south,text width=7cm,align=center] at (ttl) 
{(b)};
\end{tikzpicture}%

        
    \end{subfigure}
    \hfill
    \begin{subfigure}{0.34\textwidth}
        \centering
%
%
\definecolor{chocolate2267451}{RGB}{226,74,51}
\definecolor{dimgray85}{RGB}{85,85,85}
\definecolor{gainsboro229}{RGB}{229,229,229}
\definecolor{lightgray204}{RGB}{204,204,204}
\definecolor{steelblue52138189}{RGB}{52,138,189}
\definecolor{black}{RGB}{0, 0, 0}
\definecolor{GTblue}{RGB}{0, 48, 87}
\definecolor{GTgold}{RGB}{179, 163, 105}
\definecolor{UFOrange}{RGB}{250, 70, 22}
\definecolor{UFblue}{RGB}{0, 33, 165}
\begin{tikzpicture}

\begin{axis}[%
width=0.22\figW,
height=0.51\figH,
axis background/.style={fill=gainsboro229},
axis line style={white},
scale only axis,
xlabel=\textcolor{black}{Error, $v$ [nats]},
xtick style={color=dimgray85},
x grid style={white},
yminorticks=true,
y grid style={white},
xmajorgrids,
ymajorgrids,
yminorgrids,
tick align=outside,
tick pos=left,
yminorticks=true,
max space between ticks=20,
legend pos = north east,
legend style={nodes={scale=0.65, transform shape}},
legend columns=1,
ymode = log,
]
        






        \addplot[UFOrange, ultra thick, mark={*}, mark size={1.0 pt}, unbounded coords=discard] table [col sep=comma, x=v, y=CCDF_Chen] {Figures/wikispeedia_ccdf_data.csv};
        \addplot[UFblue, ultra thick, mark={*}, mark size={1.0 pt}] table [col sep=comma, x=v, y=CCDF_Mechanism4, unbounded coords=discard] {Figures/wikispeedia_ccdf_data.csv};
\coordinate (ttl) at (axis description cs:0.5,-0.67);
\end{axis}
\node[anchor=south,text width=7cm,align=center] at (ttl) 
{(c)};
\end{tikzpicture}%

        
    \end{subfigure}
    \caption{Probability of large errors in private trajectories in the (a) credit migration, (b) Gainesville traffic, and (c) Wikispeedia datasets. Across all examples, the probability of large errors is substantially lower than in~\cite{chen2023differentialsymbolic}, with the largest decrease occurring in the credit migration example at~$v = 15$; the probability of errors this large is~$0.003$ under~\cite{chen2023differentialsymbolic} and~$2\times 10^{-7}$ under Mechanism~\ref{mech:PF}, which is a~$4$-order-of-magnitude decrease.}
    \label{fig:probability_gainesville}
\end{figure*}

\subsection{Empirical Entropy}
Figure~\ref{fig:entropy_gainesville} shows the average empirical entropy over~$10,000$ private state trajectories produced by Mechanism~\ref{mech:PF} (with~$\rho = 1$) and~\cite{chen2023differentialsymbolic}.  
We vary the trajectory length for each~$\epsilon \in \{1, 3, 5\}$ and show~$2$ standard deviation bounds on the empirical entropy from~\eqref{eq:entropy_y} of the state trajectories generated by the original Markov chain. For all values of~$\epsilon$ and all trajectory lengths, the average empirical entropy of Mechanism~\ref{mech:PF} is lower than~\cite{chen2023differentialsymbolic}. In the Gainesville traffic example, we find up to an~$80\%$ decrease in average empirical entropy compared to~\cite{chen2023differentialsymbolic} at~$\epsilon = 3$, which highlights the increased similarity of private state trajectories generated by Mechanism~\ref{mech:PF} to their underlying sensitive state trajectories. Additionally, the entropy in~\cite{chen2023differentialsymbolic} increases with trajectory length, while the entropy of Mechanism~\ref{mech:PF} remains constant after a certain trajectory length. This difference occurs because a longer trajectory increases the probability that the mechanism in~\cite{chen2023differentialsymbolic} takes an incorrect transition and begins random walking, in which case the entropy of their mechanism approaches that of a uniform random walk over the state space.

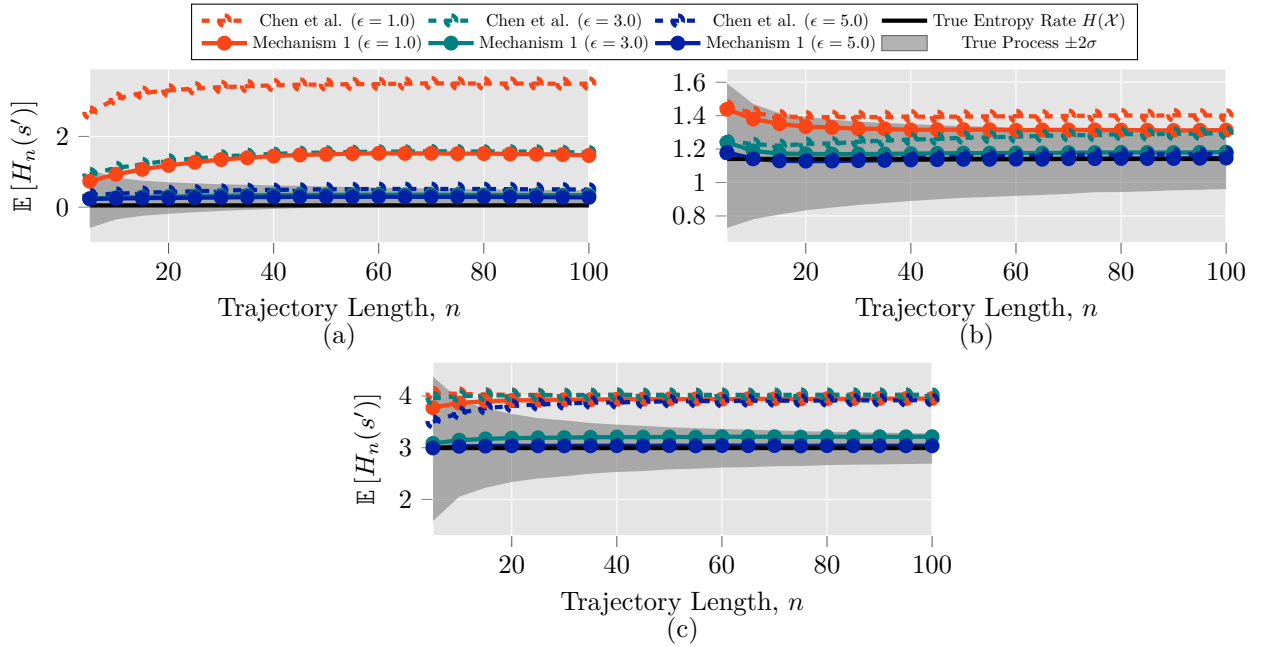
\begin{figure*}[!t]
\centering
    \begin{subfigure}{0.47\linewidth}
        \centering
%
%
\definecolor{chocolate2267451}{RGB}{226,74,51}
\definecolor{dimgray85}{RGB}{85,85,85}
\definecolor{gainsboro229}{RGB}{229,229,229}
\definecolor{lightgray204}{RGB}{204,204,204}
\definecolor{steelblue52138189}{RGB}{52,138,189}
\definecolor{black}{RGB}{0, 0, 0}
\definecolor{GTblue}{RGB}{0, 48, 87}
\definecolor{GTgold}{RGB}{179, 163, 105}
\definecolor{UFOrange}{RGB}{250, 70, 22}
\definecolor{UFblue}{RGB}{0, 33, 165}
\begin{tikzpicture}

\begin{axis}[%
width=0.4\figW,
height=0.51\figH,
axis background/.style={fill=gainsboro229},
axis line style={white},
scale only axis,
xlabel=\textcolor{black}{Trajectory Length, $n$},
xtick style={color=dimgray85},
x grid style={white},
yminorticks=true,
y grid style={white},
ylabel=\textcolor{black}{$\E{H_n(s')}$},
xmajorgrids,
ymajorgrids,
yminorgrids,
tick align=outside,
tick pos=left,
legend pos = south east,
legend style={nodes={scale=0.65, transform shape}, at={(2.1,1.05)}},
legend columns=4,
 xmin=5,
 xmax=100,
]
\addplot[name path=upper, draw=none, , forget plot] table [col sep=comma, x=T, y=TrueUpper] {Figures/moodys_fidelity_data.csv};
        \addplot[name path=lower, draw=none, forget plot] table [col sep=comma, x=T, y=TrueLower] {Figures/moodys_fidelity_data.csv};
        






        \addplot[UFOrange, dashed, ultra thick, mark={o}, mark size={2.0 pt}] table [col sep=comma, x=T, y=Chen_1_0] {Figures/moodys_fidelity_data.csv};
        \addlegendentry{Chen et al. ($\epsilon = 1.0$)}

        \addplot[teal, dashed, ultra thick, mark={o}, mark size={2.0 pt}] table [col sep=comma, x=T, y=Chen_3_0] {Figures/moodys_fidelity_data.csv};
        \addlegendentry{Chen et al. ($\epsilon = 3.0$)}

        \addplot[UFblue, dashed, ultra thick, mark={o}, mark size={2.0 pt}] table [col sep=comma, x=T, y=Chen_5_0] {Figures/moodys_fidelity_data.csv};
        \addlegendentry{Chen et al. ($\epsilon = 5.0$)}
        \addplot[black, ultra thick] table [col sep=comma, x=T, y=TrueMean] {Figures/moodys_fidelity_data.csv};
        \addlegendentry{True Entropy Rate $H(\mathcal{X})$}
        \addplot[UFOrange, ultra thick, mark={*}, mark size={2.0 pt}] table [col sep=comma, x=T, y=PF_1_0] {Figures/moodys_fidelity_data.csv};
        \addlegendentry{Mechanism~\ref{mech:PF} ($\epsilon = 1.0$)}


        \addplot[teal, ultra thick, mark={*}, mark size={2.0 pt}] table [col sep=comma, x=T, y=PF_3_0] {Figures/moodys_fidelity_data.csv};
        \addlegendentry{Mechanism~\ref{mech:PF} ($\epsilon = 3.0$)}

        \addplot[UFblue, ultra thick, mark={*}, mark size={2.0 pt}] table [col sep=comma, x=T, y=PF_5_0] {Figures/moodys_fidelity_data.csv};
        \addlegendentry{Mechanism~\ref{mech:PF} ($\epsilon = 5.0$)}

        \addplot[fill=black!50, opacity=0.6] fill between[of=upper and lower];
        \addlegendentry{True Process $\pm 2\sigma$}
\coordinate (ttl) at (axis description cs:0.5,-0.675);
\end{axis}
\node[anchor=south,text width=7cm,align=center] at (ttl) 
{(a)};
\end{tikzpicture}%

        
    \end{subfigure}
    \hfill 
    \begin{subfigure}{0.47\textwidth}
        \centering
%
%
\definecolor{chocolate2267451}{RGB}{226,74,51}
\definecolor{dimgray85}{RGB}{85,85,85}
\definecolor{gainsboro229}{RGB}{229,229,229}
\definecolor{lightgray204}{RGB}{204,204,204}
\definecolor{steelblue52138189}{RGB}{52,138,189}
\definecolor{black}{RGB}{0, 0, 0}
\definecolor{GTblue}{RGB}{0, 48, 87}
\definecolor{GTgold}{RGB}{179, 163, 105}
\definecolor{UFOrange}{RGB}{250, 70, 22}
\definecolor{UFblue}{RGB}{0, 33, 165}
\begin{tikzpicture}

\begin{axis}[%
width=0.4\figW,
height=0.51\figH,
axis background/.style={fill=gainsboro229},
axis line style={white},
scale only axis,
xlabel=\textcolor{black}{Trajectory Length, $n$},
xtick style={color=dimgray85},
x grid style={white},
yminorticks=true,
y grid style={white},
xmajorgrids,
ymajorgrids,
yminorgrids,
tick align=outside,
tick pos=left,
legend pos = south east,
legend style={nodes={scale=0.65, transform shape}},
legend columns=2,
 xmin=5,
 xmax=100,
]
\addplot[name path=upper, draw=none, , forget plot] table [col sep=comma, x=T, y=TrueUpper] {Figures/Gainesville_fidelity_data.csv};
        \addplot[name path=lower, draw=none, forget plot] table [col sep=comma, x=T, y=TrueLower] {Figures/Gainesville_fidelity_data.csv};
        






        \addplot[UFOrange, dashed, ultra thick, mark={o}, mark size={2.0 pt}] table [col sep=comma, x=T, y=Chen_1_0] {Figures/Gainesville_fidelity_data.csv};

        \addplot[UFOrange, ultra thick, mark={*}, mark size={2.0 pt}] table [col sep=comma, x=T, y=PF_1_0] {Figures/Gainesville_fidelity_data.csv};


        \addplot[teal, dashed, ultra thick, mark={o}, mark size={2.0 pt}] table [col sep=comma, x=T, y=Chen_3_0] {Figures/Gainesville_fidelity_data.csv};

        \addplot[teal, ultra thick, mark={*}, mark size={2.0 pt}] table [col sep=comma, x=T, y=PF_3_0] {Figures/Gainesville_fidelity_data.csv};

        \addplot[UFblue, dashed, ultra thick, mark={o}, mark size={2.0 pt}] table [col sep=comma, x=T, y=Chen_5_0] {Figures/Gainesville_fidelity_data.csv};

        \addplot[UFblue, ultra thick, mark={*}, mark size={2.0 pt}] table [col sep=comma, x=T, y=PF_5_0] {Figures/Gainesville_fidelity_data.csv};
        \addplot[black, ultra thick] table [col sep=comma, x=T, y=TrueMean] {Figures/Gainesville_fidelity_data.csv};

        \addplot[fill=black!50, opacity=0.6] fill between[of=upper and lower];
\coordinate (ttl) at (axis description cs:0.5,-0.675);
\end{axis}
\node[anchor=south,text width=7cm,align=center] at (ttl) 
{(b)};
\end{tikzpicture}%

        
    \end{subfigure}
    
    \begin{subfigure}{0.45\textwidth}
        \centering
%
%
\definecolor{chocolate2267451}{RGB}{226,74,51}
\definecolor{dimgray85}{RGB}{85,85,85}
\definecolor{gainsboro229}{RGB}{229,229,229}
\definecolor{lightgray204}{RGB}{204,204,204}
\definecolor{steelblue52138189}{RGB}{52,138,189}
\definecolor{black}{RGB}{0, 0, 0}
\definecolor{GTblue}{RGB}{0, 48, 87}
\definecolor{GTgold}{RGB}{179, 163, 105}
\definecolor{UFOrange}{RGB}{250, 70, 22}
\definecolor{UFblue}{RGB}{0, 33, 165}
\begin{tikzpicture}

\begin{axis}[%
width=0.4\figW,
height=0.51\figH,
axis background/.style={fill=gainsboro229},
axis line style={white},
scale only axis,
xlabel=\textcolor{black}{Trajectory Length, $n$},
xtick style={color=dimgray85},
x grid style={white},
yminorticks=true,
y grid style={white},
ylabel=\textcolor{black}{$\E{H_n(s')}$},
xmajorgrids,
ymajorgrids,
yminorgrids,
tick align=outside,
tick pos=left,
legend pos = south east,
legend style={nodes={scale=0.65, transform shape}},
legend columns=2,
 xmin=5,
 xmax=100,
]
\addplot[name path=upper, draw=none, , forget plot] table [col sep=comma, x=T, y=TrueUpper] {Figures/wikispeedia_fidelity_data.csv};
        \addplot[name path=lower, draw=none, forget plot] table [col sep=comma, x=T, y=TrueLower] {Figures/wikispeedia_fidelity_data.csv};
        






        \addplot[UFOrange, dashed, ultra thick, mark={o}, mark size={2.0 pt}] table [col sep=comma, x=T, y=Chen_1_0] {Figures/wikispeedia_fidelity_data.csv};

        \addplot[UFOrange, ultra thick, mark={*}, mark size={2.0 pt}] table [col sep=comma, x=T, y=PF_1_0] {Figures/wikispeedia_fidelity_data.csv};


        \addplot[teal, dashed, ultra thick, mark={o}, mark size={2.0 pt}] table [col sep=comma, x=T, y=Chen_3_0] {Figures/wikispeedia_fidelity_data.csv};

        \addplot[teal, ultra thick, mark={*}, mark size={2.0 pt}] table [col sep=comma, x=T, y=PF_3_0] {Figures/wikispeedia_fidelity_data.csv};

        \addplot[UFblue, dashed, ultra thick, mark={o}, mark size={2.0 pt}] table [col sep=comma, x=T, y=Chen_5_0] {Figures/wikispeedia_fidelity_data.csv};

        \addplot[UFblue, ultra thick, mark={*}, mark size={2.0 pt}] table [col sep=comma, x=T, y=PF_5_0] {Figures/wikispeedia_fidelity_data.csv};
        \addplot[black, ultra thick] table [col sep=comma, x=T, y=TrueMean] {Figures/wikispeedia_fidelity_data.csv};

        \addplot[fill=black!50, opacity=0.6] fill between[of=upper and lower];
\coordinate (ttl) at (axis description cs:0.5,-0.675);
\end{axis}
\node[anchor=south,text width=7cm,align=center] at (ttl) 
{(c)};
\end{tikzpicture}%

        
    \end{subfigure}
    \caption{Empirical entropy with varying~$\epsilon$ and trajectory length in the (a) credit migration, (b) Gainesville traffic, and (c) Wikispeedia examples. In Mechanism~\ref{mech:PF} the average empirical entropy remains constant beyond a certain trajectory length, while the average empirical entropy in~\cite{chen2023differentialsymbolic} grows. This growth is a result of the random walking in their mechanism, because longer trajectories have a higher probability of an incorrect transition, which leads to their mechanism to begin random walking. }
    \label{fig:entropy_gainesville}
\end{figure*}
\section{Conclusion}
We have presented a framework for privatizing Markov chain state trajectories \emph{online}
in a way that improves on the state of the art. We proved that with high probability the errors induced by our mechanism remain small, and that with high probability private state trajectories remain in the~$\eta$-typical set of the Markov chain that generated the underlying sensitive state trajectory. Empirically, we improved on the state of the art by lowering the expected entropy by up to~$80\%$ and lowering the probability of large errors by up to~$4$ orders of magnitude. Future work will develop an \emph{offline} mechanism and will develop a framework to privatize state and action trajectories in other stochastic processes, including Markov decision processes.

\section*{Acknowledgements}
This work was partially supported by ONR under grant N00014-24-1-2432, AFOSR under grant FA9550-19-1-0169, and NSF Graduate Research Fellowship under grant DGE-2039655. Any opinions, findings and conclusions or recommendations expressed in this material are those of the authors and do not necessarily reflect the views of sponsoring agencies.

\bibliographystyle{tmlr}
\bibliography{references}

\appendix

\section{Supporting Definitions and Lemmas}\label{app:supporting}
The following definitions and lemmas are used throughout these appendices.
\begin{definition}[Quasi-metric for Directed Graphs]\label{def:quasi}
    Fix a directed weighted graph~$\mathcal{G} = (V, E, W)$ such that all the weights in~$W$ are non-negative. Given~$x, y, z\in V$, the shortest-path distance~$G$ from~\eqref{eq:G} forms a quasi-metric over the node set~$V$, i.e.,~$G$ has the following properties:
    \begin{enumerate}
        \item $G(x, y)\geq 0$
        \item $G(x, y) = 0 \iff x = y$ 
        \item $G(x, y) \leq G(x, z) +G(z, y)$. 
    \end{enumerate}
\end{definition}
Throughout these appendices, we refer to the third property of Definition~\ref{def:quasi} as the ``triangle inequality for directed graphs''. 
The following lemma is used to derive the concentration bound in Theorem~\ref{thm:cont_entropy}.
\begin{lemma}[McDiarmid's Inequality for Markov Chains; \cite{paulin2015concentration}]\label{lem:mcdairmid}
    Let~$X_1\ldots, X_n$ be a time homogeneous Markov chain taking values in a Polish state space~$\Lambda = \Lambda_1\times \ldots \times \Lambda_n$ with mixing time~$t_{mix}$ as defined in~\eqref{eq:mix_time}.
    Suppose that~$f:\Lambda\to \mathbb{R}$ satisfies
    \begin{equation}
        f(x) - f(y) \leq \sum_{i=1}^n c_i\ones[x_i\neq y_i]
    \end{equation}
    for some~$c\in\mathbb{R}^N_+$. Then for any~$k\geq 0$,
    \begin{equation}
        \prob{f(X)-\E{f(X)}\geq k}\leq \exp\left(\frac{-2k^2}{9\norm{c}^2_2t_{mix}} \right).  
    \end{equation}
    \hfill $\blacklozenge$
\end{lemma}

\section{Existing privacy mechanism from~\cite{chen2023differentialsymbolic}}\label{ap:bo_mech}
Consider a Markov\footnote{In order to match the definition of adjacency from~\cite{chen2023differentialsymbolic},
the definition of a Markov chain in this section considers a fixed initial
state~$s_0$ rather than a distribution over initial states. Therefore, a sensitive
word and its privatization both begin at~$s_0$.} 
chain~$M = (\mathcal{S}, P, s_0)$. Given a state trajectory~$w = s_0s_1\cdots s_n\in\mathcal{S}^{n+1}$, we want to privately generate a new trajectory~$w' = s_0s_1'\cdots s_n'$ 
where each state~$s'_t$ is selected \emph{online}. That is, the private state trajectory~$w'$ is generated at the same time as~$w$. 
Consider generating~$s_t'$ for a fixed~$t \in [n]\setminus\{0\}$. 
First the actual state of the Markov chain transitions from~$s_{t-1}$ to~$s_t$. Then the mechanism selects a state~$s_t'\in\mathcal{S}$ to publicly share. Let~$v, w$ be two state trajectories of length~$n+1$. 
Then, the Hamming distance between~$v$ and~$w$, denoted~$d_H(v, w)$, is the number of entries where~$v$ and $w$ differ, i.e.,~$d_H(w, v) = |\{i\mid w_i\neq v_i\}|$. In~\cite{chen2023differentialsymbolic}, they consider the following notion of adjacency for Markov chain state trajectories.

\begin{definition}[Word Adjacency; \citep{chen2023differentialsymbolic}]\label{def:bo_adj}
    Fix a number of state-to-state transitions~$n\in\mathbb{N}$ 
    and an adjacency parameter~$b\in\mathbb{N}$. Two trajectories~$w, v\in\mathcal{S}^{n+1}$ are said to be adjacent if~$d_H(w, v)\leq b$. \hfill$\blacktriangle$
\end{definition}

Next we formally state the mechanism in~\cite{chen2023differentialsymbolic}. First, define
\begin{equation}
    \beta(s_t, s_{t-1}') = 
    \begin{cases}
        1 & \text{if } \prob{s_t\mid s_{t-1}'}>0\\
        0 & \text{otherwise}
    \end{cases}
\end{equation}
as an indicator function that returns~$1$ if there exists a feasible transition from $s_{t-1}'$ to~$s_t$, and~$0$ otherwise. A transition from~$s_{t-1}'$ to~$s_t$ is referred to as a ``correct'' transition because it ensures that the state at time~$t$ in the private state trajectory is the same as the non-private state at time~$t$ in the true state trajectory. To minimize the Hamming distance between the private state trajectory and the underlying sensitive state trajectory, a correct transition should be the most likely to occur, which is what motivates the following mechanism.
\begin{mechanism}[\cite{chen2023differentialsymbolic}]\label{mech:bo}
    Fix a probability space~$(\Omega, \mathcal{F}, \mathbb{P})$ and a Markov chain~$M = (\mathcal{S}, P, s_0)$. Given a state trajectory~$w = s_0s_1\cdots s_n\in\mathcal{S}^{n+1}$, select an output state trajectory~$w' = s_0s_1'\cdots s_n'\in\mathcal{S}^{n+1}$
    such that
    \begin{equation}
        \prob{\mechbo(s_t) = s_t'} = \begin{cases}
            \tau(s_t, s'_{t-1}) & \text{if } s_t = s_t' \text{ and } \beta(s_t', s_{t-1}')=1\\
            \frac{1-\tau(s_t, s'_{t-1})\beta(s_t, s_{t-1}')}{|N(s_{t-1}')|-\beta(s_t,s_{t-1}')} & \text{if } s_t \neq s_t' \text{ and } \beta(s_t', s_{t-1}')=1\\
            0 & \text{otherwise.} 
        \end{cases}
    \end{equation}
    \hfill $\triangle$
\end{mechanism}

In Mechanism~\ref{mech:bo}, the probability of a correct transition is~$\tau(s_t, s_{t-1}')$. If~$N(s_{t-1}')$ is the neighborhood set of state~$s_{t-1}'$, and the correct transition is selected with probability~$\tau(s_t, s_{t-1}')$, then the remaining states are each selected with probability
\begin{equation}
    \frac{1-\tau(s_t, s_{t-1}')}{|N(s_{t-1}')| - 1}.
\end{equation}

If the correct transition is not feasible, then the next state is sampled uniformly from the neighborhood set~$N(s_{t-1}')$. All non-feasible transitions have probability~$0$ to ensure that the private state trajectory is a feasible (though not necessarily likely) state trajectory from the Markov chain~$M$.

\begin{theorem}[\cite{chen2023differentialsymbolic}]
    Mechanism~\ref{mech:bo} is~$\epsilon$-differentially private with respect to the adjacency relation in Definition~\ref{def:bo_adj} if
    \begin{equation}
        \tau(s_t, s_{t-1}') = \frac{1}{(|N(s_{t-1}')|-1)\exp(-\frac{\epsilon}{b})+1}.
    \end{equation}
\end{theorem}

One benefit of Mechanism~\ref{mech:bo} is that it provably gives the highest probability of making the correct transition of any possible mechanism. However, the major drawback is that this mechanism exhibits random walking if the correct transition is infeasible. Specifically, if~$\beta(s_t, s_{t-1}') = 0$, then all of the mechanism's output probabilities become
\begin{equation}
    \prob{\mathcal{M}_2(s_t) = s_t'} =\begin{cases}
        \frac{1}{|N(s_{t-1}')|} & \text{if } \beta(s_t', s_{t-1}')=1\\
        0 & \text{otherwise},
    \end{cases}
\end{equation} 
which initiates a uniform random walk over the state space.
This random walking behavior 
maximizes the probability of the correct transition when available. However, in a setting where the Markov chain states have structural similarities, e.g., in a city or internet traffic setting, there exists a notion of states being closer together and further away, 
which means that some incorrect transitions introduce larger errors than others. 
Thus, random walking may result in private state trajectories with little structural similarity to the underlying sensitive state trajectory.

\section{Alternate Mechanism Design}\label{ap:alt}
In this section\footnote{To continue matching the definition of adjacency in~\cite{chen2023differentialsymbolic}, we use the definition of a Markov chain considered in Appendix~\ref{ap:bo_mech} with a fixed initial state~$s_0$.}, we show how a mechanism similar to Mechanism~\ref{mech:PF} may be designed using Definition~\ref{def:bo_adj}. Such a mechanism uses the notion of adjacency for state trajectories 
in~\cite{chen2023differentialsymbolic}. However, we demonstrate why this notion yields a privacy mechanism with poor accuracy at strong levels of privacy, and thus why a different adjacency relation is useful.
First, we must compute the sensitivity of the utility function~$u$ with respect to the adjacency relation in Definition~\ref{def:bo_adj}. 

\begin{lemma}[Sensitivity]\label{lem:sensitivity}
    Fix a Markov chain~$M=(\mathcal{S}, P, s_0)$ 
    satisfying Assumption~\ref{ass:MC} and an adjacency parameter~$b\in\mathbb{N}\setminus\{0\}$. Let~$w = s_0s_1\cdots s_n\in \mathcal{S}^{n+1}$ and $v = s_0\tilde{s}_1\cdots \tilde{s}_n\in\mathcal{S}^{n+1}$ 
    be two adjacent state trajectories in the sense of Definition~\ref{def:bo_adj}. The sensitivity of~$u$ from~\eqref{eq:udef} under this adjacency relation is
    \begin{equation}
        \Delta u = \max_{y\in \mathcal{S}^{n+1}}\max_{\substack{w, v\in \mathcal{S}^{n+1}\\ \Adjacentbo{w}{v}}} |u(w, y) - u(v, y)| \leq bD_{max},
    \end{equation}
    where~$D_{max}$ is the diameter of the graph~$\mathcal{G}$ from Definition~\ref{def:MC-graph}.\hfill$\blacklozenge$
\end{lemma}

\begin{proof}
        Let~$y = y_0\cdots y_n\in\mathcal{S}^{n+1}$. Substituting the utility function into the definition of sensitivity gives
    \begin{equation}\label{eq:sens_step1}
        \Delta u = \max_{y\in \mathcal{S}^{n+1}}\max_{\substack{w, v\in \mathcal{S}^{n+1}\\ \Adjacentbo{w}{v}}} \left|\sum_{t\in[n]}G(y_t, \tilde{s}_t) - \sum_{t\in[n]}G(y_t, s_t)\right|.
    \end{equation}
    Let~$\mathcal{K}$ be the set of indices~$t \in [n]$ such that~$\tilde{s}_t\neq s_t$. From Definition~\ref{def:word_adj},~$|\mathcal{K}|\leq b$. We then have
    \begin{equation}\label{eq:sens_step2}
        \Delta u = \max_{y\in \mathcal{S}^{n+1}}\max_{\substack{w, v\in \mathcal{S}^{n+1}\\ \Adjacentbo{w}{v}}} \sum_{t\in\mathcal{K}}\left|G(y_t, \tilde{s}_t) - G(y_t, s_t)\right|.
    \end{equation}
    From the triangle inequality for directed graphs
    in Definition~\ref{def:quasi}, we have 
    \begin{equation}
        G(y_t, s_t) \leq G(y_t, \tilde{s}_t)+G(\tilde{s}_t, s_t),
    \end{equation}
    and rearranging yields
    \begin{equation}\label{eq:call_this_something}
        |G(y_t, \tilde{s}_t) -G(y_t, s_t)| \leq G(\tilde{s}_t, s_t).
    \end{equation}
    The farthest apart two states can be is the graph diameter, and therefore~$G(\tilde{s}_t, s_t)\leq D_{max}$. Accordingly, using~\eqref{eq:sens_step2} and~\eqref{eq:call_this_something} gives
    \begin{equation}
        \Delta u \leq \sum_{t\in\mathcal{K}}D_{max} \leq bD_{max}.
    \end{equation}
\end{proof}

Next we state the mechanism of interest. Given the neighborhood set of the previous private state, namely~$N(s_{t-1}')$, and the current true state, $s_t$, we select the next private state~$s_t'$ via the following mechanism.

\begin{mechanism}\label{mech:1}
Fix a Markov chain~$M = (\mathcal{S}, P, p_0)$ satisfying Assumption~\ref{ass:MC}, an adjacency parameter~$b\in\mathbb{N}\setminus\{0\}$, and a privacy parameter~$\epsilon\geq 0$. Consider~$\mathcal{G} = (\mathcal{S}, E, W)$ from Definition~\ref{def:MC-graph} and~$G$ be from~\eqref{eq:G}. Let~$D_{max}$ be the diameter of~$\mathcal{G}$. At each time~$t$, given the current state of the sensitive trajectory~$s_t$, the mechanism~$\mathcal{M}_3$ randomly selects the output state according to
    \begin{equation}\label{eq:mech}
        \prob{\mathcal{M}_3(s_t) = s_t'} = \frac{\exp\left(-\frac{\epsilon G(s_t', s_t)}{2bD_{max}}\right)}{\sum_{y\in N(s_{t-1}')}\exp\left(-\frac{\epsilon G(y, s_t)}{2bD_{max}}\right)}. 
    \end{equation}
    \hfill $\triangle$
\end{mechanism}

The diameter of graphs in many practical applications is large, resulting in a mechanism that provides similar performance to random walking because
the numerator~$\exp\left(-\frac{\epsilon G(s_t', s_t)}{2bD_{max}}\right)$ approaches~$1$
when~$D_{max}$ is large. 
Thus,  Mechanism~\ref{mech:1}'s performance scales poorly with strong privacy and large graphs, which motivates the approach in Section~\ref{sec:mech_design}.

\section{Exponential Mechanism for Online Trajectory Privacy}\label{ap:exp_mechanism}
The exponential mechanism is a privacy mechanism related to the permute-and-flip mechanism by virtue of also being used to privately select an approximation to some sensitive data~\citep{dwork2014algorithmic}. While the permute-and-flip (PF) mechanism is often more accurate than (and is always at least as accurate as) the exponential mechanism, analysis of the PF mechanism is often difficult due to the lack of closed form for the PMF it induces over outputs. However, the PF mechanism stochastically dominates the exponential mechanism, i.e., the PF mechanism selects the highest-utility choice with at least the same probability as the exponential mechanism. 
Therefore, accuracy and privacy results for the exponential mechanism can be used to characterze the PF mechanism~\citep[Theorem 2]{mckenna2020permute}. Accordingly, we state the following exponential mechanism that we use to prove results about Mechanism~\ref{mech:PF}.

\begin{mechanism}[Alternate Solution to Problem~\ref{prob:1}]\label{mech:2}
    Fix a Markov chain~$M = (\mathcal{S}, P, p_0)$ satisfying Assumption~\ref{ass:MC}, an adjacency parameter~$\rho>0$, and a privacy parameter~$\epsilon\geq0$. Let~$\mathcal{G} = (\mathcal{S}, E, W)$ be from Definition~\ref{def:MC-graph}. At each time~$t \in [n]$, given the state~$s_t$, the mechanism~$\mathcal{M}_{4}$ selects the output state~$s_t'$ according to
    \begin{equation}\label{eq:mech2_pmf}
        \prob{\mathcal{M}_4(s_t) = s'_t} = \frac{\exp\left(-\frac{\epsilon G(s_t', s_t)}{2\rho}\right)}{\sum_{y\in N(s_{t-1}')}\exp\left(-\frac{\epsilon G(y, s_t)}{2\rho}\right)},
    \end{equation}
    where~$G$ is from~\eqref{eq:G}. 
    \hfill $\triangle$
\end{mechanism}
In Appendix~\ref{prf:dp_bo_adj}, we show that Mechanism~\ref{mech:2} is~$\epsilon$-differentially private as a means of showing that Mechanism~\ref{mech:PF} is~$\epsilon$-differentially private.

Next, we show how~\eqref{eq:pi1} follows from~\eqref{eq:mech2_pmf}. Consider the case with one best next candidate state (i.e., the candidate next state with the shortest shortest-path distance to the true state~$s_{t+1}$), and let it's shortest-path distance be~$G(s_t^{\prime, *}, s_{t+1})$. Additionally, let every other non-best candidate next state have the utility~$G(s_t^{\prime, \dagger}, s_{t+1})$. Then, the probability that the best state is selected is 
\begin{align}
        \prob{\mathcal{M}_4(s_{t+1}) = s_t^{\prime, *}} &= \frac{\exp\left(-\frac{\epsilon G(s_t^{\prime, *}, s_{t+1})}{2\rho}\right)}{\sum_{y\in N(s_{t-1}')}\exp\left(-\frac{\epsilon G(y, s_t)}{2\rho}\right)} \\&= 
        \frac{\exp\left(-\frac{\epsilon G(s_t^{\prime, *}, s_{t+1})}{2\rho}\right)}{\exp\left(-\frac{\epsilon G(s_t^{\prime, *}, s_{t+1})}{2\rho}\right) +(| N(s_{t-1}')|-1)\exp\left(-\frac{\epsilon G(s_t^{\prime, \dagger}, s_{t+1})}{2\rho}\right)}.\label{eq:divide_top_bottom}
    \end{align}
Dividing the numerator and denominator of~\eqref{eq:divide_top_bottom} by~$\exp\left(-\frac{\epsilon G(s_t^{\prime, *}, s_{t+1})}{2\rho}\right)$ gives
\begin{equation}
    \frac{\exp\left(-\frac{\epsilon G(s_t^{\prime, *}, s_{t+1})}{2\rho}\right)}{\exp\left(-\frac{\epsilon G(s_t^{\prime, *}, s_{t+1})}{2\rho}\right) +(| N(s_{t-1}')|-1)\exp\left(-\frac{\epsilon G(s_t^{\prime, \dagger}, s_{t+1})}{2\rho}\right)} = \frac{1}{1+(|N(s_t')|-1)\exp\left(-\frac{\epsilon}{2\rho}\gamma_{s_t'}\right)},\label{eq:this_is_pi}
\end{equation}
where~$\gamma$ is from~\eqref{eq:gamma}. The expression in~\eqref{eq:this_is_pi} matches that of~$\pi(s_{t}^{\prime, *}\mid s_{t+1},s_t')$ from~\eqref{eq:pi1}.
\section{Omitted Proofs}\label{app:proofs}
\subsection{Proof of Proposition~\ref{prop:shortest_path}}
    From the Markov property, the probability of a path~$w = s_0\cdots s_n$ where~$s_0 = u$ is
    \begin{equation}
        \prob{w} = \prod_{i=1}^n P_{s_{i-1} s_i}.
    \end{equation}
    The most likely path between~$u$ and~$v$ is then any
    \begin{equation}\label{eq:w_star}
        w^* \in \argmax_{w\in\mathcal{V}(u, v)}  \prod_{i=1}^n P_{s_{i-1} s_i},
    \end{equation}
    where~$\mathcal{V}(u, v)$ is the set of all trajectories of any length
 between states~$u$ and~$v$. 
    The~$\log$ function is strictly monotonically increasing, and we transform~\eqref{eq:w_star} by maximizing over the log probability via
    \begin{equation}
        w^* \in \argmax_{w\in\mathcal{V}(u, v)}  \log\left(\prod_{i=1}^n P_{s_{i-1} s_i}\right) = \argmax_{w\in\mathcal{V}(u, v)}\sum_{i=1}^n \log(P_{s_{i-1} s_i}).
    \end{equation}
    Since finding the maximum of a function is equivalent to finding the minimum of its negative, we have
    \begin{equation}
         w^* \in \argmin_{w\in\mathcal{V}(u, v)}\sum_{i=1}^n -\log(P_{s_{i-1} s_i}) = \argmin_{w\in\mathcal{V}(u, v)}\sum_{i=1}^n W_{s_{i-1} s_i}.
    \end{equation}
    We see that~$\min_{w\in\mathcal{V}(u, v)}\sum_{i=1}^n W_{s_{i-1} s_i} = G(u, v)$, and thus the most likely path from~$u$ to~$v$, namely~$w^*$, yields the shortest path distance from~$u$ to~$v$. \qed

    \subsection{Proof of Lemma~\ref{lem:new_sens}}
    Let~$y = y_0\cdots y_n\in\mathcal{S}^{n+1}$. Substituting the utility function into the definition of the sensitivity, for adjacent~$w, v$ we have that
    \begin{equation}
        \max_{y\in\mathcal{S}^{n+1}}\max_{\substack{w, v\in\mathcal{S}^{n+1}\\\Adjacentr{w}{v}}}|u(w,y) - u(v, y)| = \max_{y\in\mathcal{S}^{n+1}}\max_{\substack{w, v\in\mathcal{S}^{n+1}\\\Adjacentr{w}{v}}}\left|\sum_{t=0}^n G(y_t, \tilde{s}_t) - G(y_t, s_t)\right|.
    \end{equation}
    From the triangle inequality we find
    \begin{equation}
        \max_{y\in\mathcal{S}^{n+1}}\max_{\substack{w, v\in\mathcal{S}^{n+1}\\\Adjacentr{w}{v}}}\left|\sum_{t=0}^n G(y_t, \tilde{s}_t) - G(y_t, s_t)\right| \leq \max_{y\in\mathcal{S}^{n+1}}\max_{\substack{w, v\in\mathcal{S}^{n+1}\\\Adjacentr{w}{v}}}\sum_{t=0}^n |G(y_t, \tilde{s}_t) - G(y_t, s_t)|.
    \end{equation}
    An application of the triangle inequality for directed graphs 
    in Definition~\ref{def:quasi}
    then yields
    \begin{align}
        \max_{y\in\mathcal{S}^{n+1}}\max_{\substack{w, v\in\mathcal{S}^{n+1}\\\Adjacentr{w}{v}}}\sum_{t=0}^n |G(y_t, \tilde{s}_t) - G(y_t, s_t)|&\leq\max_{\substack{w, v\in\mathcal{S}^{n+1}\\\Adjacentr{w}{v}}}\sum_{t=0}^n G(s_t, \tilde{s}_t) \\&\leq \max_{\substack{w, v\in\mathcal{S}^{n+1}\\\Adjacentr{w}{v}}}\sum_{t=0}^n G_{sym}(s_t, \tilde{s}_t),
    \end{align}
    where~$G_{sym}$ is from Definition~\ref{def:traj_adjacency}. 
    Substituting in the adjacency relation from Definition~\ref{def:traj_adjacency} gives
    \begin{equation}
        \max_{y\in\mathcal{S}^{n+1}}\max_{\substack{w, v\in\mathcal{S}^{n+1}\\\Adjacentr{w}{v}}}|u(w, y) - u(v, y)|\leq \max_{\substack{w, v\in\mathcal{S}^{n+1}\\\Adjacentr{w}{v}}}\sum_{t=0}^n G_{sym}(s_t, \tilde{s}_t)\leq \rho,
    \end{equation}
    which completes the proof.\qed

    \subsection{Proof of Theorem~\ref{thm:dp_bo_adj}}\label{prf:dp_bo_adj}
    To show that Mechanism~\ref{mech:PF} is $\epsilon$-differentially private, we first prove that Mechanism~\ref{mech:2} from Appendix~\ref{ap:exp_mechanism} is $\epsilon$-differentially private, and the fact that Mechanism~\ref{mech:PF} is $\epsilon$-differentially private follows from its stochastic dominance of Mechanism~\ref{mech:2}. For adjacent sensitive input words~$x = x_0\cdots x_n$ and~$y = y_0\cdots y_n$ and a private output word~$z' = z_0\cdots z_n$ we have
    \begin{equation}\label{eq:big_prod2}
        \frac{\prob{z' \mid x}}{\prob{z' \mid y}} =\frac{\prob{z_0\mid x_0}}{\prob{z_0\mid y_0}} \prod_{t=1}^n \frac{\prob{z_t \mid x_t, z_{t-1}}}{\prob{z_t \mid y_t, z_{t-1}}}.
    \end{equation}
   In Mechanism~\ref{mech:PF} we have~$z_0\sim p_0$ independent of~$x$ and~$y$. Thus,~\eqref{eq:big_prod2} becomes
   \begin{equation} \label{eq:t1bigprod_reduced}
        \frac{\prob{z' \mid x}}{\prob{z' \mid y}} = \prod_{t=1}^n \frac{\prob{z_t \mid x_t, z_{t-1}}}{\prob{z_t \mid y_t, z_{t-1}}}.
    \end{equation}
   Substituting the PMF from Mechanism~\ref{mech:2} gives
    \begin{equation}\label{eq:use:left_size}
        \frac{\prob{z_t \mid x_t, z_{t-1}}}{\prob{z_t \mid y_t, z_{t-1}}} = \exp\left(\frac{\epsilon}{2\rho} \big(G(z_t, y_t) -  G(z_t, x_t)\big) \right)\left(\frac{\sum_{s\in N(z_{t-1})}\exp\left(-\frac{\epsilon}{2\rho}G(s, y_t)\right)}{\sum_{s\in N(z_{t-1})}\exp\left(-\frac{\epsilon}{2\rho}G(s, x_t)\right)}\right).
    \end{equation}
    From the triangle inequality for directed graphs in Definition~\ref{def:quasi}, we have
    \begin{equation}\label{eq:use_right_side}
        \exp\left(\frac{\epsilon}{2\rho} \big(G(z_t, y_t) -  G(z_t, x_t)\big) \right)\leq \exp\left(\frac{\epsilon}{2\rho} G_{sym}(x_t, y_t)\right),
    \end{equation}
    where~$G_{sym}$ is from Definition~\ref{def:traj_adjacency}. 
    Next we bound the rightmost term in~\eqref{eq:use:left_size}. Again using the triangle inequality for directed graphs gives
    \begin{equation}
        -G(s, y_t) \leq -G(s, x_t) + G_{sym}(x_t, y_t).
    \end{equation}
    Multiplying by~$\frac{\epsilon}{2\rho}$, taking the exponential, and taking the sum over~$s \in N(z_{t-1})$ gives
    \begin{equation}\label{eq:left_term_now}
        \sum_{s\in N(z_{t-1})} \exp\left(-\frac{\epsilon}{2\rho} G(s, y_t)\right) \leq \sum_{s\in N(z_{t-1})} \exp\left(-\frac{\epsilon}{2\rho} G(s, x_t)\right)\exp\left(\frac{\epsilon}{2\rho} G_{sym}(x_t, y_t)\right).
    \end{equation}
    Substituting~\eqref{eq:use_right_side} and~\eqref{eq:left_term_now} into~\eqref{eq:use:left_size} gives
    \begin{align}
        \frac{\prob{z_t \mid x_t, z_{t-1}}}{\prob{z_t \mid y_t, z_{t-1}}} &\leq \exp\left(\frac{\epsilon}{2\rho} \big(G(z_t, y_t) -  G(z_t, x_t)\big) \right)\left(\frac{\sum_{s\in N(z_{t-1})}\exp\left(-\frac{\epsilon}{2\rho}G(s, y_t)\right)}{\sum_{s\in N(z_{t-1})}\exp\left(-\frac{\epsilon}{2\rho}G(s, x_t)\right)}\right)\\
        &\leq \exp\left(\frac{\epsilon}{2\rho} \big(G(z_t, y_t) -  G(z_t, x_t)\big) \right)\left(\frac{\sum_{s\in N(z_{t-1})} \exp\left(-\frac{\epsilon}{2\rho} G(s, x_t)\right)\exp\left(\frac{\epsilon}{2\rho} G_{sym}(x_t, y_t)\right)}{\sum_{s\in N(z_{t-1})}\exp\left(-\frac{\epsilon}{2\rho}G(s, x_t)\right)}\right)\\
        &\leq  \exp\left(\frac{\epsilon}{2\rho} G_{sym}(x_t, y_t)\right)\exp\left(\frac{\epsilon}{2\rho} G_{sym}(x_t, y_t)\right) \\
        &= \exp\left(\frac{\epsilon}{\rho} G_{sym}(x_t, y_t)\right),\label{eq:sub_into_big_prod}
    \end{align}
    where the third inequality uses the triangle inequality for directed graphs from Definition~\ref{def:quasi}. 
    Substituting~\eqref{eq:sub_into_big_prod} into~\eqref{eq:t1bigprod_reduced} gives
    \begin{equation}
        \prod_{t=1}^n \frac{\prob{z_t \mid x_t, z_{t-1}}}{\prob{z'_t \mid y_t, z_{t-1}}} \leq \prod_{t=1}^n \exp\left(\frac{\epsilon}{\rho} G_{sym}(x_t, y_t)\right) = \exp\left(\frac{\epsilon}{\rho}\sum_{t=1}^n G_{sym}(x_t, y_t)\right).
    \end{equation}
    Substituting in the definition of adjacency gives 
    \begin{equation}
        \exp\left(\frac{\epsilon}{\rho} \sum_{t=1}^nG_{sym}(x_t, y_t)\right) \leq \exp(\epsilon),
    \end{equation}
    which completes the proof.\qed

\subsection{Proof of Lemma~\ref{thm:old_adj_safe_set}}

    To satisfy the Foster-Lyapunov condition in~\eqref{eq:neg_drift}, we see that~$G(s_{t+1}', s_{t+1})$ must satisfy
    \begin{equation} \label{eq:main_fl_condition}
        \mathbb{E}_{s_{t+1}, s_{t+1}'}\left[G(s_{t+1}', s_{t+1}) - G(s_t', s_t)\mid (s_t, s_t')\right]<0.
    \end{equation}
    We then add and subtract~$G(s_t', s_{t+1})$ to find
    \begin{multline}
        \mathbb{E}_{s_{t+1}, s_{t+1}'}\left[G(s_{t+1}', s_{t+1}) - G(s_t', s_t)\mid (s_t, s_t')\right] \\= \mathbb{E}_{s_{t+1}, s'_{t+1}}\left[(G(s_{t+1}', s_{t+1}) - G(s_t', s_{t+1})) +(G(s_t', s_{t+1})- G(s_t', s_t))\mid (s_t, s_t')\right].
    \end{multline}
    We next separate the expectations to yield
    \begin{multline}\label{eq:bound_left_right_3}
        \mathbb{E}_{s_{t+1}, s'_{t+1}}\left[(G(s_{t+1}', s_{t+1}) - G(s_t', s_{t+1})) +(G(s_t', s_{t+1})- G(s_t', s_t))\mid (s_t, s_t')\right] \\=\mathbb{E}_{s_{t+1}}\left[\mathbb{E}_{ s'_{t+1}}\left[G(s_{t+1}', s_{t+1}) - G(s_t', s_{t+1}) \mid s_t', s_{t+1}\right]\mid s_t\right] + \mathbb{E}_{s_{t+1}}\left[G(s_t', s_{t+1}) - G(s_t', s_t)\mid (s_t, s_t')\right].
    \end{multline}
    We now bound the first and second terms on the right-hand side of~\eqref{eq:bound_left_right_3} individually. Beginning with the second term on the right-hand side of~\eqref{eq:bound_left_right_3}, from the triangle inequality for directed graphs in Definition~\ref{def:quasi} we have
    \begin{equation}
        G(s'_t, s_{t+1}) -G(s_t', s_t) \leq G(s_t, s_{t+1}).
    \end{equation}
    Taking the expectation gives
    \begin{equation}\label{eq:compute_this_ex}
        \mathbb{E}_{s_{t+1}}\left[G(s'_t, s_{t+1}) -G(s_t', s_t) \mid (s_t, s_t')\right] \leq \mathbb{E}_{s_{t+1}}\left[G(s_t, s_{t+1})\mid s_t\right].
    \end{equation}
    We can then compute the expectation in~\eqref{eq:compute_this_ex} as
    \begin{equation}\label{eq:exp_1}
        \mathbb{E}_{s_{t+1}}\left[G(s'_t, s_{t+1}) -G(s_t', s_t) \mid (s_t, s_t')\right]\leq \mathbb{E}_{s_{t+1}}\left[G(s_t, s_{t+1})\mid s_t\right] =  \sum_{j\in N(s_t)} -P_{s_tj}\log(P_{s_tj}) = \delta_{s_{t}}.
    \end{equation}
    Next, we bound the first term on the right-hand side of~\eqref{eq:bound_left_right_3}. Suppose that at time~$t$, there is at least one possible next state along a shortest path     
    from the private state at time~$t$ to the sensitive state at time~$t+1$. There are at most~$|N(s_t')|-1$ possible next states not on a shortest path. Let~$y_{good}$ be a state along a shortest path from~$s_t'$ to the sensitive state~$s_{t+1}$, and let~$y_{bad}$ be a state not on a shortest path from~$s_t'$ to~$s_{t+1}$
    if such a non-shortest path exists\footnote{We state this proof in the case where~$y_{bad}$ exists, though it is not guaranteed to exist. However, the resulting bound holds whether~$y_{bad}$ exists or not.}. 
    If the private state trajectory 
    transitions from~$s_t'$ to~$y_{good}$, then the distance from~$y_{good}$ to~$s_{t+1}$ is
    less than the distance from~$s_{t}'$ to~$s_{t+1}$ by an amount equal to
    the edge weight~$W_{s_t', y_{good}}$. Thus, in this case,
\begin{equation}\label{eq:wgood_sub}
    G(y_{good}, s_{t+1}) - G(s_t', s_{t+1}) = (G(s_{t}', s_{t+1})-W_{s_t', y_{good}}) - G(s_{t}', s_{t+1})  =- W_{s_t', y_{good}},
\end{equation}
where~$W_{s_t', y_{good}}$ is the edge weight from~$s_t'$ 
to~$y_{good}$, i.e.,
\begin{equation}
     W_{s_t', y_{good}} = \min_{w\in N(s_t')} G(w, s_{t+1}).
\end{equation}
If the private state trajectory transitions from~$s_t'$ to~$y_{bad}$, then 
the distance from~$y_{bad}$ to~$s_{t+1}$ can be greater than the distance from~$s_t'$
to~$s_{t+1}$ by up to 
$D_{s_t' ,max} = \max_{v\in N(s_t')} G(v, s_{t+1})$.
From the triangle inequality for directed graphs in Definition~\ref{def:quasi}, we have
\begin{equation}
    G(y_{bad}, s_{t+1}) \leq G(y_{bad}, s_{t}') + G(s_t', s_{t+1}), 
\end{equation}
and thus rearranging gives
\begin{equation}\label{eq:dmax_sub}
    G(y_{bad}, s_{t+1}) - G(s_t', s_{t+1}) \leq G(s_t', s_{t+1})+D_{s_t' ,max} - G(s_t', s_{t+1}) = D_{s_t' ,max},
\end{equation}
where~$D_{s_t', max}$ is from~\eqref{eq:D}.
From applying the law of total expectation to the first term on the right-hand side of~\eqref{eq:bound_left_right_3} and using~\eqref{eq:wgood_sub} and~\eqref{eq:dmax_sub}, we then have
\begin{equation}\label{eq:exp_2}
    \mathbb{E}_{ s'_{t+1}}\left[G(s'_{t+1}, s_{t+1}) - G(s'_t, s_{t+1}) \mid s_t', s_{t+1}\right] \leq -P_{good}  W_{s_t', y_{good}} + (1-P_{good})D_{s_t' ,max},
\end{equation}
where~$P_{good}$ is the probability that~$s_{t+1}' = y_{good}$ for any~$y_{good}$, i.e., any candidate next state along a shortest path from~$s_t'$ to~$s_{t+1}$. 
Let~$s_{t}^{\prime, *} \in \arg\min_{j\in N(s_{t}')} G(j, s_{t+1})$, 
and let~$s_{t}^{\prime,\dagger} \in S_t^{\prime, \dagger}$,  where~$S_t^{\prime, \dagger}$ is from~\eqref{eq:bigS}.
For any~$y_{good}$, 
the lowest probability of having~$s_{t+1}' = y_{good}$ 
under Mechanism~\ref{mech:2} (from Appendix~\ref{ap:exp_mechanism}) occurs when there exists only one~$y_{good}$ and all other candidate next states have equal utility. In this scenario, we find
\begin{equation}\label{eq:P_good_M1}
    P_{good} = \pi(s_{t}^{\prime, *}\mid s_{t+1},s_t') = \frac{1}{1+(|N(s_t')|-1)\exp\left(-\frac{\epsilon}{2\rho}\gamma_{s_t'}\right)},
\end{equation}
where 
\begin{equation}
    \gamma_{s_{t+1},s_t'} = G(s_t^{\prime, \dagger}, s_{t+1}) - G(s_t^{\prime, *}, s_{t+1}).
\end{equation}

Substituting~\eqref{eq:exp_2} and~\eqref{eq:P_good_M1} into the first term on the right-hand side of~\eqref{eq:bound_left_right_3} and substituting~\eqref{eq:exp_1} into the second term on the right hand side of~\eqref{eq:bound_left_right_3} gives
\begin{multline}\label{eq:bound_by_local1}
    \mathbb{E}_{s_{t+1}, s'_{t+1}}\left[G(s_{t+1}', s_{t+1}) - G(s_t', s_t)\mid (s_t, s_t')\right] \\\leq \mathbb{E}_{s_{t+1}}\left[-\pi(s_{t}^{\prime, *}\mid s_{t+1},s_t') W_{s_t', y_{good}} + (1-\pi(s_{t}^{\prime, *}\mid s_{t+1},s_t'))D_{s_t' ,max}\mid s_t \right]+\delta_{s_t}.
\end{multline}
A sufficient condition for the right-hand side of~\eqref{eq:bound_by_local1} to be strictly negative is for
\begin{equation}
    -\pi(s_{t}^{\prime, *}\mid s_{t+1},s_t') W_{s_t', y_{good}} + (1-\pi(s_{t}^{\prime, *}\mid s_{t+1},s_t'))D_{s_t' ,max}+\delta_{s_t} <0
\end{equation}
to hold for all~$s_{t+1}\in N(s_t)$. Rearranging to solve for~$W_{s_t', y_{good}}$ yields the equivalent sufficient condition 
\begin{equation}\label{eq:old_stab_req}
    W_{s_t', y_{good}} > \left(\frac{1 - \pi(s_{t}^{\prime, *} \mid s_{t+1},s_t')}{\pi(s_{t}^{\prime, *} \mid s_{t+1},s_t')}\right) D_{s_t',max} + \frac{\delta_{s_{t}}}{\pi(s_{t}^{\prime, *} \mid s_{t+1},s_t')}.
\end{equation}
Thus, given the joint state~$(s_t, s_t')$, from~\eqref{eq:main_fl_condition} we find that
\begin{equation}
     \E{V(Z_{t+1}) - V(Z_t)\mid Z_t} < 0
\end{equation}
holds if~\eqref{eq:old_stab_req} holds for all~$s_{t+1}\in N(s_t)$, which completes the proof.\qed 

\subsection{Proof of Theorem~\ref{cor:hajek_bound}}
For the concentration bound, we will invoke~\cite[Theorem 2.3]{hajek1982hitting}. To do so, we require 
several constants. First, we bound how much~$V(Z_t)$ may change in
one timestep within~$\mathcal{Z}_{attract}(\epsilon)$ via
\begin{equation}
    |V(Z_{t+1})-V(Z_{t})| \leq \max_{(s_t, s_t')\in\mathcal{Z}_{attract}(\epsilon)} \left(\max_{w\in N(s_t'), y\in N(s_t)} G(w, y) - G(s_t', s_t)\right) = D_{\mathcal{Z}_{attract}(\epsilon)}.
\end{equation}
Next, we require the smallest expected drift 
in one timestep within~$\mathcal{Z}_{attract}(\epsilon)$, which is
\begin{equation}
    \gamma_{\mathcal{Z}_{attract}(\epsilon)} = \min_{(s_t, s_t')\in \mathcal{Z}_{attract}(\epsilon)} \left| \sum_{y\in N(s_t')} \pi(y \mid s_t') (G(y, s_t) - G(s_t',s_t))\right|.
\end{equation}
Finally, we require the largest value of~$V(Z_t)$ outside of~$\mathcal{Z}_{attract}(\epsilon)$, which is 
\begin{equation}
B_{\epsilon} = \max_{(s_t, s_t')\in \mathcal{Z}_{attract}^c(\epsilon)} G(s_t', s_t),
\end{equation}
where~$\mathcal{Z}_{attract}^c(\epsilon) = \mathcal{Z}\setminus \mathcal{Z}_{attract}(\epsilon)$ and~$\mathcal{Z}$ is from the joint Markov chain in~\eqref{eq:joint_chain}.
Then, from Equation (2.8) in~\cite[Theorem 2.3]{hajek1982hitting} we have 
\begin{equation}\label{eq:take_exp}
    \prob{V(Z_t)>v\mid Z_0} \leq \zeta^te^{\alpha (V(Z_0)-v)}+ \frac{1-\zeta^t}{1-\zeta}D_He^{\alpha(B_{\epsilon}-v)}, 
\end{equation}
where
\begin{align}
    \alpha &= \frac{2\gamma_{\mathcal{Z}_{attract}(\epsilon)}}{D_{\mathcal{Z}_{attract}(\epsilon)}},\\
    \zeta &= 1 - \frac{\gamma_{\mathcal{Z}_{attract}(\epsilon)}^2}{D_{\mathcal{Z}_{attract}}(\epsilon)},\\
    D_H &= \exp(\alpha D_{\mathcal{Z}_{attract}(\epsilon)})
\end{align}
are from Lemma 2.1 and Conditions D1 and D2 in~\cite{hajek1982hitting}. Taking the expectation over the randomness in~$Z_0$ in~\eqref{eq:take_exp} to remove the conditioning on~$Z_0$ gives
\begin{equation}\label{eq:add_subtract}
    \prob{V(Z_t)>v} \leq \zeta^t\E{e^{\alpha (V(Z_0)-v)}}+ \frac{1-\zeta^t}{1-\zeta}D_H\exp(-\alpha (v-B_{\epsilon})).
\end{equation}
Adding and subtracting~$\alpha B_{\epsilon}$ in the exponent in the first term of~\eqref{eq:add_subtract} gives
\begin{align}
    \prob{V(Z_t)>v} &\leq \zeta^t\E{e^{\alpha (V(Z_0)-B_{\epsilon} +B_{\epsilon}-v)}}+ \frac{1-\zeta^t}{1-\zeta}D_H\exp(-\alpha (v-B_{\epsilon}))\\
    &= \zeta^t\E{e^{\alpha (V(Z_0)-B_{\epsilon})}} e^{-\alpha(v-B_{\epsilon})}+ \frac{1-\zeta^t}{1-\zeta}D_H\exp(-\alpha (v-B_{\epsilon}))\\
    & = \left(\zeta^t\E{e^{\alpha (V(Z_0)-B_{\epsilon})}} + \frac{1-\zeta^t}{1-\zeta}D_H\right)\exp(-\alpha (v-B_{\epsilon})).
\end{align}
Then, we compute the expectation~$\E{e^{\alpha (V(Z_0)-B_{\epsilon})}}$ as
\begin{equation}
    \E{e^{\alpha (V(Z_0)-B_{\epsilon})}} = e^{-\alpha B_{\epsilon}}\E{e^{\alpha V(Z_0)}} = e^{-\alpha B_{\epsilon}}\sum_{x,y\in\mathcal{S}} p_0(x)p_0(y) e^{\alpha G(x,y)},
\end{equation}
where~$p_0$ is the initial distribution of~$M$, which gives the bound
\begin{equation}
    \prob{V(Z_t)>v} \leq \left(\zeta^te^{-\alpha B_{\epsilon}}\sum_{x\in\mathcal{S}}\sum_{y\in\mathcal{S}} p_0(x)p_0(y) e^{\alpha G(x,y)}+ \frac{1-\zeta^t}{1-\zeta}e^{\alpha D_{\mathcal{Z}_{attract}(\epsilon)}} \right)\exp(-\alpha (v-B_{\epsilon})).
\end{equation}
Letting 
\begin{equation}
    C_t = \zeta^te^{-\alpha B_{\epsilon}}\sum_{x\in\mathcal{S}}\sum_{y\in\mathcal{S}} p_0(x)p_0(y) e^{\alpha G(x,y)}+ \frac{1-\zeta^t}{1-\zeta}e^{\alpha D_{\mathcal{Z}_{attract}(\epsilon)}}, 
\end{equation}
we have
\begin{equation}
\prob{V(Z_t)>v} \leq C_{t} \exp\left(-\alpha (v-B_{\epsilon})\right),
\end{equation}
which completes the proof.\qed


\subsection{Proof of Theorem~\ref{thm:avg_entropy}}
Taking the expectation of~\eqref{eq:entropy_y} 
with~$Y = s'$ gives 
\begin{equation}\label{eq:EH}
    \E{H_n(s')} = \frac{1}{n-1}\sum_{t=1}^n \E{G(s_{t-1}', s_t')}.
\end{equation}
From the triangle inequality for directed graphs in Definition~\ref{def:quasi}, we have
\begin{equation}\label{eq:subbed_in_H}
    \E{G(s_{t-1}', s_t')} \leq \E{G(s'_{t-1}, s_{t-1})} + \E{G(s_{t-1}, s_t)} + \E{G(s_t, s'_{t})}.
\end{equation}
Because~$p_0 = \mu$, from the definition of the expectation we have~$\E{G(s_{t-1}, s_t)} = H(\mathcal{S})$, where~$H(\mathcal{S})$ is from~\eqref{eq:empirical_entropy}. From the definition of the Foster-Lyapunov function~$V$ in~\eqref{eq:V}, 
we have~$\E{G(s'_{t-1}, s_{t-1})} = \E{V(Z_{t-1})}$. Then the bound in~\eqref{eq:subbed_in_H} implies that 
\begin{equation} \label{eq:three_term_bound} 
    \E{G(s_{t-1}', s_t')} \leq \E{V(Z_{t-1})} + H(\mathcal{S}) + \E{G(s_t, s'_{t})}.
\end{equation}
The term~$\E{G(s_t, s'_{t})}$ is just the reverse distance of~$\E{G(s'_{t}, s_t)} = \E{V(Z_t)}$. To bound~$\E{G(s_t, s'_{t})}$ in terms of the Foster-Lyapunov function, we use the constant~$K_{asym}$ from~\eqref{eq:asym_coeff} to write
\begin{equation}
    \E{G(s_t, s'_{t})}\leq K_{asym}\E{V(Z_t)}.
\end{equation}
Then~\eqref{eq:three_term_bound} implies that 
\begin{equation}\label{eq:sub_more_here}
    \E{G(s_{t-1}', s_t')} \leq \E{V(Z_{t-1})} + H(\mathcal{S}) +  K_{asym}\E{V(Z_t)}.
\end{equation}
By virtue of the Markov chain~$M$ being ergodic, the Markov chain~$M_{Z}$ is also ergodic, and thus the state distribution of~$M_Z$ at time~$t$, denoted~$p^{Z}(t)$, converges to a stationary distribution~$\mu_{Z}$. Let~$Z_{\infty}\sim \mu_{Z}$ denote a joint state drawn from this distribution, and define the expected steady-state shortest-path distance from~$s_t'$ to~$s_t$ as~$V_{\infty} = \E{V(Z_{\infty})}$, where the expectation is taken over the randomness induced by the stationary distribution of~$M_{Z}$.
An additional property of ergodicity is that the PMF over the states~$p^{Z}(t)$ converges to the stationary distribution of the Markov chain~$M_{Z}$ geometrically fast~\cite[Theorem 4.3.5]{gallager2013stochastic}. As a result, the expectation of bounded, continuous functions of~$Z_t$ 
converge to a steady-state value geometrically fast, including the shortest-path distance. 
For each time~$t$, from the definition of geometric convergence 
we have
\begin{equation}\label{eq:sub_Vbar_here}
    \E{V(Z_t)} \leq V_{\infty}+\psi_2^t\E{V(Z_0)},
\end{equation}
where the term~$\psi_2^t\E{V(Z_0)}$ is the expected transient error at time~$t$ and~$\psi_2$ is the second largest eigenvalue modulus (SLEM) of the joint transition matrix~$P^{Z}_{\epsilon}$ from~\eqref{eq:joint_chain}.
We can then compute~$\E{V(Z_0)}$ directly as
\begin{equation}\label{eq:Vbar}
    \E{V(Z_0)} = \sum_{x\in\mathcal{S}}\sum_{y\in\mathcal{S}} p_0(x)p_0(y) G(x,y) = \bar{V}_0.
\end{equation}
 Substituting~\eqref{eq:Vbar} and~\eqref{eq:sub_Vbar_here} into~\eqref{eq:sub_more_here} gives
\begin{equation}\label{eq:Vbar_subbed}
    \E{V(Z_{t-1})} + H(\mathcal{S}) +  K_{asym}\E{V(Z_t)} \leq H(\mathcal{S}) +\left(\frac{1}{\psi_2}+K_{asym}\right)(V_{\infty}+\psi_2^t\bar{V}_0).
\end{equation}
Then using~\eqref{eq:EH},~\eqref{eq:sub_more_here}, and~\eqref{eq:Vbar_subbed}, we find
\begin{multline}
        \E{H_n(s')} \leq \frac{1}{n-1}\sum_{t=1}^n H(\mathcal{S})+\left(\frac{1}{\psi_2}+K_{asym}\right)V_{\infty}+\left(\frac{1}{\psi_2}+K_{asym}\right)(\psi_2^t\bar{V}_0) \\\leq \frac{n}{n-1}\left(H(\mathcal{S})+\left(\frac{1}{\psi_2}+K_{asym}\right)V_{\infty}\right)+\frac{\left(1+\psi_2 K_{asym}\right)\bar{V}_0}{(n-1)(1-\psi_2)}\label{eq:sub_vinf_here_to_finish},
\end{multline}
where the right-most term in~\eqref{eq:sub_vinf_here_to_finish} follows from the 
formula for a convergent infinite geometric series in~$\psi_2$. 

Now we bound~$V_{\infty}$. From the definition of the expectation~$V_{\infty} = \mathbb{E}[V(Z_{\infty})]$ 
we have
\begin{equation}\label{eq:exp_V}
    V_{\infty} = \int_{0}^{\infty} \prob{V(Z_{\infty})>v} dv.
\end{equation}
We split the integral in~\eqref{eq:exp_V} into two regions: one where~$v\leq B_{\epsilon}$ and the other where~$v>B_{\epsilon}$. For~$v\leq B_{\epsilon}$, we note that~$V(Z_{\infty})\in[0, D_{max}]$, where~$D_{max}$ is the diameter of~$\mathcal{G}$. Thus, from Hoeffdings inequality, we have 
\begin{equation}\label{eq:int_sub1}
    \prob{V(Z_{\infty})>v} \leq \exp\left(-\frac{2v^2}{D_{max}^2}\right).
\end{equation}
For~$v>B_{\epsilon}$, we use~\eqref{eq:hajek} in Theorem~\ref{cor:hajek_bound}, i.e.,
\begin{equation}\label{eq:int_sub2}
    \prob{V(Z_{\infty})>v}\leq  C_{\infty} \exp\left(-\frac{2\gamma_{\mathcal{Z}_{attract}(\epsilon)}}{D_{\mathcal{Z}_{attract}(\epsilon)}}(v-B_{\epsilon})\right).
\end{equation}
Substituting~\eqref{eq:int_sub1} and~\eqref{eq:int_sub2} into~\eqref{eq:exp_V} yields
\begin{align}
    V_{\infty} &\leq \int_{0}^{B_{\epsilon}} \exp\left(-\frac{2v^2}{D_{max}^2}\right) dv +  C_{\infty}\int_{B_{\epsilon}}^{\infty} \exp\left(-\frac{2\gamma_{\mathcal{Z}_{attract}(\epsilon)}}{D_{\mathcal{Z}_{attract}(\epsilon)}}(v-B_{\epsilon})\right) dv \\&= \frac{\sqrt{\pi}}{2\sqrt{\frac{2}{D_{max}}}}\erf\left(B_{\epsilon}\sqrt{\frac{2}{D_{max}}}\right)  + C_{\infty}\frac{D_{\mathcal{Z}_{attract}(\epsilon)}}{2\gamma_{\mathcal{Z}_{attract}(\epsilon)}},
\end{align}
which we substitute into~\eqref{eq:sub_vinf_here_to_finish} to obtain
\begin{multline}
         \E{H_n(s')}\leq\\ \frac{n}{n-1}\left(H(\mathcal{S})+\left(\frac{1}{\psi_2}+K_{asym}\right)\left(\frac{\sqrt{\pi}}{2\sqrt{\frac{2}{D_{max}}}}\erf\left(B_{\epsilon}\sqrt{\frac{2}{D_{max}}}\right)+C_{\infty}\frac{D_{\mathcal{Z}_{attract}(\epsilon)}}{2\gamma_{\mathcal{Z}_{attract}(\epsilon)}}\right)\right) +\frac{\left(1+\psi_2K_{asym}\right)\bar{V}_0}{(n-1)(1-\psi_2)},
     \end{multline}
which completes the proof.\qed

\subsection{Proof of Theorem~\ref{thm:cont_entropy}}
 We begin by bounding how much the value of 
    \begin{equation}
        H_n(s') = \frac{1}{n-1}\sum_{t=1}^n G(s'_{t-1}, s'_t) 
    \end{equation}
    may change when changing a single state in the trajectory~$s' = s_0's_1'\cdots s_n'$. Let the state with the change be~$s_t'$ for a fixed~$t \in [n]$. Changing~$s_t'$ may at most change~$G(s'_{t-1}, s_t')$ by~$D_{max}$, where~$D_{max}$ is the diameter of~$\mathcal{G}$ from Definition~\ref{def:MC-graph}. Similarly, changing~$s_t'$ may also at most change~$G(s'_{t}, s_{t+1}')$ by~$D_{max}$. Then, changing~$s_t'$ changes~$H_n(s')$ by at most~$\frac{2D_{max}}{n-1}$ for each~$t\in\{1, \ldots, n-1\}$. If~$t \in \{0, n\}$, then~$s_t'$ appears only in one value of~$G$, and therefore~$H_n(s')$ changes by at most~$\frac{D_{max}}{n-1}$ for~$t \in \{0, n\}$. Let~$\tilde{s}'$ be a state trajectory equal to~$s'$ but with a change in a single state. Then 
    \begin{equation}\label{eq:plug_c}
        H_n(s') - H(\tilde{s}') \leq \sum_{i\in \{0, n\}} \frac{D_{max}}{n-1} \ones [s_i' \neq \tilde{s}_i'] +\sum_{i=1}^{n-1} \frac{2D_{max}}{n-1} \ones [s_i' \neq \tilde{s}_i']. 
    \end{equation}
    Let
    \begin{equation}\label{eq:def_c}
         c = \left[\frac{D_{max}}{n-1}, \frac{2D_{max}}{n-1}, \cdots, \frac{2D_{max}}{n-1}, \frac{D_{max}}{n-1}\right].
    \end{equation}
    Then, for~$n\geq 2$
    \begin{equation}\label{eq:norm_c}
        \norm{c}_2^2 = \sum_{i=0}^n c_i^2 = 2\left(\frac{D_{max}}{n-1} \right)^2 + (n-2)\left(\frac{2D_{max}}{n-1} \right)^2 = \frac{(4n-6)D_{max}^2}{(n-1)^2}.
    \end{equation}
    From~\eqref{eq:plug_c} and~\eqref{eq:def_c}, we have
    \begin{equation}\label{eq:paulin}
        H_n(s') - H(\tilde{s}') \leq \sum_{i=0}^n c_i  \ones [s_i' \neq \tilde{s}_i'].
    \end{equation}
    As a result of~\eqref{eq:paulin}, Lemma~\ref{lem:mcdairmid} implies that for any~$\nu\in\mathbb{R}$ 
    we have 
    \begin{equation}
        \prob{H_n(s')-\E{H_n(s')}\geq \nu}\leq \exp\left(\frac{-2\nu^2}{9\norm{c}^2_2t_{mix}} \right),
    \end{equation}
    where~$t_{mix}$ is from~\eqref{eq:mix_time}.
    Substituting the upper bound on~$\E{H_n(s')}$ from Theorem~\ref{thm:avg_entropy} yields
    \begin{equation}\label{eq:sub_c_here}
        \prob{H_n(s')- \frac{n}{n-1}H(\mathcal{S})\geq \phi_{\epsilon} +\nu}\leq \exp\left(\frac{-2\nu^2}{9\norm{c}^2_2t_{mix}} \right).
    \end{equation}   
    Substituting~\eqref{eq:norm_c} into~\eqref{eq:sub_c_here} yields
   \begin{equation}
        \prob{H_n(s') - \frac{n}{n-1}H(\mathcal{S})\geq \phi_{\epsilon} +\nu} \leq \exp\left(-\frac{2\nu^2(n-1)^2}{9(4n-6)D_{max}^2t_{mix}} \right),
    \end{equation}
    which completes the proof.\qed

    \section{Additional Simulation Details}\label{app:sims}
\subsection{Data Pre-Processesing}
\begin{table}[htpb]
\centering
\caption{Empirical 1-Year Transition Matrix for US Obligors.}
\label{tab:us_transition_matrix}
\begin{tabular}{l ccc ccc ccc}
\toprule
& \multicolumn{9}{c}{\textbf{Terminal Rating}} \\
\cmidrule(lr){2-10}
\textbf{Initial Rating} & Aaa & Aa & A & Baa & Ba & B & Caa & Ca/C & Def \\
\midrule
\textbf{Aaa} & 0.919 & 0.069 & 0.011 & 0 & 0.001 & 0 & 0 & 0 & 0 \\
\textbf{Aa}  & 0.012 & 0.893 & 0.088 & 0.005 & 0.002 & 0 & 0 & 0 & 0 \\
\textbf{A}   & 0.001 & 0.023 & 0.920 & 0.049 & 0.006 & 0.002 & 0 & 0 & 0 \\
\textbf{Baa} & 0 & 0.002 & 0.055 & 0.889 & 0.045 & 0.006 & 0.001 & 0 & 0.001 \\
\textbf{Ba}  & 0 & 0.001 & 0.005 & 0.054 & 0.855 & 0.069 & 0.003 & 0 & 0.014 \\
\textbf{B}   & 0 & 0.001 & 0.002 & 0.007 & 0.065 & 0.829 & 0.019 & 0.005 & 0.072 \\
\textbf{Caa} & 0 & 0 & 0 & 0.010 & 0.025 & 0.076 & 0.673 & 0.035 & 0.181 \\
\textbf{Ca/C}& 0 & 0 & 0 & 0 & 0.010 & 0.057 & 0.143 & 0.581 & 0.210 \\
\textbf{Def} & 0 & 0 & 0 & 0 & 0 & 0 & 0 & $10^{-4}$ & $1{-}10^{-4}$ \\
\bottomrule
\end{tabular}
\end{table}
To compute the Markov chain model for the credit rating migration example, we take the transition probabilities from~\cite[Table 3]{nickell2000stability} from the Moody's $1$-year credit rating transition matrix, yielding a Markov chain with~$|\mathcal{S}| = 9$ states. Because the Markov chain in~\cite[Table 3]{nickell2000stability} is \emph{not} irreducible, we include a small probability ($10^{-4}$) of a firm recovering from default, which guarantees the Markov chain is irreducible for our analysis. Previous work on credit migration matrices shows that the default state being ``absorbing'' is unrealistic and harms analysis, which motivates our choice to include this probability~\citep{jafry2004measurement}. The transition probabilities of the resulting Markov chain are shown in Table~\ref{tab:us_transition_matrix}. Table~\ref{tab:credit} shows sample trajectories for~$\epsilon\in\{0.1, 1, 2, 5\}$ for both Mechanism~\ref{mech:PF} and Mechanism~\ref{mech:bo}. 
The sensitive state trajectory is the credit rating of a European country over the past~$10$ years.

\begin{table}[ht]
\centering
\caption{Sample private state trajectories for the credit migration example. Bold entries indicate a deviation from the sensitive state trajectory.}
\label{tab:credit}
\begin{tabular}{ll cccccccccc}
\toprule
$\epsilon$ & \textbf{Mechanism} & $t=0$ & $t=1$ & $t=2$ & $t=3$ & $t=4$ & $t=5$ & $t=6$ & $t=7$ & $t=8$ & $t=9$ \\
\midrule
--- & \textbf{Sensitive} & Caa & Ca/C & Caa & Caa & Caa & Caa & B & Caa & Ca/C & Ca/C \\ 
\midrule
\multirow{2}{*}{$0.1$} 
 & Chen et al. & Caa & \textbf{Baa} & \textbf{Aa} & \textbf{Aaa} & \textbf{Aaa} & \textbf{Aa} & \textbf{Ba} & \textbf{Aa} & \textbf{A} & \textbf{Aaa} \\
 & Mechanism~\ref{mech:PF} & Caa & \textbf{Def} & \textbf{Ca/C} & Caa & \textbf{Ba} & \textbf{Baa} & \textbf{Caa} & \textbf{Ca/C} & \textbf{Caa} & \textbf{Def} \\ 
\addlinespace
\multirow{2}{*}{$1.0$} 
 & Chen et al. & Caa & \textbf{Caa} & \textbf{Def} & \textbf{Ca/C} & Caa & Caa & B & \textbf{Ca/C} & \textbf{B} & \textbf{Baa} \\
 & Mechanism~\ref{mech:PF} & Caa & Ca/C & Caa & Caa & Caa & \textbf{Def} & \textbf{Ca/C} & \textbf{Def} & Ca/C & Ca/C \\ 
\addlinespace
\multirow{2}{*}{$2.0$}   
 & Chen et al. & Caa & \textbf{Caa} & Caa & Caa & Caa & \textbf{Def} & \textbf{Def} & \textbf{Def} & Ca/C & Ca/C \\
 & Mechanism~\ref{mech:PF} & Caa & Ca/C & Caa & \textbf{Ca/C} & Caa & Caa & B & Caa & Ca/C & Ca/C \\ 
\addlinespace
\multirow{2}{*}{$5.0$}   
 & Chen et al. & Caa & Ca/C & Caa & Caa & Caa & Caa & B & Caa & Ca/C & Ca/C \\
 & Mechanism~\ref{mech:PF} & Caa & Ca/C & Caa & Caa & Caa & Caa & B & Caa & Ca/C & Ca/C \\ 
\bottomrule
\end{tabular}
\end{table}

To compute the Markov chain model for the Gainesville traffic example, we divide the roads around the University of Florida into segments, where a road segment is a section of road between two intersections. If a road does not extend past any intersections, it has only a single segment. 
Each of these segments is a state in the Markov chain, and
the resulting Markov chain has~$|\mathcal{S}| = 43$ states. To compute the transition probabilities, we count the number of times drivers transitioned from one road segment to another, and we divide it by the total number of times drivers transitioned away from the first segment. Figure~\ref{fig:gainesville_trajectories}a shows the state space for the Markov chain, and the remainder of Figure~\ref{fig:gainesville_trajectories} shows sample trajectories for varying levels of privacy.

To develop the Markov chain model for the Wikispeedia example, we count the number of times users transitioned from one Wikipedia article to another, and we take the largest strongly connected subgraph of the resulting Markov chain, resulting in~$|\mathcal{S}| = 3,817$. Figure~\ref{fig:wiki_map} shows a visualization of the Wikispeedia dataset, and Figure~\ref{fig:wiki_trajectories} shows samples of private trajectories on this visualization.

\subsection{Additional Simulation Details}
In this subsection, we state the simulation parameters used for~\cite{chen2023differentialsymbolic}, show single-sample results for all three simulation examples in Section~\ref{sec:sims} at~$\epsilon = 1$, and show figures showcasing the error probability results from Section~\ref{subsec:error_prob} with varying values of~$\epsilon$. Additionally we include results on the size of the attractive set~$\mathcal{Z}_{attract}(\epsilon)$.

All simulation examples use~$b = 1$ in Mechanism~\ref{mech:bo}. This choice of~$b$ implies that only changing one state-to-state transition is protected by differential privacy. This is a strong privacy guarantee, and it is standard in the differential privacy literature. However, choosing~$\rho = 1$ while using Mechanism~\ref{mech:PF} implies that changes in multiple state-to-state transitions are protected by differential privacy, which provides a stronger notion of privacy. For example, 
under Definition~\ref{def:traj_adjacency}
the maximum number of state-to-state transitions that can change between adjacent trajectories is~$\left \lfloor\frac{\rho}{\min_{i, j\in\mathcal{S}} G(i, j)}\right \rfloor$. In the credit migration example, we find that the maximum number of state-to-state transitions in which two trajectories can differ in this way is~$11$, implying that the privacy guarantee from Mechanism~\ref{mech:PF} at~$\rho = 1$ is stronger than that of Mechanism~\ref{mech:bo} at~$b = 1$.

Figure~\ref{fig:stability_ratio} shows the fraction of the state space with negative drift for Mechanism~\ref{mech:PF} both empirically and according to Theorem~\ref{thm:old_adj_safe_set} for~$\epsilon \in \{0.5, 1,\ldots, 10\}$ for all three examples.
In all three examples, we find empirically that nearly the entire joint state space has negative drift. Theorem~\ref{thm:old_adj_safe_set} gives a sufficient condition, and thus it provides a lower bound on the fraction of joint states with negative drift. However, in smaller Markov chains, such as the credit migration example, Theorem~\ref{thm:old_adj_safe_set} predicts nearly~$100\%$ of the state space exhibits negative drift under Mechanism~\ref{mech:PF}, which matches the empirical results. Figures~\ref{fig:probability_gainesville_vary_ic} and \ref{fig:entropy_gainesville_vary_ic} show the same scenario as Figures~\ref{fig:probability_gainesville} and~\ref{fig:entropy_gainesville} in the case where the initial state distribution is the uniform distribution over the state space~$\mathcal{S}$. 
Those figures show 
 similar results to the
 setting in Section~\ref{sec:typical} in which 
 the initial distribution is the Dirac distribution, which shows that Mechanism~\ref{mech:PF} is robust to the use of arbitrary initial states.

    \begin{figure*}[!t]
\centering
    \begin{subfigure}{0.4\linewidth}
        \centering
%
%
\definecolor{chocolate2267451}{RGB}{226,74,51}
\definecolor{dimgray85}{RGB}{85,85,85}
\definecolor{gainsboro229}{RGB}{229,229,229}
\definecolor{lightgray204}{RGB}{204,204,204}
\definecolor{steelblue52138189}{RGB}{52,138,189}
\definecolor{black}{RGB}{0, 0, 0}
\definecolor{GTblue}{RGB}{0, 48, 87}
\definecolor{GTgold}{RGB}{179, 163, 105}
\definecolor{UFOrange}{RGB}{250, 70, 22}
\definecolor{UFblue}{RGB}{0, 33, 165}
\begin{tikzpicture}

\begin{axis}[%
width=0.4\figW,
height=\figH,
axis background/.style={fill=gainsboro229},
axis line style={white},
scale only axis,
xlabel=\textcolor{black}{Time, $t$},
xtick style={color=dimgray85},
x grid style={white},
yminorticks=true,
y grid style={white},
ylabel=\textcolor{black}{Shortest-Path distance, $V(Z_t)$},
xmajorgrids,
ymajorgrids,
yminorgrids,
tick align=outside,
tick pos=left,
legend pos = north east,
legend style={nodes={scale=0.65, transform shape}},
legend columns=1,
 xmin=5,
 xmax=100,
]
        






        \addplot[UFOrange, ultra thick, mark={*}, mark size={2.0 pt}] table [col sep=comma, x=TimeStep, y=V_t] {Figures/moodys_chen_trajectory_data.csv};
        \addplot[UFblue, ultra thick, mark={*}, mark size={2.0 pt}] table [col sep=comma, x=TimeStep, y=V_t] {Figures/moodys_trajectory_data.csv};
\coordinate (ttl) at (axis description cs:0.5,-0.4);
\end{axis}
\node[anchor=south,text width=7cm,align=center] at (ttl) 
{(a)};
\end{tikzpicture}%

        
    \end{subfigure}
    \hfill 
    \begin{subfigure}{0.45\textwidth}
        \centering
%
%
\definecolor{chocolate2267451}{RGB}{226,74,51}
\definecolor{dimgray85}{RGB}{85,85,85}
\definecolor{gainsboro229}{RGB}{229,229,229}
\definecolor{lightgray204}{RGB}{204,204,204}
\definecolor{steelblue52138189}{RGB}{52,138,189}
\definecolor{black}{RGB}{0, 0, 0}
\definecolor{GTblue}{RGB}{0, 48, 87}
\definecolor{GTgold}{RGB}{179, 163, 105}
\definecolor{UFOrange}{RGB}{250, 70, 22}
\definecolor{UFblue}{RGB}{0, 33, 165}
\begin{tikzpicture}

\begin{axis}[%
width=0.4\figW,
height=\figH,
axis background/.style={fill=gainsboro229},
axis line style={white},
scale only axis,
xlabel=\textcolor{black}{Time, $t$},
xtick style={color=dimgray85},
x grid style={white},
yminorticks=true,
y grid style={white},
xmajorgrids,
ymajorgrids,
yminorgrids,
tick align=outside,
tick pos=left,
legend pos = north east,
legend style={nodes={scale=0.65, transform shape}},
legend columns=1,
 xmin=5,
 xmax=100,
]
        






        \addplot[UFOrange, ultra thick, mark={*}, mark size={2.0 pt}] table [col sep=comma, x=TimeStep, y=V_t] {Figures/gainesville_chen_trajectory_data.csv};
        \addlegendentry{Chen et al. ($\epsilon = 1.0$)}
        \addplot[UFblue, ultra thick, mark={*}, mark size={2.0 pt}] table [col sep=comma, x=TimeStep, y=V_t] {Figures/gainesville_trajectory_data.csv};
        \addlegendentry{Mechanism~\ref{mech:PF} ($\epsilon = 1.0$)}

\coordinate (ttl) at (axis description cs:0.5,-0.4);
\end{axis}
\node[anchor=south,text width=7cm,align=center] at (ttl) 
{(b)};
\end{tikzpicture}%

        
    \end{subfigure}
    
    \begin{subfigure}{0.45\textwidth}
        \centering
%
%
\definecolor{chocolate2267451}{RGB}{226,74,51}
\definecolor{dimgray85}{RGB}{85,85,85}
\definecolor{gainsboro229}{RGB}{229,229,229}
\definecolor{lightgray204}{RGB}{204,204,204}
\definecolor{steelblue52138189}{RGB}{52,138,189}
\definecolor{black}{RGB}{0, 0, 0}
\definecolor{GTblue}{RGB}{0, 48, 87}
\definecolor{GTgold}{RGB}{179, 163, 105}
\definecolor{UFOrange}{RGB}{250, 70, 22}
\definecolor{UFblue}{RGB}{0, 33, 165}
\begin{tikzpicture}

\begin{axis}[%
width=0.4\figW,
height=\figH,
axis background/.style={fill=gainsboro229},
axis line style={white},
scale only axis,
xlabel=\textcolor{black}{Time, $t$},
xtick style={color=dimgray85},
x grid style={white},
yminorticks=true,
y grid style={white},
ylabel=\textcolor{black}{Shortest-Path distance, $V(Z_t)$},
xmajorgrids,
ymajorgrids,
yminorgrids,
tick align=outside,
tick pos=left,
legend pos = north east,
legend style={nodes={scale=0.65, transform shape}},
legend columns=1,
 xmin=5,
 xmax=100,
]
        






        \addplot[UFOrange, ultra thick, mark={*}, mark size={2.0 pt}] table [col sep=comma, x=TimeStep, y=V_t] {Figures/wikispeedia_chen_trajectory_data.csv};
        \addplot[UFblue, ultra thick, mark={*}, mark size={2.0 pt}] table [col sep=comma, x=TimeStep, y=V_t] {Figures/wikispeedia_trajectory_data.csv};
\coordinate (ttl) at (axis description cs:0.5,-0.4);
\end{axis}
\node[anchor=south,text width=7cm,align=center] at (ttl) 
{(c)};
\end{tikzpicture}%

        
    \end{subfigure}
    \caption{Single sample shortest-path distances from the private state~$s_t'$ to the underlying sensitive state~$s_t$ for each~$t\in\{0, 1, \ldots, 100\}$    
    for the (a) credit migration, (b) Gainesville traffic, and (c) Wikispeedia datasets.  }
    \label{fig:trajectory_gainesville}
\end{figure*}
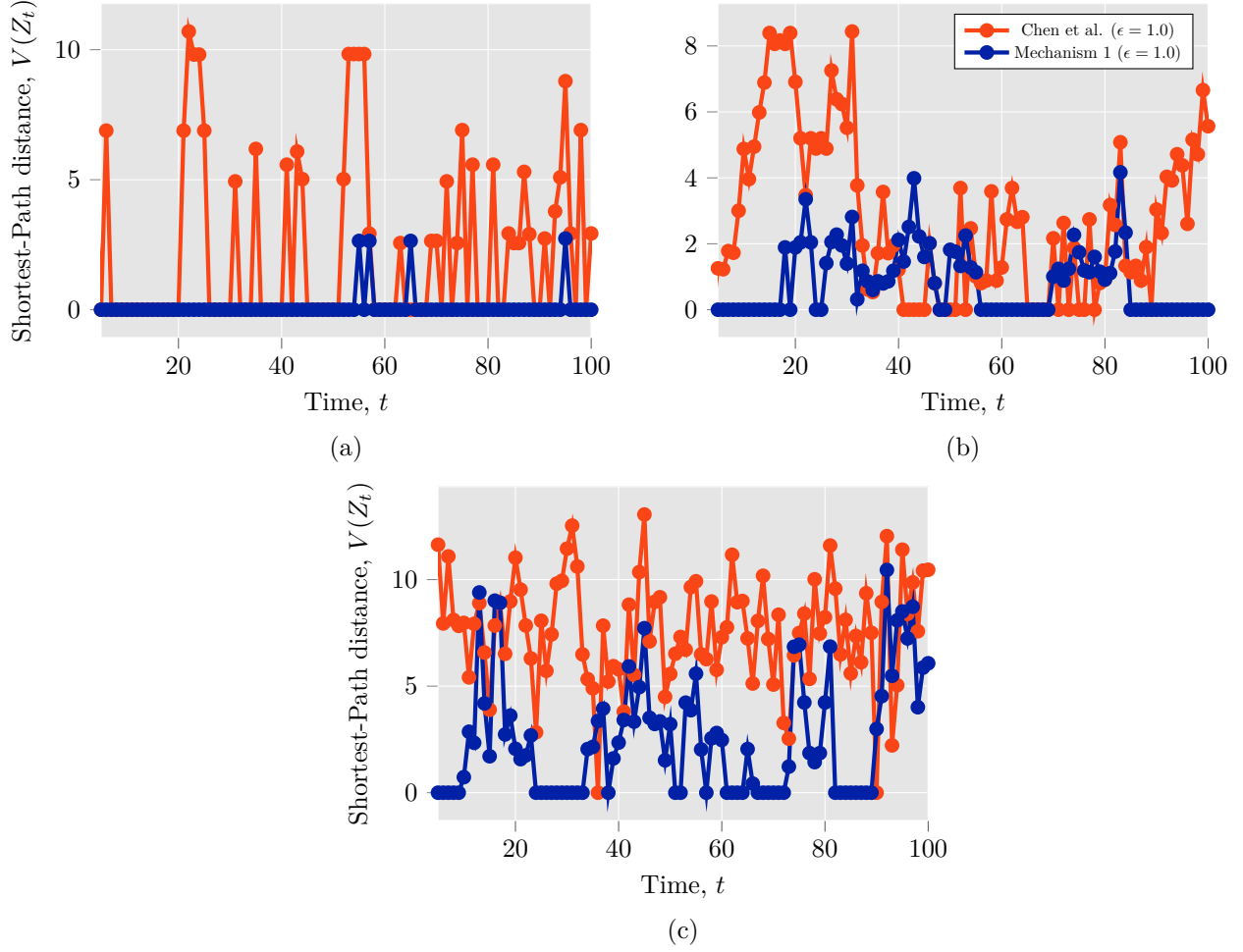

\begin{figure*}
\centering
    \begin{subfigure}{0.49\linewidth}
        \centering
        \includegraphics[width=\linewidth]{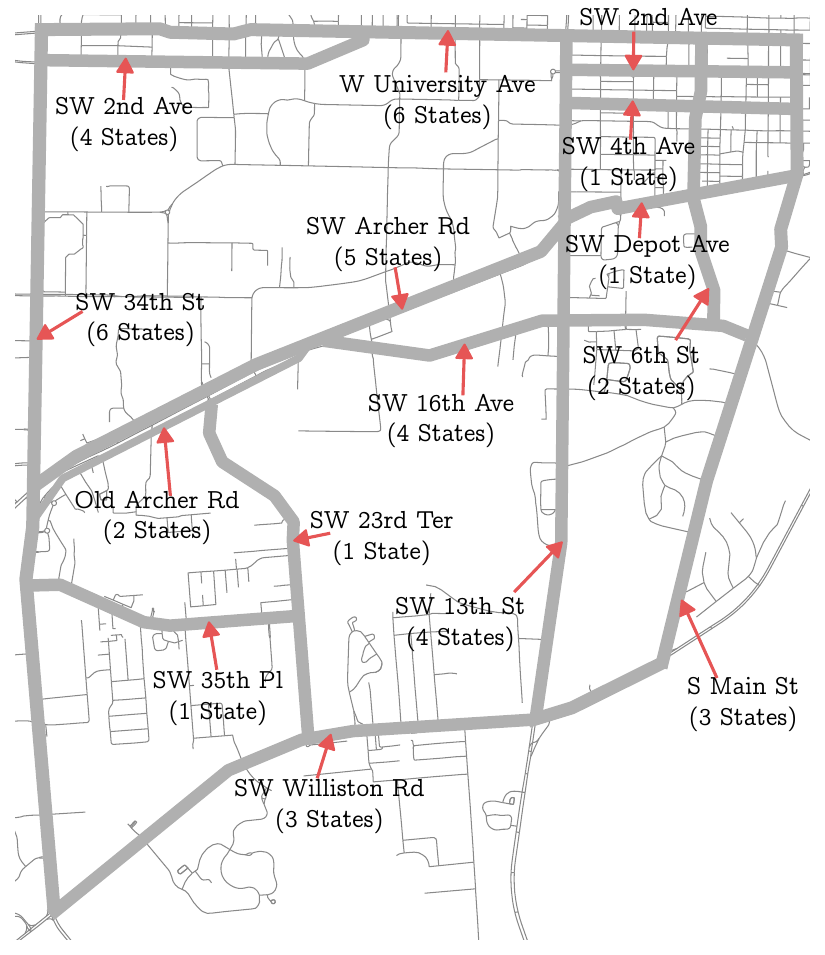}
        \caption{}
    \end{subfigure}
    \hfill 
    \begin{subfigure}{0.49\textwidth}
        \centering
        \includegraphics[width=\linewidth]{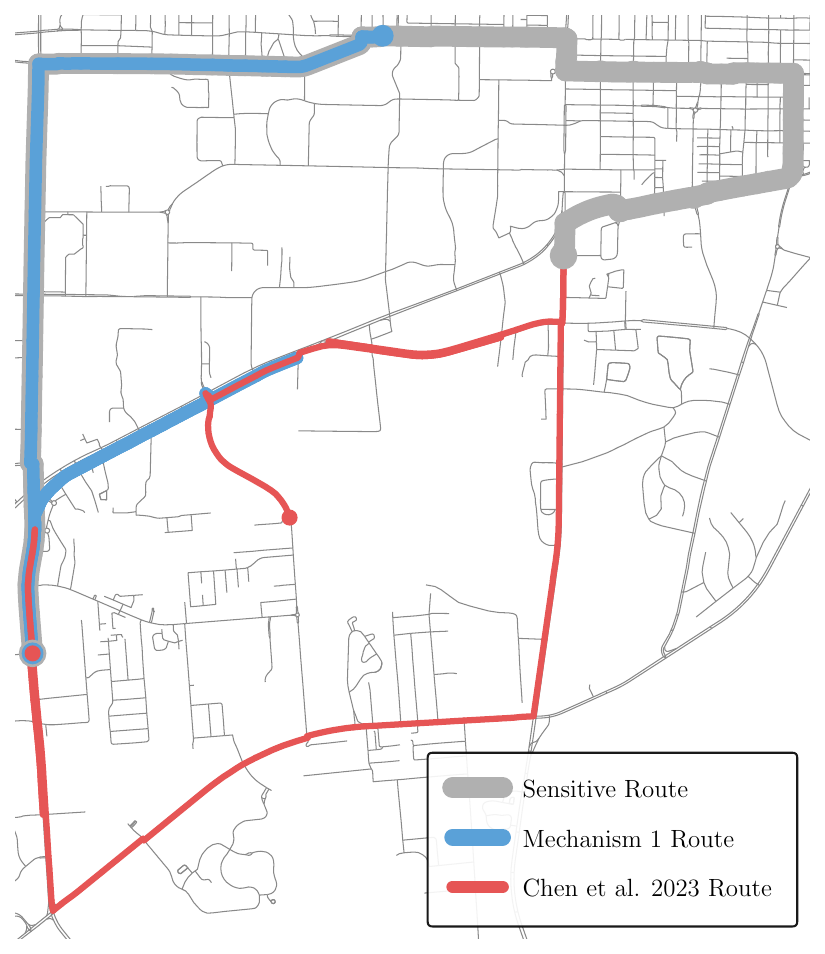}
        \caption{$\epsilon = 1$}
    \end{subfigure}
    \hfill
    \begin{subfigure}{0.49\textwidth}
        \centering
        \includegraphics[width=\linewidth]{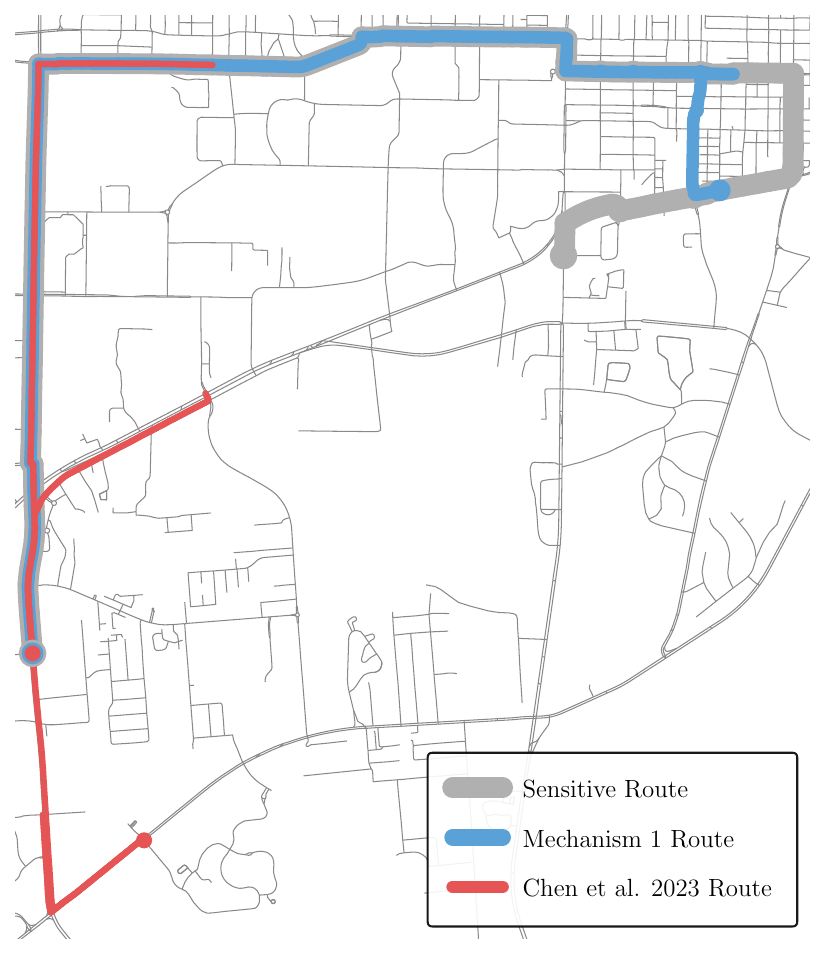}
        \caption{$\epsilon = 2$}
    \end{subfigure}
    \hfill
    \begin{subfigure}{0.49\textwidth}
        \centering
        \includegraphics[width=\linewidth]{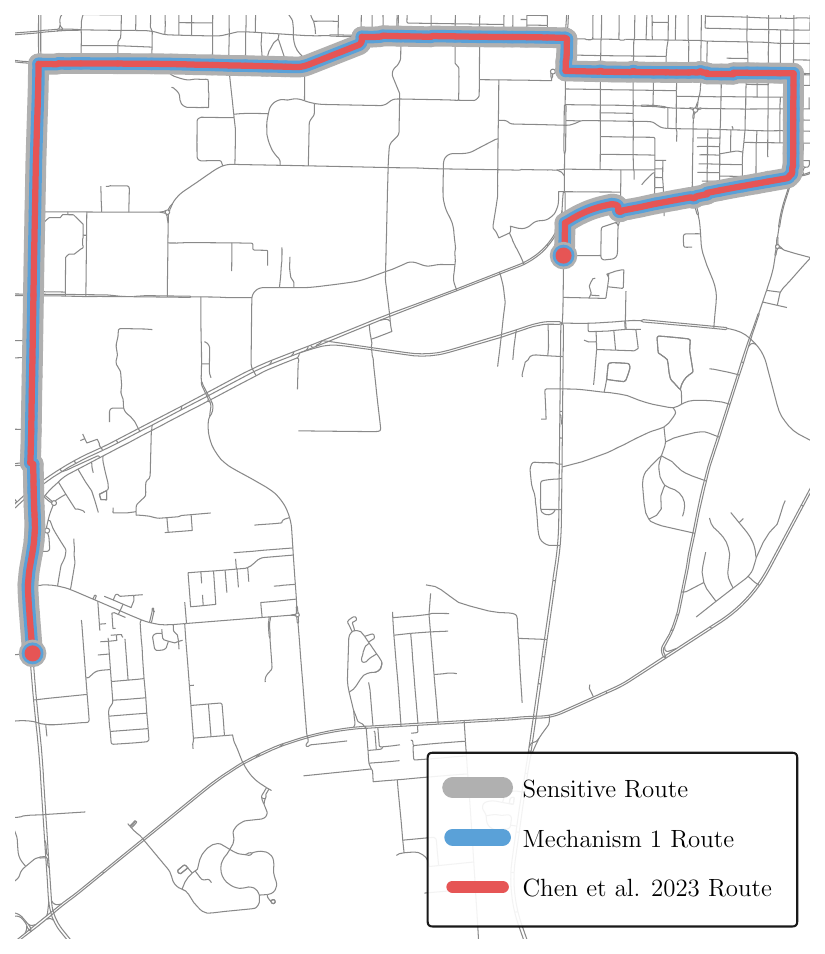}
        \caption{$\epsilon = 5$}
    \end{subfigure}
    \caption{A map of Gainesville streets and state names in (a), along with sample private state trajectories for varying levels of privacy in (b), (c), and (d). Even when Mechanism~\ref{mech:PF}'s outputs
    deviate from the underlying sensitive state trajectory,
    the private state trajectories it generates bear more similarity to the sensitive state trajectory than those generated by~\cite{chen2023differentialsymbolic}.}
    \label{fig:gainesville_trajectories}
\end{figure*}

\begin{figure}
     \centering
        \includegraphics[width=\linewidth]{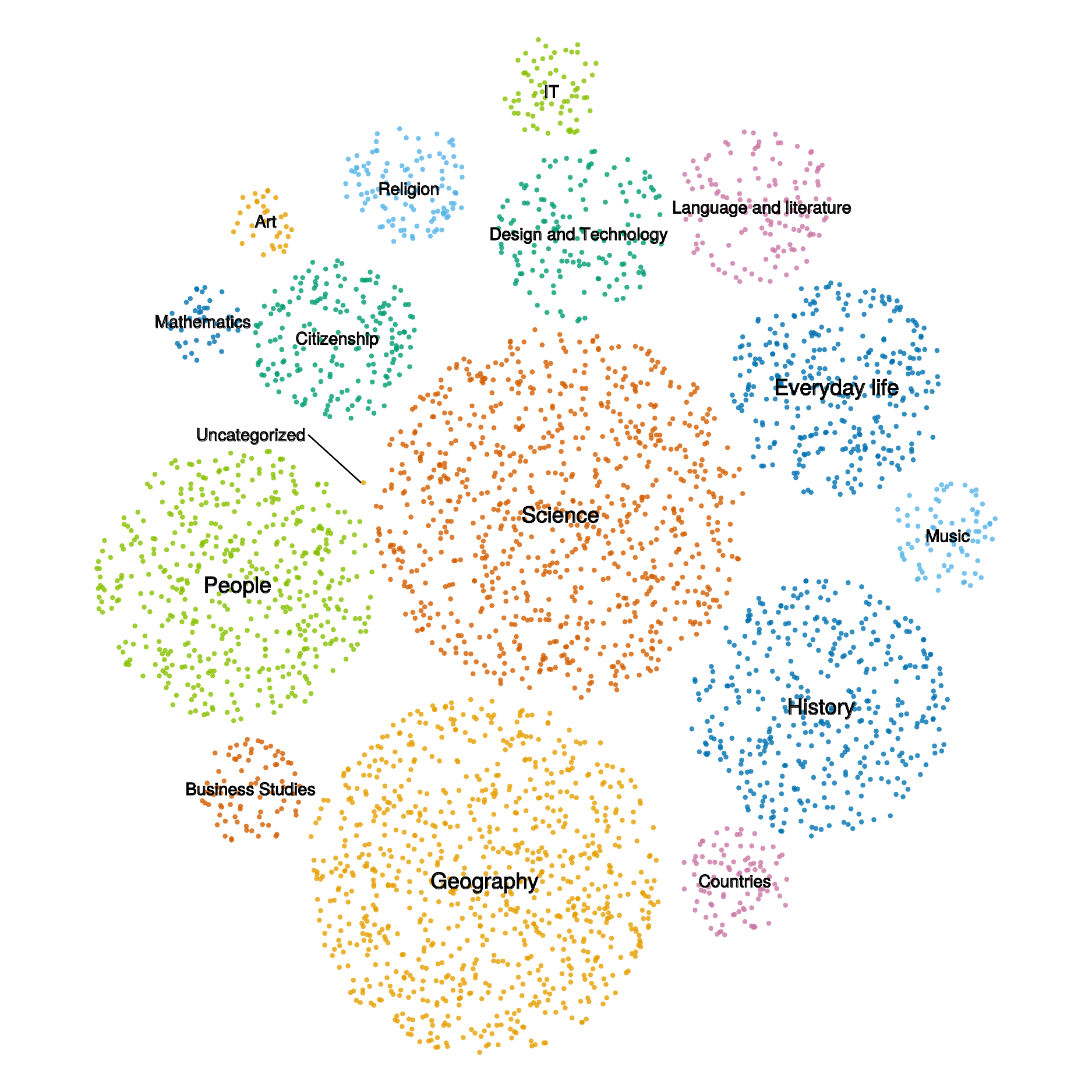}
        \caption{Map of Wikispeedia dataset with category names. The ``Uncategorized'' label only applies to a single article.        
        }
        \label{fig:wiki_map}
\end{figure}

\begin{figure*}
\centering
    \begin{subfigure}{\textwidth}
        \centering
        \includegraphics[width=\linewidth]{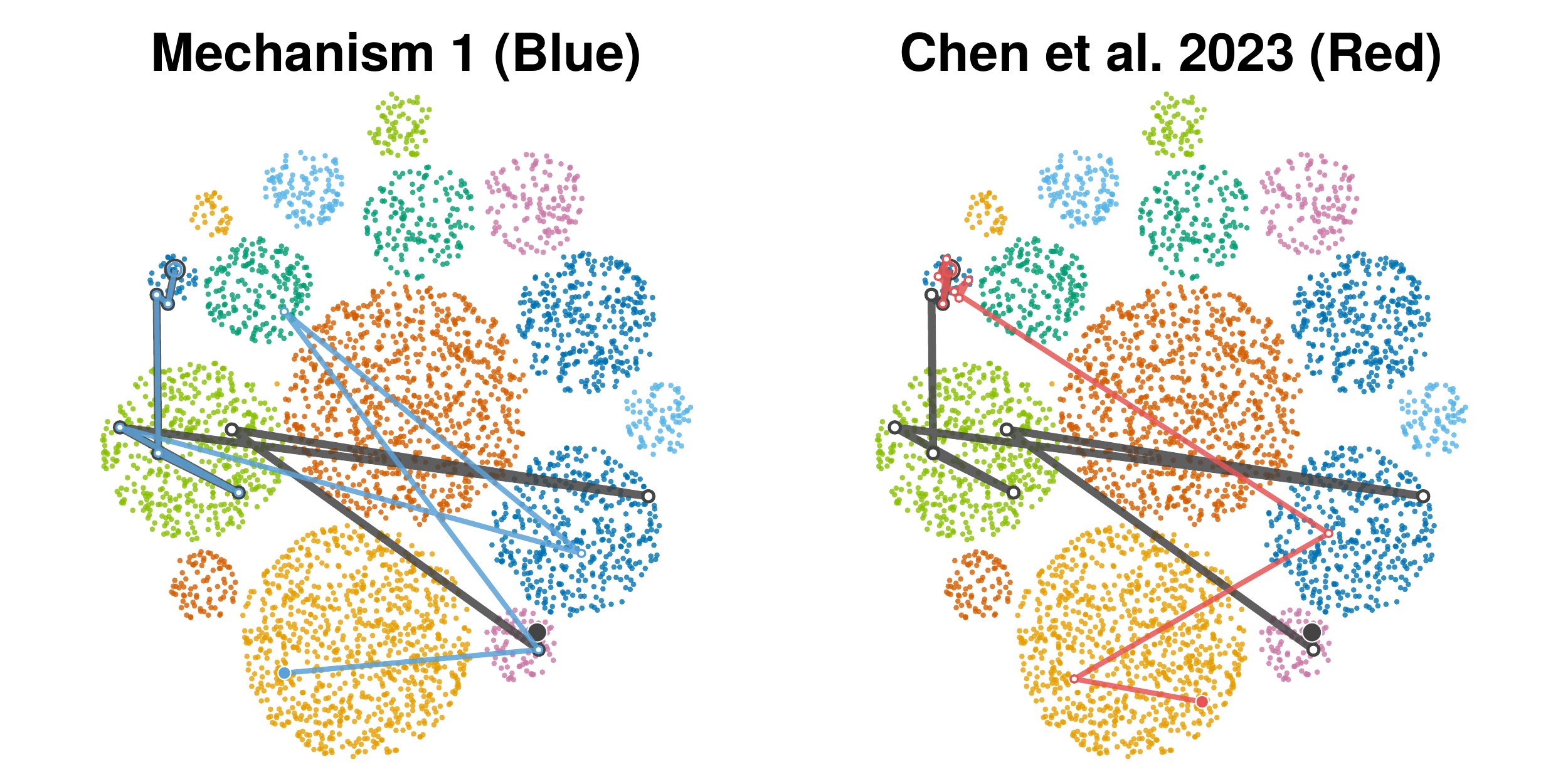}
        \caption{$\epsilon = 2$}
    \end{subfigure}
    \hfill
    \begin{subfigure}{\textwidth}
        \centering
        \includegraphics[width=\linewidth]{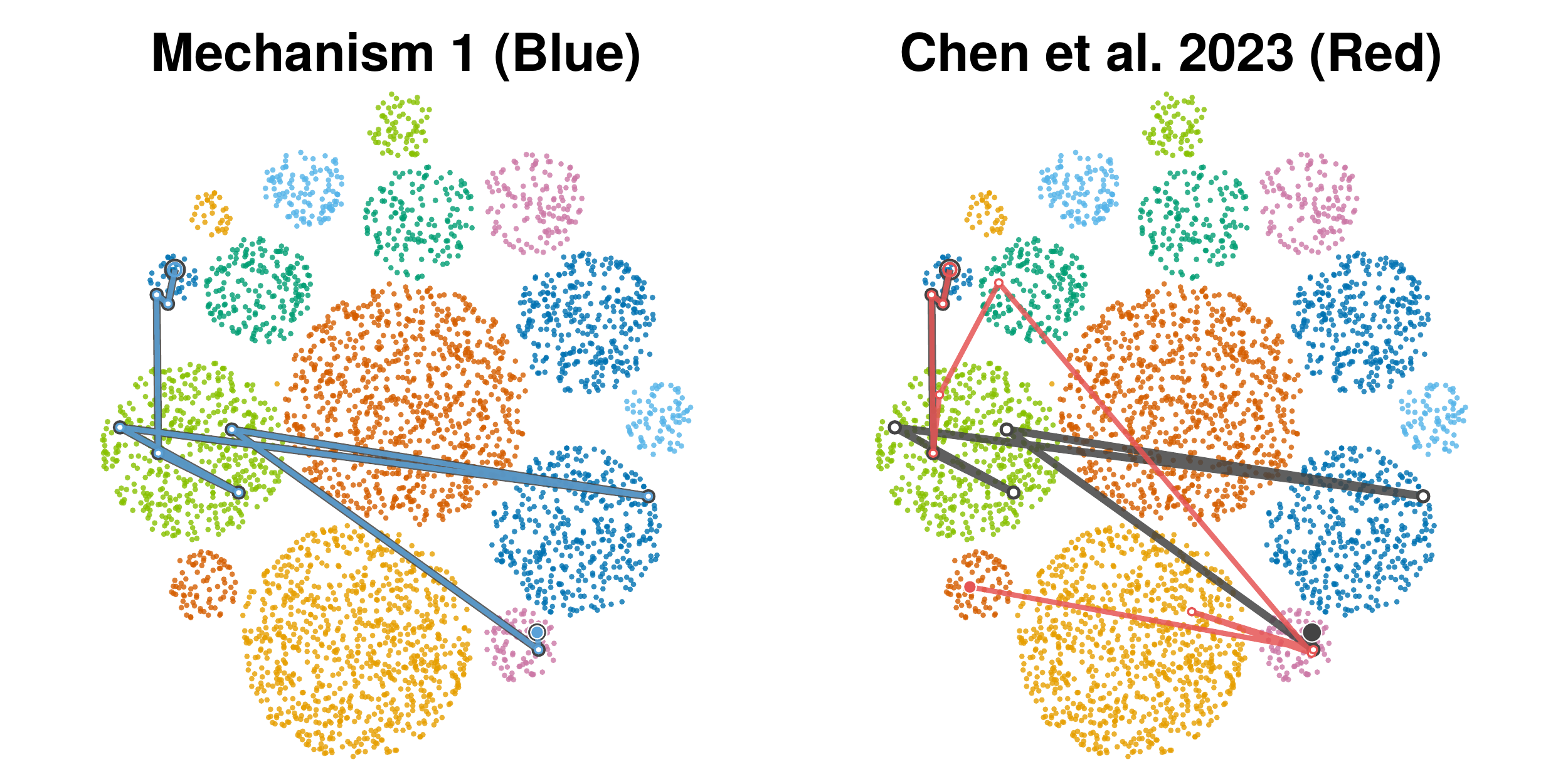}
        \caption{$\epsilon = 5$}
    \end{subfigure}
    \caption{Sample state trajectories for varying levels of privacy, with the sensitive trajectory in gray. Even when Mechanism~\ref{mech:PF} is incorrect, the private state trajectories often return to their corresponding true trajectories, while the state trajectories generated by Mechanism~\ref{mech:bo}     
    often do not.}
    \label{fig:wiki_trajectories}
\end{figure*}

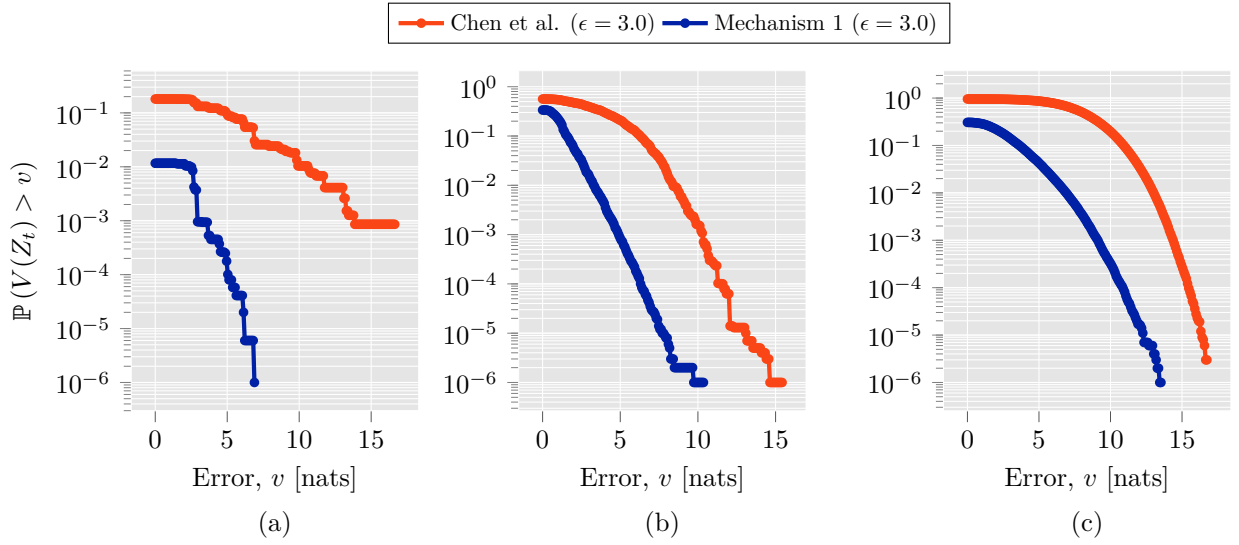
\begin{figure*}[t]
\centering
    \begin{subfigure}{0.3\linewidth}
        \centering
%
%
\definecolor{chocolate2267451}{RGB}{226,74,51}
\definecolor{dimgray85}{RGB}{85,85,85}
\definecolor{gainsboro229}{RGB}{229,229,229}
\definecolor{lightgray204}{RGB}{204,204,204}
\definecolor{steelblue52138189}{RGB}{52,138,189}
\definecolor{black}{RGB}{0, 0, 0}
\definecolor{GTblue}{RGB}{0, 48, 87}
\definecolor{GTgold}{RGB}{179, 163, 105}
\definecolor{UFOrange}{RGB}{250, 70, 22}
\definecolor{UFblue}{RGB}{0, 33, 165}
\begin{tikzpicture}

\begin{axis}[%
width=0.23\figW,
height=\figH,
axis background/.style={fill=gainsboro229},
axis line style={white},
scale only axis,
xlabel=\textcolor{black}{Error, $v$ [nats]},
xtick style={color=dimgray85},
x grid style={white},
yminorticks=true,
y grid style={white},
ylabel=\textcolor{black}{$\prob{V(Z_t)>v}$},
xmajorgrids,
ymajorgrids,
yminorgrids,
tick align=outside,
tick pos=left,
yminorticks=true,
max space between ticks=20,
legend style={nodes={scale=0.85, transform shape}, at={(2.85,1.2)}},
legend columns=2,
ymode = log,
]
        






        \addplot[UFOrange, ultra thick, mark={*}, mark size={1.0 pt}, unbounded coords=discard] table [col sep=comma, x=v, y=CCDF_Chen] {Figures/moodys_ccdf_data_e3.csv};
        \addlegendentry{Chen et al. ($\epsilon = 3.0$)}
        \addplot[UFblue, ultra thick, mark={*}, mark size={1.0 pt}] table [col sep=comma, x=v, y=CCDF_Mechanism3, unbounded coords=discard] {Figures/moodys_ccdf_data_e3.csv};
        \addlegendentry{Mechanism~\ref{mech:PF} ($\epsilon = 3.0$)}

\coordinate (ttl) at (axis description cs:0.5,-0.4);
\end{axis}
\node[anchor=south,text width=7cm,align=center] at (ttl) 
{(a)};
\end{tikzpicture}%

        
    \end{subfigure}
    \hfill 
    \begin{subfigure}{0.33\textwidth}
        \centering
%
%
\definecolor{chocolate2267451}{RGB}{226,74,51}
\definecolor{dimgray85}{RGB}{85,85,85}
\definecolor{gainsboro229}{RGB}{229,229,229}
\definecolor{lightgray204}{RGB}{204,204,204}
\definecolor{steelblue52138189}{RGB}{52,138,189}
\definecolor{black}{RGB}{0, 0, 0}
\definecolor{GTblue}{RGB}{0, 48, 87}
\definecolor{GTgold}{RGB}{179, 163, 105}
\definecolor{UFOrange}{RGB}{250, 70, 22}
\definecolor{UFblue}{RGB}{0, 33, 165}
\begin{tikzpicture}

\begin{axis}[%
width=0.23\figW,
height=\figH,
axis background/.style={fill=gainsboro229},
axis line style={white},
scale only axis,
xlabel=\textcolor{black}{Error, $v$ [nats]},
xtick style={color=dimgray85},
x grid style={white},
yminorticks=true,
y grid style={white},
xmajorgrids,
ymajorgrids,
yminorgrids,
tick align=outside,
tick pos=left,
yminorticks=true,
max space between ticks=20,
legend pos = north east,
legend style={nodes={scale=0.65, transform shape}},
legend columns=1,
ymode = log,
]
        






        \addplot[UFOrange, ultra thick, mark={*}, mark size={1.0 pt}, unbounded coords=discard] table [col sep=comma, x=v, y=CCDF_Chen] {Figures/gainesville_ccdf_data_e3.csv};
        \addplot[UFblue, ultra thick, mark={*}, mark size={1.0 pt}] table [col sep=comma, x=v, y=CCDF_Mechanism3, unbounded coords=discard] {Figures/gainesville_ccdf_data_e3.csv};
\coordinate (ttl) at (axis description cs:0.5,-0.4);
\end{axis}
\node[anchor=south,text width=7cm,align=center] at (ttl) 
{(b)};
\end{tikzpicture}%

        
    \end{subfigure}
    \hfill
    \begin{subfigure}{0.35\textwidth}
        \centering
%
%
\definecolor{chocolate2267451}{RGB}{226,74,51}
\definecolor{dimgray85}{RGB}{85,85,85}
\definecolor{gainsboro229}{RGB}{229,229,229}
\definecolor{lightgray204}{RGB}{204,204,204}
\definecolor{steelblue52138189}{RGB}{52,138,189}
\definecolor{black}{RGB}{0, 0, 0}
\definecolor{GTblue}{RGB}{0, 48, 87}
\definecolor{GTgold}{RGB}{179, 163, 105}
\definecolor{UFOrange}{RGB}{250, 70, 22}
\definecolor{UFblue}{RGB}{0, 33, 165}
\begin{tikzpicture}

\begin{axis}[%
width=0.23\figW,
height=\figH,
axis background/.style={fill=gainsboro229},
axis line style={white},
scale only axis,
xlabel=\textcolor{black}{Error, $v$ [nats]},
xtick style={color=dimgray85},
x grid style={white},
yminorticks=true,
y grid style={white},
xmajorgrids,
ymajorgrids,
yminorgrids,
tick align=outside,
tick pos=left,
yminorticks=true,
max space between ticks=20,
legend pos = north east,
legend style={nodes={scale=0.65, transform shape}},
legend columns=1,
ymode = log,
]
        






        \addplot[UFOrange, ultra thick, mark={*}, mark size={1.0 pt}, unbounded coords=discard] table [col sep=comma, x=v, y=CCDF_Chen] {Figures/wikispeedia_ccdf_data_e3.csv};
        \addplot[UFblue, ultra thick, mark={*}, mark size={1.0 pt}] table [col sep=comma, x=v, y=CCDF_Mechanism3, unbounded coords=discard] {Figures/wikispeedia_ccdf_data_e3.csv};
\coordinate (ttl) at (axis description cs:0.5,-0.4);
\end{axis}
\node[anchor=south,text width=7cm,align=center] at (ttl) 
{(c)};
\end{tikzpicture}%

        
    \end{subfigure}
    \caption{Probability of large errors in private trajectories in the (a) credit migration, (b) Gainesville traffic, and (c) Wikispeedia datasets. With a larger
    value of~$\epsilon$ compared to Figure~\ref{fig:probability_gainesville}, in (b)~$B_{\epsilon}$ is small enough such that bound in Theorem~\ref{cor:hajek_bound} holds, 
    and thus we see a linear relationship on the log-scale axes.}
    \label{fig:probability_gainesville_e3}
\end{figure*}

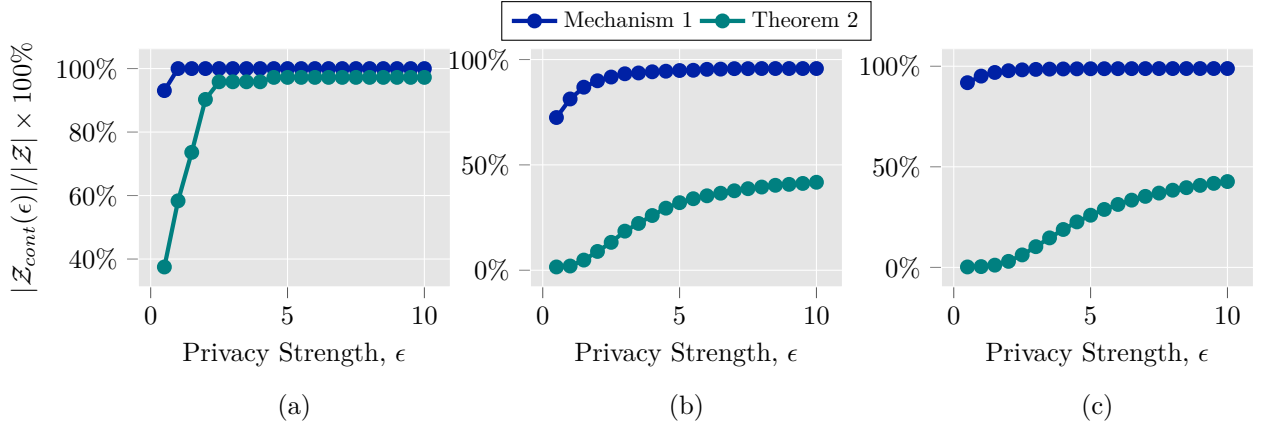
\begin{figure}[!t]
    \begin{subfigure}{0.3\linewidth}
        \centering
%
%
\definecolor{chocolate2267451}{RGB}{226,74,51}
\definecolor{dimgray85}{RGB}{85,85,85}
\definecolor{gainsboro229}{RGB}{229,229,229}
\definecolor{lightgray204}{RGB}{204,204,204}
\definecolor{steelblue52138189}{RGB}{52,138,189}
\definecolor{black}{RGB}{0, 0, 0}
\definecolor{GTblue}{RGB}{0, 48, 87}
\definecolor{GTgold}{RGB}{179, 163, 105}
\definecolor{UFOrange}{RGB}{250, 70, 22}
\definecolor{UFblue}{RGB}{0, 33, 165}
\begin{tikzpicture}

\begin{axis}[%
width=0.25\figW,
height=0.7\figH,
axis background/.style={fill=gainsboro229},
axis line style={white},
scale only axis,
xlabel=\textcolor{black}{Privacy Strength, $\epsilon$},
xtick style={color=dimgray85},
x grid style={white},
yminorticks=true,
y grid style={white},
ylabel=\textcolor{black}{$|\mathcal{Z}_{cont}(\epsilon)|/|\mathcal{Z}|\times 100\%$},
xmajorgrids,
ymajorgrids,
yminorgrids,
tick align=outside,
tick pos=left,
yticklabel={$\pgfmathprintnumber{\tick}$\%},
legend style={nodes={scale=0.85, transform shape}, at={(2.35,1.2)}},
legend columns=3,
]
\addplot[UFblue, ultra thick, mark={*}, mark size={2.0 pt}] table [col sep=comma, x=epsilon, y=Moody's, unbounded coords=discard] {Figures/fig3_stability_data2.csv};
\addplot[teal, ultra thick, mark={*}, mark size={2.0 pt}] table [col sep=comma, x=epsilon, y=lemma_2_Moody's, unbounded coords=discard] {Figures/fig3_stability_data2.csv};
\addlegendentry{Mechanism~\ref{mech:PF}}
\addlegendentry{Theorem~\ref{thm:old_adj_safe_set}}
\coordinate (ttl) at (axis description cs:0.5,-0.6);
\end{axis}
\node[anchor=south,text width=7cm,align=center] at (ttl) 
{(a)};
\end{tikzpicture}%
    \end{subfigure}
    \hspace{0.25cm} 
    \begin{subfigure}{0.3\textwidth}
        \centering
%
%
\definecolor{chocolate2267451}{RGB}{226,74,51}
\definecolor{dimgray85}{RGB}{85,85,85}
\definecolor{gainsboro229}{RGB}{229,229,229}
\definecolor{lightgray204}{RGB}{204,204,204}
\definecolor{steelblue52138189}{RGB}{52,138,189}
\definecolor{black}{RGB}{0, 0, 0}
\definecolor{GTblue}{RGB}{0, 48, 87}
\definecolor{GTgold}{RGB}{179, 163, 105}
\definecolor{UFOrange}{RGB}{250, 70, 22}
\definecolor{UFblue}{RGB}{0, 33, 165}
\begin{tikzpicture}

\begin{axis}[%
width=0.25\figW,
height=0.7\figH,
axis background/.style={fill=gainsboro229},
axis line style={white},
scale only axis,
xlabel=\textcolor{black}{Privacy Strength, $\epsilon$},
xtick style={color=dimgray85},
x grid style={white},
yminorticks=true,
y grid style={white},
xmajorgrids,
ymajorgrids,
yminorgrids,
tick align=outside,
tick pos=left,
yticklabel={$\pgfmathprintnumber{\tick}$\%},
]
\addplot[UFblue, ultra thick, mark={*}, mark size={2.0 pt}] table [col sep=comma, x=epsilon, y=Gainesville, unbounded coords=discard] {Figures/fig3_stability_data2.csv};
\addplot[teal, ultra thick, mark={*}, mark size={2.0 pt}] table [col sep=comma, x=epsilon, y=lemma_2_Gainesville, unbounded coords=discard] {Figures/fig3_stability_data2.csv};

\coordinate (ttl) at (axis description cs:0.5,-0.6);
\end{axis}
\node[anchor=south,text width=7cm,align=center] at (ttl) 
{(b)};
\end{tikzpicture}%
    \end{subfigure}
    \hspace{0.25cm}
    \begin{subfigure}{0.3\textwidth}
        \centering
%
%
\definecolor{chocolate2267451}{RGB}{226,74,51}
\definecolor{dimgray85}{RGB}{85,85,85}
\definecolor{gainsboro229}{RGB}{229,229,229}
\definecolor{lightgray204}{RGB}{204,204,204}
\definecolor{steelblue52138189}{RGB}{52,138,189}
\definecolor{black}{RGB}{0, 0, 0}
\definecolor{GTblue}{RGB}{0, 48, 87}
\definecolor{GTgold}{RGB}{179, 163, 105}
\definecolor{UFOrange}{RGB}{250, 70, 22}
\definecolor{UFblue}{RGB}{0, 33, 165}
\begin{tikzpicture}

\begin{axis}[%
width=0.25\figW,
height=0.7\figH,
axis background/.style={fill=gainsboro229},
axis line style={white},
scale only axis,
xlabel=\textcolor{black}{Privacy Strength, $\epsilon$},
xtick style={color=dimgray85},
x grid style={white},
yminorticks=true,
y grid style={white},
xmajorgrids,
ymajorgrids,
yminorgrids,
tick align=outside,
tick pos=left,
yticklabel={$\pgfmathprintnumber{\tick}$\%},
]
\addplot[UFblue, ultra thick, mark={*}, mark size={2.0 pt}] table [col sep=comma, x=epsilon, y=Wikispeedia, unbounded coords=discard] {Figures/fig3_stability_data2.csv};
\addplot[teal, ultra thick, mark={*}, mark size={2.0 pt}] table [col sep=comma, x=epsilon, y=lemma_2_Wikispeedia, unbounded coords=discard] {Figures/fig3_stability_data2.csv};
\coordinate (ttl) at (axis description cs:0.5,-0.6);
\end{axis}
\node[anchor=south,text width=7cm,align=center] at (ttl) 
{(c)};
\end{tikzpicture}%
    \end{subfigure}
    \caption{Fraction of the joint state space with negative drift in the (a) credit migration, (b) Gainesville traffic, and (c) Wikispeedia examples. Because Theorem~\ref{thm:old_adj_safe_set} gives a sufficient condition, in practice we find that more states have negative drift than is suggested by Theorem~\ref{thm:old_adj_safe_set}.}
    \label{fig:stability_ratio}
\end{figure}

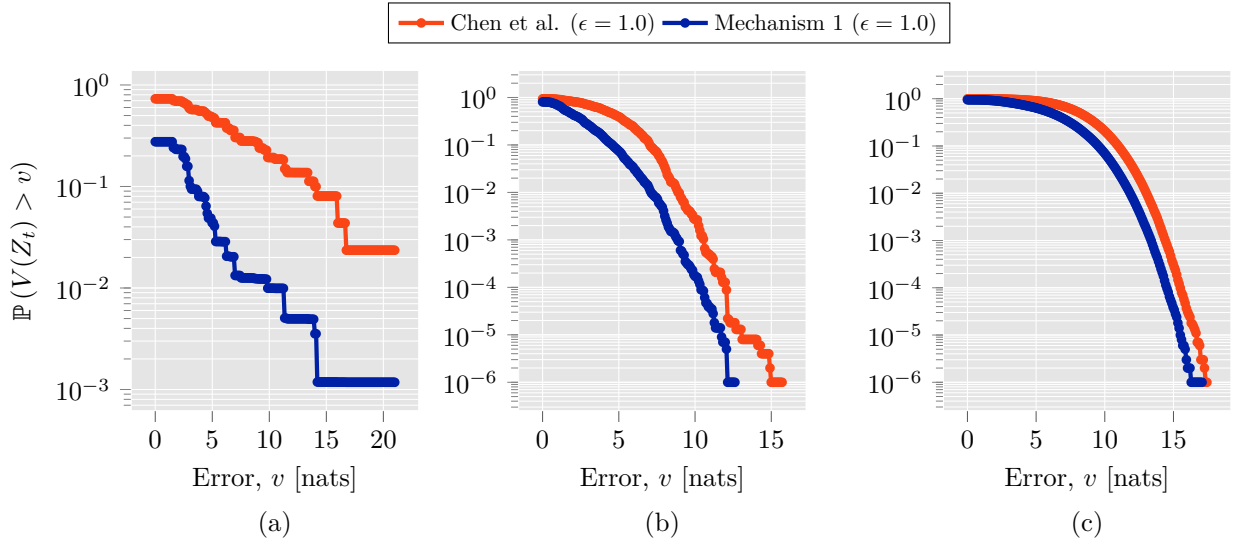
\begin{figure*}
\centering
    \begin{subfigure}{0.3\linewidth}
        \centering
%
%
\definecolor{chocolate2267451}{RGB}{226,74,51}
\definecolor{dimgray85}{RGB}{85,85,85}
\definecolor{gainsboro229}{RGB}{229,229,229}
\definecolor{lightgray204}{RGB}{204,204,204}
\definecolor{steelblue52138189}{RGB}{52,138,189}
\definecolor{black}{RGB}{0, 0, 0}
\definecolor{GTblue}{RGB}{0, 48, 87}
\definecolor{GTgold}{RGB}{179, 163, 105}
\definecolor{UFOrange}{RGB}{250, 70, 22}
\definecolor{UFblue}{RGB}{0, 33, 165}
\begin{tikzpicture}

\begin{axis}[%
width=0.23\figW,
height=\figH,
axis background/.style={fill=gainsboro229},
axis line style={white},
scale only axis,
xlabel=\textcolor{black}{Error, $v$ [nats]},
xtick style={color=dimgray85},
x grid style={white},
yminorticks=true,
y grid style={white},
ylabel=\textcolor{black}{$\prob{V(Z_t)>v}$},
xmajorgrids,
ymajorgrids,
yminorgrids,
tick align=outside,
tick pos=left,
yminorticks=true,
max space between ticks=20,
legend style={nodes={scale=0.85, transform shape}, at={(2.85,1.2)}},
legend columns=2,
ymode = log,
]
        






        \addplot[UFOrange, ultra thick, mark={*}, mark size={1.0 pt}, unbounded coords=discard] table [col sep=comma, x=v, y=CCDF_Chen] {Figures/moodys_ccdf_data_vary_ic.csv};
        \addlegendentry{Chen et al. ($\epsilon = 1.0$)}
        \addplot[UFblue, ultra thick, mark={*}, mark size={1.0 pt}] table [col sep=comma, x=v, y=CCDF_Mechanism3, unbounded coords=discard] {Figures/moodys_ccdf_data_vary_ic.csv};
        \addlegendentry{Mechanism~\ref{mech:PF} ($\epsilon = 1.0$)}

\coordinate (ttl) at (axis description cs:0.5,-0.4);
\end{axis}
\node[anchor=south,text width=7cm,align=center] at (ttl) 
{(a)};
\end{tikzpicture}%

        
    \end{subfigure}
    \hfill 
    \begin{subfigure}{0.33\textwidth}
        \centering
%
%
\definecolor{chocolate2267451}{RGB}{226,74,51}
\definecolor{dimgray85}{RGB}{85,85,85}
\definecolor{gainsboro229}{RGB}{229,229,229}
\definecolor{lightgray204}{RGB}{204,204,204}
\definecolor{steelblue52138189}{RGB}{52,138,189}
\definecolor{black}{RGB}{0, 0, 0}
\definecolor{GTblue}{RGB}{0, 48, 87}
\definecolor{GTgold}{RGB}{179, 163, 105}
\definecolor{UFOrange}{RGB}{250, 70, 22}
\definecolor{UFblue}{RGB}{0, 33, 165}
\begin{tikzpicture}

\begin{axis}[%
width=0.23\figW,
height=\figH,
axis background/.style={fill=gainsboro229},
axis line style={white},
scale only axis,
xlabel=\textcolor{black}{Error, $v$ [nats]},
xtick style={color=dimgray85},
x grid style={white},
yminorticks=true,
y grid style={white},
xmajorgrids,
ymajorgrids,
yminorgrids,
tick align=outside,
tick pos=left,
yminorticks=true,
max space between ticks=20,
legend style={nodes={scale=0.85, transform shape}, at={(2.65,1.2)}},
legend columns=2,
ymode = log,
]
        






        \addplot[UFOrange, ultra thick, mark={*}, mark size={1.0 pt}, unbounded coords=discard] table [col sep=comma, x=v, y=CCDF_Chen] {Figures/gainesville_ccdf_data_vary_ic.csv};
        \addplot[UFblue, ultra thick, mark={*}, mark size={1.0 pt}] table [col sep=comma, x=v, y=CCDF_Mechanism4, unbounded coords=discard] {Figures/gainesville_ccdf_data_vary_ic.csv};
\coordinate (ttl) at (axis description cs:0.5,-0.4);
\end{axis}
\node[anchor=south,text width=7cm,align=center] at (ttl) 
{(b)};
\end{tikzpicture}%

        
    \end{subfigure}
    \hfill
    \begin{subfigure}{0.35\textwidth}
        \centering
%
%
\definecolor{chocolate2267451}{RGB}{226,74,51}
\definecolor{dimgray85}{RGB}{85,85,85}
\definecolor{gainsboro229}{RGB}{229,229,229}
\definecolor{lightgray204}{RGB}{204,204,204}
\definecolor{steelblue52138189}{RGB}{52,138,189}
\definecolor{black}{RGB}{0, 0, 0}
\definecolor{GTblue}{RGB}{0, 48, 87}
\definecolor{GTgold}{RGB}{179, 163, 105}
\definecolor{UFOrange}{RGB}{250, 70, 22}
\definecolor{UFblue}{RGB}{0, 33, 165}
\begin{tikzpicture}

\begin{axis}[%
width=0.23\figW,
height=\figH,
axis background/.style={fill=gainsboro229},
axis line style={white},
scale only axis,
xlabel=\textcolor{black}{Error, $v$ [nats]},
xtick style={color=dimgray85},
x grid style={white},
yminorticks=true,
y grid style={white},
xmajorgrids,
ymajorgrids,
yminorgrids,
tick align=outside,
tick pos=left,
yminorticks=true,
max space between ticks=20,
legend style={nodes={scale=0.85, transform shape}, at={(2.65,1.2)}},
legend columns=2,
ymode = log,
]
        






        \addplot[UFOrange, ultra thick, mark={*}, mark size={1.0 pt}, unbounded coords=discard] table [col sep=comma, x=v, y=CCDF_Chen] {Figures/wikispeedia_ccdf_data_vary_ic.csv};
        \addplot[UFblue, ultra thick, mark={*}, mark size={1.0 pt}] table [col sep=comma, x=v, y=CCDF_Mechanism4, unbounded coords=discard] {Figures/wikispeedia_ccdf_data_vary_ic.csv};
\coordinate (ttl) at (axis description cs:0.5,-0.4);
\end{axis}
\node[anchor=south,text width=7cm,align=center] at (ttl) 
{(c)};
\end{tikzpicture}%

        
    \end{subfigure}
    \caption{Probability of large errors in private trajectories in the (a) credit migration, (b) Gainesville traffic, and (c) Wikispeedia datasets with~$p_0$ as the uniform distribution over all~$s\in\mathcal{S}$.}
    \label{fig:probability_gainesville_vary_ic}
\end{figure*}

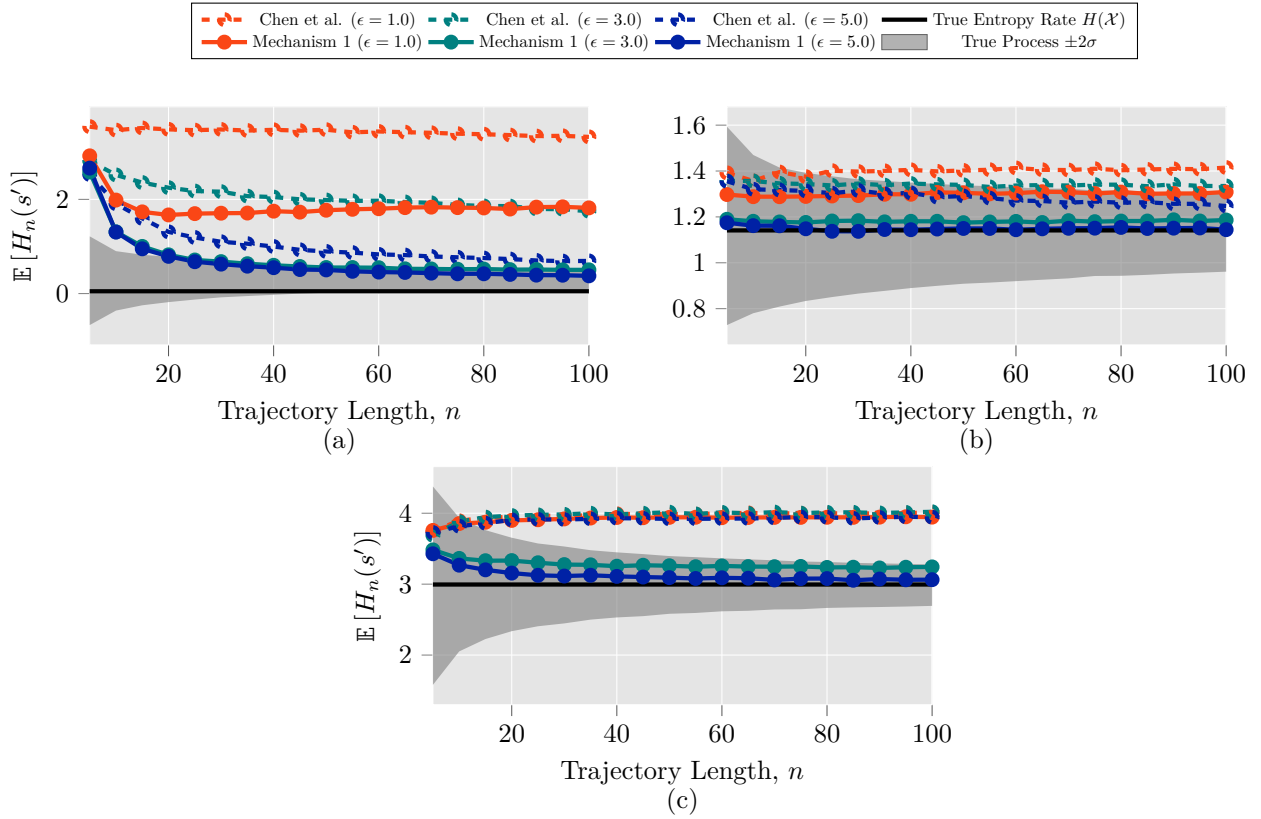
\begin{figure*}
\centering
    \begin{subfigure}{0.4\linewidth}
        \centering
%
%
\definecolor{chocolate2267451}{RGB}{226,74,51}
\definecolor{dimgray85}{RGB}{85,85,85}
\definecolor{gainsboro229}{RGB}{229,229,229}
\definecolor{lightgray204}{RGB}{204,204,204}
\definecolor{steelblue52138189}{RGB}{52,138,189}
\definecolor{black}{RGB}{0, 0, 0}
\definecolor{GTblue}{RGB}{0, 48, 87}
\definecolor{GTgold}{RGB}{179, 163, 105}
\definecolor{UFOrange}{RGB}{250, 70, 22}
\definecolor{UFblue}{RGB}{0, 33, 165}
\begin{tikzpicture}

\begin{axis}[%
width=0.4\figW,
height=0.7\figH,
axis background/.style={fill=gainsboro229},
axis line style={white},
scale only axis,
xlabel=\textcolor{black}{Trajectory Length, $n$},
xtick style={color=dimgray85},
x grid style={white},
yminorticks=true,
y grid style={white},
ylabel=\textcolor{black}{$\E{H_n(s')}$},
xmajorgrids,
ymajorgrids,
yminorgrids,
tick align=outside,
tick pos=left,
legend pos = south east,
legend style={nodes={scale=0.65, transform shape}, at={(2.1,1.2)}},
legend columns=4,
 xmin=5,
 xmax=100,
]
\addplot[name path=upper, draw=none, , forget plot] table [col sep=comma, x=T, y=TrueUpper] {Figures/moodys_fidelity_data_vary_ic.csv};
        \addplot[name path=lower, draw=none, forget plot] table [col sep=comma, x=T, y=TrueLower] {Figures/moodys_fidelity_data_vary_ic.csv};
        






        \addplot[UFOrange, dashed, ultra thick, mark={o}, mark size={2.0 pt}] table [col sep=comma, x=T, y=Chen_1_0] {Figures/moodys_fidelity_data_vary_ic.csv};
        \addlegendentry{Chen et al. ($\epsilon = 1.0$)}

        \addplot[teal, dashed, ultra thick, mark={o}, mark size={2.0 pt}] table [col sep=comma, x=T, y=Chen_3_0] {Figures/moodys_fidelity_data_vary_ic.csv};
        \addlegendentry{Chen et al. ($\epsilon = 3.0$)}

        \addplot[UFblue, dashed, ultra thick, mark={o}, mark size={2.0 pt}] table [col sep=comma, x=T, y=Chen_5_0] {Figures/moodys_fidelity_data_vary_ic.csv};
        \addlegendentry{Chen et al. ($\epsilon = 5.0$)}
        \addplot[black, ultra thick] table [col sep=comma, x=T, y=TrueMean] {Figures/moodys_fidelity_data_vary_ic.csv};
        \addlegendentry{True Entropy Rate $H(\mathcal{X})$}
        \addplot[UFOrange, ultra thick, mark={*}, mark size={2.0 pt}] table [col sep=comma, x=T, y=PF_1_0] {Figures/moodys_fidelity_data_vary_ic.csv};
        \addlegendentry{Mechanism~\ref{mech:PF} ($\epsilon = 1.0$)}


        \addplot[teal, ultra thick, mark={*}, mark size={2.0 pt}] table [col sep=comma, x=T, y=PF_3_0] {Figures/moodys_fidelity_data_vary_ic.csv};
        \addlegendentry{Mechanism~\ref{mech:PF} ($\epsilon = 3.0$)}

        \addplot[UFblue, ultra thick, mark={*}, mark size={2.0 pt}] table [col sep=comma, x=T, y=PF_5_0] {Figures/moodys_fidelity_data_vary_ic.csv};
        \addlegendentry{Mechanism~\ref{mech:PF} ($\epsilon = 5.0$)}

        \addplot[fill=black!50, opacity=0.6] fill between[of=upper and lower];
        \addlegendentry{True Process $\pm 2\sigma$}
\coordinate (ttl) at (axis description cs:0.5,-0.5);
\end{axis}
\node[anchor=south,text width=7cm,align=center] at (ttl) 
{(a)};
\end{tikzpicture}%

        
    \end{subfigure}
    \hfill 
    \begin{subfigure}{0.47\textwidth}
        \centering
%
%
\definecolor{chocolate2267451}{RGB}{226,74,51}
\definecolor{dimgray85}{RGB}{85,85,85}
\definecolor{gainsboro229}{RGB}{229,229,229}
\definecolor{lightgray204}{RGB}{204,204,204}
\definecolor{steelblue52138189}{RGB}{52,138,189}
\definecolor{black}{RGB}{0, 0, 0}
\definecolor{GTblue}{RGB}{0, 48, 87}
\definecolor{GTgold}{RGB}{179, 163, 105}
\definecolor{UFOrange}{RGB}{250, 70, 22}
\definecolor{UFblue}{RGB}{0, 33, 165}
\begin{tikzpicture}

\begin{axis}[%
width=0.4\figW,
height=0.7\figH,
axis background/.style={fill=gainsboro229},
axis line style={white},
scale only axis,
xlabel=\textcolor{black}{Trajectory Length, $n$},
xtick style={color=dimgray85},
x grid style={white},
yminorticks=true,
y grid style={white},
xmajorgrids,
ymajorgrids,
yminorgrids,
tick align=outside,
tick pos=left,
legend pos = south east,
legend style={nodes={scale=0.65, transform shape}},
legend columns=2,
 xmin=5,
 xmax=100,
]
\addplot[name path=upper, draw=none, , forget plot] table [col sep=comma, x=T, y=TrueUpper] {Figures/Gainesville_fidelity_data_vary_ic.csv};
        \addplot[name path=lower, draw=none, forget plot] table [col sep=comma, x=T, y=TrueLower] {Figures/Gainesville_fidelity_data_vary_ic.csv};
        






        \addplot[UFOrange, dashed, ultra thick, mark={o}, mark size={2.0 pt}] table [col sep=comma, x=T, y=Chen_1_0] {Figures/Gainesville_fidelity_data_vary_ic.csv};

        \addplot[UFOrange, ultra thick, mark={*}, mark size={2.0 pt}] table [col sep=comma, x=T, y=PF_1_0] {Figures/Gainesville_fidelity_data_vary_ic.csv};


        \addplot[teal, dashed, ultra thick, mark={o}, mark size={2.0 pt}] table [col sep=comma, x=T, y=Chen_3_0] {Figures/Gainesville_fidelity_data_vary_ic.csv};

        \addplot[teal, ultra thick, mark={*}, mark size={2.0 pt}] table [col sep=comma, x=T, y=PF_3_0] {Figures/Gainesville_fidelity_data_vary_ic.csv};

        \addplot[UFblue, dashed, ultra thick, mark={o}, mark size={2.0 pt}] table [col sep=comma, x=T, y=Chen_5_0] {Figures/Gainesville_fidelity_data_vary_ic.csv};

        \addplot[UFblue, ultra thick, mark={*}, mark size={2.0 pt}] table [col sep=comma, x=T, y=PF_5_0] {Figures/Gainesville_fidelity_data_vary_ic.csv};
        \addplot[black, ultra thick] table [col sep=comma, x=T, y=TrueMean] {Figures/Gainesville_fidelity_data_vary_ic.csv};

        \addplot[fill=black!50, opacity=0.6] fill between[of=upper and lower];
\coordinate (ttl) at (axis description cs:0.5,-0.5);
\end{axis}
\node[anchor=south,text width=7cm,align=center] at (ttl) 
{(b)};
\end{tikzpicture}%

        
    \end{subfigure}
    
    \begin{subfigure}{0.45\textwidth}
        \centering
%
%
\definecolor{chocolate2267451}{RGB}{226,74,51}
\definecolor{dimgray85}{RGB}{85,85,85}
\definecolor{gainsboro229}{RGB}{229,229,229}
\definecolor{lightgray204}{RGB}{204,204,204}
\definecolor{steelblue52138189}{RGB}{52,138,189}
\definecolor{black}{RGB}{0, 0, 0}
\definecolor{GTblue}{RGB}{0, 48, 87}
\definecolor{GTgold}{RGB}{179, 163, 105}
\definecolor{UFOrange}{RGB}{250, 70, 22}
\definecolor{UFblue}{RGB}{0, 33, 165}
\begin{tikzpicture}

\begin{axis}[%
width=0.4\figW,
height=0.7\figH,
axis background/.style={fill=gainsboro229},
axis line style={white},
scale only axis,
xlabel=\textcolor{black}{Trajectory Length, $n$},
xtick style={color=dimgray85},
x grid style={white},
yminorticks=true,
y grid style={white},
ylabel=\textcolor{black}{$\E{H_n(s')}$},
xmajorgrids,
ymajorgrids,
yminorgrids,
tick align=outside,
tick pos=left,
legend pos = south east,
legend style={nodes={scale=0.65, transform shape}},
legend columns=2,
 xmin=5,
 xmax=100,
]
\addplot[name path=upper, draw=none, , forget plot] table [col sep=comma, x=T, y=TrueUpper] {Figures/wikispeedia_fidelity_data_vary_ic.csv};
        \addplot[name path=lower, draw=none, forget plot] table [col sep=comma, x=T, y=TrueLower] {Figures/wikispeedia_fidelity_data_vary_ic.csv};
        






        \addplot[UFOrange, dashed, ultra thick, mark={o}, mark size={2.0 pt}] table [col sep=comma, x=T, y=Chen_1_0] {Figures/wikispeedia_fidelity_data_vary_ic.csv};

        \addplot[UFOrange, ultra thick, mark={*}, mark size={2.0 pt}] table [col sep=comma, x=T, y=PF_1_0] {Figures/wikispeedia_fidelity_data_vary_ic.csv};


        \addplot[teal, dashed, ultra thick, mark={o}, mark size={2.0 pt}] table [col sep=comma, x=T, y=Chen_3_0] {Figures/wikispeedia_fidelity_data_vary_ic.csv};

        \addplot[teal, ultra thick, mark={*}, mark size={2.0 pt}] table [col sep=comma, x=T, y=PF_3_0] {Figures/wikispeedia_fidelity_data_vary_ic.csv};

        \addplot[UFblue, dashed, ultra thick, mark={o}, mark size={2.0 pt}] table [col sep=comma, x=T, y=Chen_5_0] {Figures/wikispeedia_fidelity_data_vary_ic.csv};

        \addplot[UFblue, ultra thick, mark={*}, mark size={2.0 pt}] table [col sep=comma, x=T, y=PF_5_0] {Figures/wikispeedia_fidelity_data_vary_ic.csv};
        \addplot[black, ultra thick] table [col sep=comma, x=T, y=TrueMean] {Figures/wikispeedia_fidelity_data_vary_ic.csv};

        \addplot[fill=black!50, opacity=0.6] fill between[of=upper and lower];
\coordinate (ttl) at (axis description cs:0.5,-0.5);
\end{axis}
\node[anchor=south,text width=7cm,align=center] at (ttl) 
{(c)};
\end{tikzpicture}%

        
    \end{subfigure}
    \caption{Empirical entropy with varying~$\epsilon$ and trajectory length in the (a) credit migration, (b) Gainesville traffic, and (c) Wikispeedia examples with~$p_0$ as the uniform distribution over all~$s\in\mathcal{S}$.}
    \label{fig:entropy_gainesville_vary_ic}
\end{figure*}
\end{document}